\PassOptionsToPackage{unicode}{hyperref}
\PassOptionsToPackage{hyphens}{url}
\PassOptionsToPackage{dvipsnames,svgnames,x11names}{xcolor}
\documentclass[
  12pt,
  letterpaper]{article}

\usepackage{amsmath,amssymb}
\usepackage{setspace}
\usepackage{iftex}
\ifPDFTeX
  \usepackage[T1]{fontenc}
  \usepackage[utf8]{inputenc}
  \usepackage{textcomp} 
\else 
  \usepackage{unicode-math}
  \defaultfontfeatures{Scale=MatchLowercase}
  \defaultfontfeatures[\rmfamily]{Ligatures=TeX,Scale=1}
\fi
\usepackage{lmodern}
\ifPDFTeX\else  
\fi
\IfFileExists{upquote.sty}{\usepackage{upquote}}{}
\IfFileExists{microtype.sty}{
  \usepackage[]{microtype}
  \UseMicrotypeSet[protrusion]{basicmath} 
}{}
\makeatletter
\@ifundefined{KOMAClassName}{
  \IfFileExists{parskip.sty}{%
    \usepackage{parskip}
  }{
    \setlength{\parindent}{0pt}
    \setlength{\parskip}{6pt plus 2pt minus 1pt}}
}{
  \KOMAoptions{parskip=half}}
\makeatother
\usepackage{xcolor}
\usepackage[margin=1in]{geometry}
\makeatletter
\ifx\paragraph\undefined\else
  \let\oldparagraph\paragraph
  \renewcommand{\paragraph}{
    \@ifstar
      \xxxParagraphStar
      \xxxParagraphNoStar
  }
  \newcommand{\xxxParagraphStar}[1]{\oldparagraph*{#1}\mbox{}}
  \newcommand{\xxxParagraphNoStar}[1]{\oldparagraph{#1}\mbox{}}
\fi
\ifx\subparagraph\undefined\else
  \let\oldsubparagraph\subparagraph
  \renewcommand{\subparagraph}{
    \@ifstar
      \xxxSubParagraphStar
      \xxxSubParagraphNoStar
  }
  \newcommand{\xxxSubParagraphStar}[1]{\oldsubparagraph*{#1}\mbox{}}
  \newcommand{\xxxSubParagraphNoStar}[1]{\oldsubparagraph{#1}\mbox{}}
\fi
\makeatother

\providecommand{\tightlist}{%
  \setlength{\itemsep}{0pt}\setlength{\parskip}{0pt}}\usepackage{longtable,booktabs,array}
\usepackage{calc} 
\usepackage{etoolbox}
\makeatletter
\patchcmd\longtable{\par}{\if@noskipsec\mbox{}\fi\par}{}{}
\makeatother
\IfFileExists{footnotehyper.sty}{\usepackage{footnotehyper}}{\usepackage{footnote}}
\makesavenoteenv{longtable}
\usepackage{graphicx}
\makeatletter
\newsavebox\pandoc@box
\newcommand*\pandocbounded[1]{
  \sbox\pandoc@box{#1}%
  \Gscale@div\@tempa{\textheight}{\dimexpr\ht\pandoc@box+\dp\pandoc@box\relax}%
  \Gscale@div\@tempb{\linewidth}{\wd\pandoc@box}%
  \ifdim\@tempb\p@<\@tempa\p@\let\@tempa\@tempb\fi
  \ifdim\@tempa\p@<\p@\scalebox{\@tempa}{\usebox\pandoc@box}%
  \else\usebox{\pandoc@box}%
  \fi%
}
\def\fps@figure{htbp}
\makeatother
\NewDocumentCommand\citeproctext{}{}
\NewDocumentCommand\citeproc{mm}{%
  \begingroup\def\citeproctext{#2}\cite{#1}\endgroup}
\makeatletter
 \let\@cite@ofmt\@firstofone
 \def\@biblabel#1{}
 \def\@cite#1#2{{#1\if@tempswa , #2\fi}}
\makeatother
\newlength{\cslhangindent}
\newlength{\csllabelwidth}
\newenvironment{CSLReferences}[2] 
 {\begin{list}{}{%
  \setlength{\itemindent}{0pt}
  \setlength{\leftmargin}{0pt}
  \setlength{\parsep}{0pt}
  \ifodd #1
   \setlength{\leftmargin}{\cslhangindent}
   \setlength{\itemindent}{-1\cslhangindent}
  \fi
  \setlength{\itemsep}{#2\baselineskip}}}
 {\end{list}}
\usepackage{calc}

\usepackage{amsthm}
\ifPDFTeX
  \DeclareUnicodeCharacter{2212}{\ensuremath{-}}   
  \DeclareUnicodeCharacter{00D7}{\ensuremath{\times}}
  \DeclareUnicodeCharacter{2192}{\ensuremath{\rightarrow}}
  \DeclareUnicodeCharacter{0394}{\ensuremath{\Delta}}
  \DeclareUnicodeCharacter{2264}{\ensuremath{\leq}}
  \DeclareUnicodeCharacter{2265}{\ensuremath{\geq}}
\fi
\makeatletter
\@ifpackageloaded{caption}{}{\usepackage{caption}}
\AtBeginDocument{%
\ifdefined\contentsname
  \renewcommand*\contentsname{Table of contents}
\else
  \newcommand\contentsname{Table of contents}
\fi
\ifdefined\listfigurename
  \renewcommand*\listfigurename{List of Figures}
\else
  \newcommand\listfigurename{List of Figures}
\fi
\ifdefined\listtablename
  \renewcommand*\listtablename{List of Tables}
\else
  \newcommand\listtablename{List of Tables}
\fi
\ifdefined\figurename
  \renewcommand*\figurename{Figure}
\else
  \newcommand\figurename{Figure}
\fi
\ifdefined\tablename
  \renewcommand*\tablename{Table}
\else
  \newcommand\tablename{Table}
\fi
}
\@ifpackageloaded{float}{}{\usepackage{float}}
\floatstyle{ruled}
\@ifundefined{c@chapter}{\newfloat{codelisting}{h}{lop}}{\newfloat{codelisting}{h}{lop}[chapter]}
\floatname{codelisting}{Listing}

\usepackage{amsthm}
\theoremstyle{plain}
\newtheorem{corollary}{Corollary}
\theoremstyle{plain}
\newtheorem{proposition}{Proposition}
\theoremstyle{plain}
\newtheorem{theorem}{Theorem}
\theoremstyle{plain}
\newtheorem{lemma}{Lemma}
\theoremstyle{remark}
\AtBeginDocument{}
\newtheorem*{remark}{Remark}

\newtheorem{refremark}{Remark}

\makeatother
\makeatletter
\@ifpackageloaded{caption}{}{\usepackage{caption}}
\@ifpackageloaded{subcaption}{}{\usepackage{subcaption}}
\makeatother

\usepackage{bookmark}

\IfFileExists{xurl.sty}{\usepackage{xurl}}{} 
\hypersetup{
  pdftitle={Identifying Panel Conditioning with Refreshment Samples: Sharp Bounds and Design Assumptions},
  pdfauthor={Shoki Okubo},
  pdfkeywords={attrition, panel conditioning, partial
identification, refreshment samples, rotation panels, survey design},
  colorlinks=true,
  linkcolor={blue},
  filecolor={Maroon},
  citecolor={Blue},
  urlcolor={Blue},
  pdfcreator={LaTeX via pandoc}}

\title{Identifying Panel Conditioning with Refreshment Samples: Sharp
Bounds and Design Assumptions\thanks{This research benefited from discussions and feedback during presentations at the Institute of Social Science, University of Tokyo, the Japanese Association for Mathematical Sociology, and the panel survey conference at Keio University, and from comments by Hiroshi Ishida, Kazuo Yamaguchi, and Hiroki Takikawa. Code for all simulations and verification scripts in this paper is in the replication archive at \url{https://github.com/sokubo/paper-refreshment-designs-replication} (fixed version: tag \texttt{paper-v1.5}, commit \texttt{f4f7c2b}). The empirical illustration uses licensed JLPS microdata, which are not redistributed; the archive contains the code that produces the reported aggregates and a synthetic example (see Data and code availability). This work was supported by JSPS KAKENHI Grant Number 22K13525.}}
\author{Shoki Okubo\thanks{Department of Sociology, Toyo University, Tokyo, Japan. Email: okubo080@toyo.jp. Website: sokubo.github.io.}}
\date{October 1, 2026}

\begin{document}
\maketitle
\begin{abstract}
Refreshment samples are the standard remedy for panel attrition, and the
identification results behind them maintain that participation does not
change measurement. We characterize what a refreshment sample identifies
about panel conditioning, modelled as a deterministic monotone map at
reinterview, when attrition is unrestricted. A candidate map is
consistent with the data if and only if the retention-scaled
distribution of the stayers' implied latent outcomes is setwise
dominated by the refreshment distribution; every such map is
rationalized by an explicit attrition process. Without attrition the map
is identified on the latent-outcome support; with attrition, a
density-ratio condition governs the identified set, and tail behaviour
alone does not determine it. For an unrestricted map the survivors' mean
effect has the familiar trimming bounds; for an item with all categories
reported, the model reduces to a test of no conditioning. Within a
cohort, other waves, dropout patterns and entry-wave items leave the set
unchanged unless restrictions link selection across waves. Under
explicit selection restrictions, survival matching, symmetric matching
and entry-wave correction identify survivor effects. We derive their
biases and give rank conditions under which refreshment schedules
identify curvature in the conditioning path. A Japanese panel
illustrates the results.
\end{abstract}


\setstretch{1.5}
\section{Introduction}\label{introduction}

Research on refreshment samples and research on panel conditioning
impose different restrictions on selection and measurement. The
literature on attrition in panel surveys treats refreshment samples as
the key to identification: a fresh cross-section drawn at a later wave
reveals the population distribution that the retained sample no longer
represents, and under additively separable models of nonignorable
attrition the joint distribution can be recovered
(\citeproc{ref-hirano2001}{Hirano et al. 2001};
\citeproc{ref-deng2013}{Deng et al. 2013}). These models maintain that
survey participation does not alter measurement. In Hirano et al.
(\citeproc{ref-hirano2001}{2001}) the restriction enters through the
refreshment sample, which is treated as a fresh draw from the same
population. The paper discusses repeated response and the possibility
that effort during a first interview affects later response
(\citeproc{ref-hirano2001}{Hirano et al. 2001, 1646, 1648}), but its
attrition model does not explicitly model changes in measurement caused
by prior interviews. Deng et al. (\citeproc{ref-deng2013}{2013}) mention
panel conditioning only in passing, when explaining why attrition bias
in one panel differs from another {[}p.~250{]}. Their
additive-nonignorable model restricts the attrition probability and not
the measurement, and they note that a rotating panel supplies the
equivalent of a refreshment sample so long as each cohort is randomly
selected and administered the same questionnaire
(\citeproc{ref-deng2013}{Deng et al. 2013, 239}).

The literature on panel conditioning documents that participation can
alter reports (\citeproc{ref-warren2012}{Warren and Halpern-Manners
2012}; \citeproc{ref-bailar1975}{Bailar 1975}) and treats attrition as a
threat to its own comparisons. Warren and Halpern-Manners
(\citeproc{ref-warren2012}{2012, 508--10, 517}) discuss attrition and
mode as threats and describe comparisons of old and new cohorts
restricted to two-wave survivors. Halpern-Manners and Warren
(\citeproc{ref-halpernmanners2012}{2012, 1506--8}) compare adjacent
rotation groups of the Current Population Survey within a calendar
month, restricting both to respondents interviewed in both of their
first two months. Halpern-Manners, Warren, and Torche
(\citeproc{ref-halpernmanners2017}{2017, 108--9} and note 14, p.~119)
apply the same survival-matching comparison to the General Social
Survey, explain the threat of attrition processes that differ by cohort,
and give a regression diagnostic for it. Eckman and Bach
(\citeproc{ref-eckmanbach2021}{2021, 57, 64--65}) apply it to four-wave
respondents of the Consumer Expenditure Survey and state the assumption
that respondents at different waves differ only in their exposure to the
survey. A methodological review recommends such restrictions on the
ground that they require no assumption about the form of attrition
(\citeproc{ref-bach2018}{Bach 2018, 14--15}). Das, Toepoel, and van
Soest (\citeproc{ref-dasetal2011}{2011}) bound conditioning effects for
binary items under worst-case attrition and discuss the relation of
their model to that of Hirano et al. (\citeproc{ref-hirano2001}{2001}).
They also note that the literature on panel conditioning typically
assumes attrition to be completely at random or at random given the
first-wave answer, whereas the attrition analysis of Hirano et al.
(\citeproc{ref-hirano2001}{2001}) weakens the assumptions on attrition
but assumes that there is no panel conditioning
(\citeproc{ref-dasetal2011}{Das, Toepoel, and van Soest 2011, 43}).

Official-statistics work on rotation group bias documents systematic
differences by time in sample without being able to say how much of them
is measurement. Krueger, Mas, and Niu (\citeproc{ref-krueger2017}{2017,
260}, Table 1, and p.~262) find that adjusting for nonresponse accounts
for more than one-third of the rise in rotation group bias after 1993,
caution that their observational evidence cannot establish causation
unambiguously, and suggest (p.~264) that the bias need not be an
inherent consequence of repeated interviewing. Williams and Mallows
(\citeproc{ref-williamsmallows1970}{1970, 1338--39, 1342--43}) showed
that such contrasts can arise without any conditioning. If the
probability of response depends on the characteristic under study and
changes between interviews, estimates from units interviewed for the
first and for the second time differ systematically, and restricting the
comparison to units interviewed on both occasions does not remove the
difference {[}p.~1339{]}. The selection term \(\delta_e(t)\) of
Equation~\ref{eq-naive} below is the counterpart here.

In this paper we characterize the sharp identified set for a
deterministic monotone conditioning map under unrestricted attrition,
and we state the additional distributional or mean restrictions that the
comparison designs considered here require. The question is what a
refreshment sample identifies about the conditioning map when the panel
is subject to both nonignorable attrition and panel conditioning. We
show that a candidate map is consistent with the data if and only if the
retention-scaled distribution of the stayers' implied latent outcomes is
setwise dominated by the refreshment distribution
(Theorem~\ref{thm-sharp}). The set so defined is sharp: every point in
it is rationalized by an explicitly constructed attrition process.

The characterization has five implications for identification and survey
design.

First, the constraint uses only the stayers' reported outcome
distribution, the refreshment distribution, and the retention rate. The
stayers' joint distribution with their baseline reports is not needed to
identify the conditioning map; it matters only for reconstructing a
selection process afterwards.

Second, with full retention the constraint collapses to an equality that
identifies the map on the latent-outcome support as a quantile
alignment, with no assumption on selection. As retention falls, the
constraint relaxes into a set of observationally equivalent maps. For
location shifts that set is governed by the ratio of the
selection-weighted latent density to the refreshment density
(Corollary~\ref{cor-funnel}). Tails constrain the set but do not
determine it. When selection is bounded away from zero, a refreshment
density with tails thinner than exponential excludes every perturbation
that moves the implied latent distribution towards such a tail. The
constraint must nevertheless hold at every outcome value, and it can
bind in the centre of the distribution.

Third, two known bounds are related to the set. Das, Toepoel, and van
Soest (\citeproc{ref-dasetal2011}{2011}) derive sharp worst-case bounds
on the population conditioning effect for a binary item, of width equal
to the attrition rate, and, in their online Appendix 1, the
corresponding bounds on the survivors' effect. Their model lets
individual binary responses change at reinterview, whereas a weakly
increasing deterministic map on a two-point support with both values
reported can only fix both points, so the two models are compared and
not nested. For the survivors' mean effect under an unrestricted map,
the sharp bounds are the trimming bounds of Horowitz and Manski
(\citeproc{ref-horowitzmanski1995}{1995}) and Lee
(\citeproc{ref-lee2009}{2009}) (Corollary~\ref{cor-meanbounds}); they
are \(.70\) standard deviations wide at \(80\%\) retention for a
Gaussian refreshment distribution, and a location restriction narrows
them.

Fourth, suppose that attrition does not depend on the latent outcome.
The identified set for a location shift is then the sublevel set of the
refreshment density's shift-ratio function at the inverse retention rate
(Corollary~\ref{cor-independent}): a single point for the Gaussian and
for any density with tails thinner than exponential, and an interval of
half-width \(\lambda^{-1}\log(1/p)\) for the Laplace and logistic
densities. In that case conditioning is identified by a property of the
population's distribution and not by the absence of outcome-dependent
attrition alone, and a plug-in width computed on a trimmed support
measures the trim, not the selection (Remark~\ref{rem-trim}).

Fifth, continuous and discrete items differ. For an item all of whose
categories are reported by the stayers, the only admissible weakly
increasing deterministic map is the identity, so the theorem reduces to
a test of ``no conditioning'' that holds under every attrition process.
Conditioning on such items requires a model beyond weakly increasing
deterministic maps, such as a stochastic transition between response
categories (Corollary~\ref{cor-discrete}).

Within a cohort, the longitudinal structure of a panel does not change
this unless restrictions link the waves. Theorem~\ref{thm-multiwave}
shows that, without restrictions linking selection across waves, the
joint distribution of a cohort's reports over many waves, the timing of
its dropouts, and a battery of items asked once at entry do not narrow
the identified set of the map at a given wave. The identified set of the
map at a wave with a refreshment sample is the two-wave set computed
from that wave's respondents and its refreshment distribution, and the
maps at other waves are unidentified. Every observed pattern of
participation is reproduced by an explicit multiwave completion. Within
the maintained model, additional identifying restrictions must therefore
exclude otherwise feasible completions, and the necessity results below
are statements about the full distribution of the data, not about a few
moments.

We next ask which additional restrictions yield point identification. A
single refreshment sample set-identifies conditioning. Two refreshment
samples at different waves, generating several tenure configurations in
common periods, together with an entry-wave negative-control battery and
a transport restriction on selection with a known loading,
point-identify the conditioning contrasts at the tenures observed at
refreshment episodes (Theorem~\ref{thm-tworef}); a bounded loading gives
a set. Overlapping birth years do not imply the transport restriction,
which is maintained as an assumption. We distinguish two kinds of
refreshment: \emph{age-matched} refreshments, whose birth years overlap
those of the incumbent cohorts (in the Japanese panel of
Section~\ref{sec-implementation}, the 2011 refreshment), and
\emph{generational} refreshments, which recruit a new birth cohort with
no overlap (the 2019 refreshment). For the latter, levels are only
interval-identified under a bound on residual cohort effects, while
increments between incumbent tenures remain identified under the same
restrictions (Corollary~\ref{cor-generational}). We thus obtain the
identification theory behind single-refreshment comparisons and a design
theory for panels that plan several refreshments.

The third part of the paper concerns the designs used in applied work,
which combine a single refreshment sample with a restriction on
attrition, often stated informally, to obtain a point estimate; the
survival-matching comparison of Halpern-Manners, Warren, and Torche
(\citeproc{ref-halpernmanners2017}{2017}) is one example.
Section~\ref{sec-constructive} expresses the comparability requirements
of such designs as explicit restrictions on latent outcome means. Each
is a restriction on the \emph{selection functional}, the survivors'
advantage on the latent outcome, and each identifies the survivors' mean
shift. Survival matching takes the survivors' latent mean from the later
waves of the fresh cohort; a symmetric variant does so after equalizing
the number of selection decisions; and an entry-wave correction takes it
from the continuing cohort's own entry wave, where reports are
unconditioned by construction. The entry-wave functional is the
observable combination used by Das, Toepoel, and van Soest
(\citeproc{ref-dasetal2011}{2011, 43--44}, Assumption 3) under their
stationary-attrition-bias assumption, which restricts conditioned
outcomes and identifies a population conditioning effect. For binary
outcomes, the survivor interpretation and its stationarity assumption on
unconditioned outcomes are also explicit there
(\citeproc{ref-dasetal2011}{Das, Toepoel, and van Soest 2011, 41}, and
online Appendix 1, Assumption 3\(^{\mathrm{Alt}}\), pp.~27--28): the
same functional identifies the survivors' conditioning effect under that
alternative assumption. We formulate the restriction for general outcome
means within the present multiwave framework, where it is a restriction
on the selection functional, and analyse the symmetric comparison
alongside it; the distinction between the two targets matters when
effects are heterogeneous. We derive exact bias identities for each
design and determine, under an explicit model with person-level traits,
independent transient innovations and per-wave response decisions, which
designs are biased under state-dependent and under non-stationary
attrition, and the direction of the bias under state dependence
(Proposition~\ref{prp-fail}). The disagreements between designs then
assess the compatibility of the maintained restrictions without uniquely
identifying a failure mechanism (Section~\ref{sec-diagnostics}).

Section~\ref{sec-design} treats the schedule itself as the object of
choice. A rank condition determines, before any data are collected,
which functionals of the conditioning path a proposed schedule
identifies. The \emph{stride} of the schedule, the greatest common
divisor of the spacings between entry cohorts, fixes a set of directions
that are never identified. A panel refreshed at a common interval
greater than one therefore cannot identify ordinary curvature, whereas
an entry at a coprime spacing, followed for long enough, can restore it
(Theorem~\ref{thm-rank}, Proposition~\ref{prp-stride}). A
balance--information result restates, for published aggregates, the
classical finding that tenure-balanced rotation keeps time-in-sample
bias constant in levels and absent from changes
(\citeproc{ref-bailar1975}{Bailar 1975};
\citeproc{ref-parkkimchoi2001}{Park, Kim, and Choi 2001}). It adds that
when only the balanced aggregate is released, the published series
depends on the time-in-sample path only through its weighted average,
which is itself confounded with the period level, so that no functional
of the path is identified from it (Theorem~\ref{thm-balance}).

Throughout, ``conditioning'' is a measurement map applied to latent
outcomes at reinterview, in the sense of our companion paper on the
identification of time-in-sample effects
(\citeproc{ref-okubo2026panel}{Okubo 2026}) (the \emph{companion
identification paper} below; its numbered results and assumptions refer
to arXiv version 1). We take the two-wave measurement problem as
primitive and ask what cross-sectional anchors add to it. Several formal
treatments are close to ours. Feng, Hu, and Sun
(\citeproc{ref-fenghusun2022}{2022}) restrict response dynamics to
first-order dependence on the previous report, so that measurement error
at a first interview, which has no previous report to condition on,
depends on the current true status alone; combined with rank and
eigenvalue conditions, this restriction point-identifies
misclassification probabilities in the Current Population Survey.
Franguridi and Kosenkova (\citeproc{ref-franguridi2026}{2026}) and
Franguridi et al. (\citeproc{ref-franguridihahn2026}{2026}) study the
identification and estimation of the attrition process from refreshment
samples in the absence of conditioning. Chadi
(\citeproc{ref-chadi2021}{2021}) separates attrition bias from
participation-experience effects in a household panel by exploiting the
coexistence of entrants recruited by the data collector and ``natural''
entrants, under regression-based identifying restrictions that allow
experience effects common to both types of entrant; the sharp set for an
unrestricted conditioning map is not characterized there. Bertoli,
Jakli, and Pascoe (\citeproc{ref-bertoli2026}{2026}) formalize
conditioning bias in two-wave event-impact surveys and propose a dual
randomized design that avoids re-asking identical items. Okubo
(\citeproc{ref-okubo2024}{2024}) states potential-outcome estimands and
identification assumptions for conditioning with a refreshment sample
and applies the survival-matching comparison to the 2011 refreshment of
the panel analysed in Section~\ref{sec-implementation}, under the added
assumption that attrition operates in the same way at corresponding
waves of the two cohorts. Our contribution relative to these works is
the sharp identified set under both frictions in the deterministic
monotone model, its multiwave extension, and a formal statement of the
design restrictions, their biases, and the rank conditions of
refreshment schedules.

\section{Setup and Observables}\label{sec-setup}

A cohort enters at wave 1 with latent outcomes \((Y_1, Y_2^{*})\)
jointly distributed as \(F\) on \(\mathbb{R}^2\), with marginals \(F_1\)
and \(F_2\). Every cohort member reports \(Y_1\) at wave 1. Between
waves, attrition selects stayers: \(S \in \{0,1\}\) with selection
function \[
\pi(y_1, y_2) \;=\; \Pr(S = 1 \mid Y_1 = y_1,\, Y_2^{*} = y_2), \qquad p \;=\; \Pr(S=1) \;=\; \int \pi \, dF \;>\; 0,
\] left unrestricted, so that attrition may depend on the latent wave-2
outcome itself (nonignorable). At wave 2, stayers report a
\emph{conditioned} measurement \[
Y_2 \;=\; c(Y_2^{*}),
\] where \(c\) belongs to \(\mathcal{M}\), the class of increasing
homeomorphisms of \(\mathbb{R}\) (continuous, strictly increasing, and
onto, so that \(c^{-1}\) is defined on all of \(\mathbb{R}\)). (Location
shift \(c(y) = y + h\) is the leading parametric case.) A refreshment
sample drawn at wave 2 from the same population reports \(Y_2^{*}\)
without conditioning (``refreshment validity'').

The observables are the tuple \[
\mathcal{O} \;=\; \bigl(F_1,\; F_2,\; G,\; p\bigr),
\] where \(G\) is the joint distribution of \((Y_1, Y_2)\) among
stayers, with first marginal \(G_1\) and second (reported-outcome)
marginal \(Q\). All four components are directly estimable: \(F_1\) from
the full entry cohort, \(F_2\) from the refreshment sample, \(G\) and
\(p\) from the panel's retention records. Throughout, distributions on
\(\mathbb{R}\) are assumed atomless where quantile alignments are
invoked.

A \emph{structure} is a triple \(\theta = (F, \pi, c)\) with \(F\) a
distribution on \(\mathbb{R}^2\), \(\pi: \mathbb{R}^2 \to [0,1]\)
measurable, and \(c \in \mathcal{M}\). Structure \(\theta\)
\emph{rationalizes} \(\mathcal{O}\) if it reproduces all four
components. The identified set for the conditioning map is \[
\mathcal{C}_{ID} \;=\; \bigl\{\, c' \in \mathcal{M} \;:\; \exists\, (F', \pi') \text{ such that } (F', \pi', c') \text{ rationalizes } \mathcal{O} \,\bigr\}.
\]

We use two pieces of notation. For \(c' \in \mathcal{M}\), write
\(Q_{c'} := (c'^{-1})_{\#} Q\) for the pullback of the stayers'
reported-outcome distribution, that is, the distribution the stayers'
latent \(Y_2^{*}\) must have had if the conditioning map were \(c'\).
For measures \(\mu, \nu\) write \(\mu \le \nu\) for setwise domination:
\(\mu(A) \le \nu(A)\) for every Borel \(A\) (equivalently
\(d\mu/d\nu \le 1\) where densities exist).

\section{The Main Theorem}\label{sec-main}

The main tool is a completion lemma: any candidate stayer
sub-distribution consistent with the marginals can be embedded in a full
structure. The lemma is an instance of the existence results for
measures with given marginals of Strassen
(\citeproc{ref-strassen1965}{1965}, Theorems 6--7, pp.~432--434), which
include domination constraints; we use the explicit residual-product
construction below instead of appealing to that general theory.

\begin{lemma}[Completion
Lemma]\protect\hypertarget{lem-completion}{}\label{lem-completion}

Let \(F_1, F_2\) be probability measures on \(\mathbb{R}\) and let
\(\mu\) be a finite measure on \(\mathbb{R}^2\) with total mass
\(p \in (0,1)\) and marginals \(\mu_1, \mu_2\) satisfying
\(\mu_1 \le F_1\) and \(\mu_2 \le F_2\). Then \[
F' \;:=\; \mu \;+\; \frac{(F_1 - \mu_1) \otimes (F_2 - \mu_2)}{1-p}
\] is a probability measure on \(\mathbb{R}^2\) with marginals exactly
\(F_1\) and \(F_2\), satisfying \(\mu \le F'\). Consequently
\(\pi' := d\mu / dF'\) is a valid selection function
(\(0 \le \pi' \le 1\)) whose stayer measure is exactly \(\mu\) and whose
retention rate is \(p\).

\end{lemma}

\begin{proof}
\((F_1 - \mu_1)\) and \((F_2 - \mu_2)\) are nonnegative measures (by the
domination hypotheses) each of mass \(1-p\), so their product has mass
\((1-p)^2\) and the second term has mass \(1-p\); \(F'\) has mass \(1\).
Its first marginal is
\(\mu_1 + (F_1 - \mu_1)\cdot \frac{(F_2-\mu_2)(\mathbb{R})}{1-p} = \mu_1 + (F_1 - \mu_1) = F_1\),
and symmetrically for the second. Domination \(\mu \le F'\) holds
because the added term is nonnegative. The Radon--Nikodym derivative
\(\pi' = d\mu/dF'\) exists and lies in \([0,1]\) by domination; by
construction \(\pi' \, dF' = d\mu\), so the stayer measure is \(\mu\)
and \(\int \pi' dF' = p\).
\end{proof}

The completion has a design interpretation: the non-stayers are assigned
the independent coupling of the leftover marginals. The data do not
restrict the dependence between \(Y_1\) and \(Y_2^{*}\) among units
never observed at wave 2, and the construction relies on this.

\begin{theorem}[Sharp identified set for the conditioning
map]\protect\hypertarget{thm-sharp}{}\label{thm-sharp}

Under refreshment validity, data compatibility \(p\,G_1 \le F_1\) (the
stayers' first-wave sub-distribution fits under the full cohort's, which
holds automatically when the first wave is observed for the whole
cohort), and \(c \in \mathcal{M}\), the identified set is exactly \[
\mathcal{C}_{ID} \;=\; \bigl\{\, c' \in \mathcal{M} \;:\; p \cdot Q_{c'} \;\le\; F_2 \,\bigr\},
\] and for every \(c' \in \mathcal{C}_{ID}\) the rationalizing pair
\((F', \pi')\) can be taken to be the completion of
\(\mu_{c'} := (\mathrm{id} \times c'^{-1})_{\#}\,(p\,G)\) given by
Lemma~\ref{lem-completion}. The data, and membership in
\(\mathcal{C}_{ID}\), depend on \(c'\) only through its inverse on the
reported scale restricted to the support of \(Q\): maps whose inverses
agree there have the same pullback and are observationally equivalent,
and off the support of \(Q\) the inverse is restricted only by
membership in \(\mathcal{M}\).

\end{theorem}

\begin{proof}
\emph{Necessity.} Let \((F', \pi', c')\) rationalize \(\mathcal{O}\).
Since \(c'\) is a continuous, strictly increasing map onto
\(\mathbb{R}\), it has an inverse on \(\mathbb{R}\), and the stayers'
latent joint measure is determined by the data and \(c'\): \[
\mu' \;=\; (\mathrm{id} \times c'^{-1})_{\#}\,(p\,G) \;=\; \mu_{c'},
\] because the observed stayer measure of \((Y_1, Y_2)\) is \(p\,G\) and
\(Y_2^{*} = c'^{-1}(Y_2)\) almost surely among stayers. Selection
probabilities cannot exceed one: \(\mu' = \pi' \, dF' \le dF'\) setwise.
Taking second marginals, \(p\, Q_{c'} = \mu'_2 \le F'_2 = F_2\), where
the last equality is rationalization of the refreshment marginal. This
is the stated constraint. (The first-marginal condition
\(p\,G_1 \le F_1\) holds automatically in any data set generated by a
panel, since the stayers' wave-1 sub-measure is part of the cohort's
wave-1 measure.)

\emph{Sufficiency.} Let \(c' \in \mathcal{M}\) satisfy
\(p\,Q_{c'} \le F_2\). Set \(\mu = \mu_{c'}\); its mass is \(p\), its
first marginal is \(p\,G_1 \le F_1\) (automatic, as above), its second
is \(p\,Q_{c'} \le F_2\) (the hypothesis). If \(p = 1\), both marginal
inequalities are between probability measures and therefore equalities,
so \(F' := \mu_{c'}\) and \(\pi' \equiv 1\) have marginals
\((F_1, F_2)\), stayer measure \(\mu_{c'}\) and retention one. If
\(p < 1\), Lemma~\ref{lem-completion} yields \((F', \pi')\) with
marginals \((F_1, F_2)\), stayer measure \(\mu_{c'}\), and retention
\(p\). It remains to check that \((F', \pi', c')\) reproduces
\(\mathcal{O}\): \(F_1, F_2\) are the marginals of \(F'\) by
construction; the stayer joint of \((Y_1, c'(Y_2^{*}))\) is
\((\mathrm{id} \times c')_{\#}\,\mu_{c'} / p = G\); and the retention
rate is \(p\).
\end{proof}

By the theorem, the stayers' joint distribution \(G\) restricts \(c\)
only through its outcome marginal \(Q\): for this question, the
association between baseline and later reports is summarized by one
distribution. The constraint is a condition on likelihood ratios. Where
densities exist it reads \(p\, q_{c'}(y) \le f_2(y)\) for almost every
\(y\), so domination must hold at every point of the support; stochastic
ordering, which compares distribution functions, is weaker. Since \(p\)
is observed, the problem differs from missing-data settings in which the
selection rate itself must be bounded.

\subsection{Corollaries: what attrition does to
identification}\label{sec-corollaries}

\begin{corollary}[No attrition implies nonparametric identification on
the
support]\protect\hypertarget{cor-noattrition}{}\label{cor-noattrition}

Let \(F_2\) and \(Q\) be atomless. If \(p = 1\) then every
\(c' \in \mathcal{C}_{ID}\) agrees with \(c\) on the support of \(F_2\):
the conditioning map is identified on the latent-outcome support, where
it is the increasing rearrangement transporting \(F_2\) to \(Q\), \[
c(y) \;=\; Q^{-1}\!\bigl(F_2(y)\bigr) \qquad \text{for } F_2\text{-almost every } y
\] (writing CDFs for the measures), with no restriction on the
dependence between attrition and outcomes---vacuously, since there is
none. Off the support of \(F_2\) a member of \(\mathcal{C}_{ID}\) is
restricted only by membership in \(\mathcal{M}\), so
\(\mathcal{C}_{ID} = \{c\}\) if and only if \(F_2\) has full support; in
general the identified object is the restriction of \(c\) to
\(\operatorname{supp} F_2\), and maps that agree there are
observationally equivalent (Theorem~\ref{thm-sharp}).

\end{corollary}

\begin{proof}
With \(p = 1\), the constraint \(Q_{c'} \le F_2\) holds between two
probability measures, and setwise domination between measures of equal
total mass forces equality: \(Q_{c'} = F_2\), that is, \(c'\) pushes
\(F_2\) forward to \(Q\), so \(F_2(y) = Q(c'(y))\) for every \(y\).
Write \(Q^{-1}(u) := \inf\{z : Q(z) \ge u\}\). At a point \(y\) of the
support of \(F_2\) every left and every right neighbourhood of which
carries mass---all points of the support except the countably many
endpoints of its gaps and its extreme points---\(F_2\) is strictly
increasing on both sides of \(y\), so, \(c'\) being continuous and
increasing, \(Q\) is strictly increasing across \(c'(y)\) and the level
set \(\{z : Q(z) = F_2(y)\}\) is the single point \(Q^{-1}(F_2(y))\);
hence \(c'(y) = Q^{-1}(F_2(y)) = c(y)\), because \(c\) satisfies the
same equation. The exceptional points are countably many, hence
\(F_2\)-null since \(F_2\) is atomless, so this holds \(F_2\)-almost
everywhere; and two continuous functions that agree \(F_2\)-almost
everywhere agree on the support, a set of full \(F_2\)-measure being
dense in it. Outside the support the equation restricts \(c'\) only
through its values at the ends of the gap: on a gap \((a,b)\) it holds
for every increasing continuous \(c'\) with \(c'(a) = c(a)\) and
\(c'(b) = c(b)\), since \(Q\) is constant on \([c(a), c(b)]\), and
beyond the ends of a bounded support \(c'\) is free up to monotonicity
and continuity: with \(F_2 = Q\) uniform on \([0,1]\), the identity and
the map equal to \(y\) on \([0,1]\), \(2y\) below \(0\) and \(2y - 1\)
above \(1\) both belong to \(\mathcal{M}\) and rationalize the same
data.
\end{proof}

Corollary~\ref{cor-noattrition} is the map-level counterpart of the
no-attrition case of Das, Toepoel, and van Soest
(\citeproc{ref-dasetal2011}{2011}, Example 1, pp.~38--39), in which
comparing second-time with first-time respondents identifies the
conditioning effect on a binary item. In the absence of attrition,
matching quantiles of the stayer and refreshment distributions
identifies the conditioning map on the latent-outcome support,
nonparametrically for each item. Under the maintained deterministic
monotone-map model and refreshment validity, only attrition weakens the
identifying restriction.

\begin{corollary}[The identification
funnel]\protect\hypertarget{cor-funnel}{}\label{cor-funnel}

Let densities exist, and let
\(\bar\pi(y) := \mathrm{E}[\pi(Y_1, Y_2^{*}) \mid Y_2^{*} = y]\) be the
outcome-conditional selection function of the true structure, so that
the stayers' latent-outcome density satisfies
\(p\,q_c(y) = \bar\pi(y) f_2(y)\), where \(q_c\) denotes the density of
\(Q_c = (c^{-1})_{\#}Q\). Then \(c' \in \mathcal{C}_{ID}\) if and only
if \[
\bar\pi\bigl(g(y)\bigr)\, f_2\bigl(g(y)\bigr)\, g'(y) \;\le\; f_2(y) \quad \text{for a.e. } y, \qquad g := c^{-1} \circ c' ,
\] the Jacobian \(g'\) entering because the constraint compares
measures, not ordinates (the density form requires \(g\) to be
absolutely continuous with \(g' > 0\) a.e., which we assume here in
addition to membership in \(\mathcal{M}\); the measure form of
Theorem~\ref{thm-sharp} does not). For a location perturbation
\(g(y) = y + \varepsilon\) the condition reads
\(\Phi(\varepsilon) \le 1\), where \[
\Phi(\varepsilon) \;:=\; \operatorname*{ess\,sup}_y \frac{\bar\pi(y + \varepsilon)\, f_2(y + \varepsilon)}{f_2(y)} .
\] In particular: (i) \(\varepsilon = 0\) always satisfies it, since
\(\Phi(0) = \operatorname{ess\,sup}\bar\pi \le 1\); (ii) if
\(\bar\pi \ge \pi_{\min} > 0\) a.e., every feasible \(\varepsilon\)
satisfies \(R(\varepsilon) \le 1/\pi_{\min}\), where \(R\) is the
shift-ratio function of Corollary~\ref{cor-independent}; hence if
\(\log f_2(y)/|y| \to -\infty\) as \(y \to -\infty\) (as
\(y \to +\infty\)), no \(\varepsilon > 0\) (no \(\varepsilon < 0\)) is
feasible; (iii) if \(\operatorname{ess\,sup}\bar\pi \le 1 - \kappa\) for
some \(\kappa \in (0, 1)\) and \(\log f_2\) is \(L\)-Lipschitz, every
\(\varepsilon\) with \(|\varepsilon| \le L^{-1}\log\{1/(1-\kappa)\}\) is
feasible. The condition \(\Phi(\varepsilon) \le 1\) is global in \(y\):
it can fail in the centre of the distribution although \(\bar\pi\)
decays in both tails (Example below).

\end{corollary}

\begin{proof}
The candidate latent distribution is
\(Q_{c'} = (c'^{-1})_{\#}Q = (c'^{-1} \circ c)_{\#} Q_c = (g^{-1})_{\#} Q_c\).
Its density follows from the change-of-variables identity
\(\int_A q_{c'}\,dy = Q_c(g(A)) = \int_A q_c(g(y))\,g'(y)\,dy\) for
every Borel \(A\), so \(q_{c'}(y) = q_c(g(y))\,g'(y)\) a.e. Substituting
\(p\,q_c = \bar\pi f_2\), the sharp constraint
\(p\,q_{c'}(y) \le f_2(y)\) a.e. (Theorem~\ref{thm-sharp} in density
form) is exactly the display, and with \(g(y) = y + \varepsilon\) it is
\(\Phi(\varepsilon) \le 1\). (i) is immediate. (ii)
\(\Phi(\varepsilon) \ge \pi_{\min} R(\varepsilon)\), and
\(R(\varepsilon) = \infty\) for every \(\varepsilon > 0\) under the
lower-tail condition (and for every \(\varepsilon < 0\) under the
upper-tail condition) by the telescoping argument in the proof of
Corollary~\ref{cor-independent}, which uses only \(f_2\). (iii)
\(\log f_2(y + \varepsilon) - \log f_2(y) \le L|\varepsilon|\) for every
\(y\), so
\(\Phi(\varepsilon) \le (1 - \kappa)\,e^{L|\varepsilon|} \le 1\).
\end{proof}

Corollary~\ref{cor-funnel} makes precise how attrition and conditioning
are confounded: a perturbation of the map is admissible if and only if
the selection-weighted latent density, moved by the perturbation, still
fits under the refreshment density at every outcome value. Parts (ii)
and (iii) give a necessary and a sufficient condition; neither is a tail
criterion alone.

\emph{Example (central binding).} Let \(F_2 = N(0,1)\),
\(Q = N(0, \sigma^2)\) with \(\sigma \in (0,1)\), \(p = \sigma\), and
let the true map be the identity. The implied selection function
\(\bar\pi(y) = p\,q(y)/f_2(y) = \exp\{-(\sigma^{-2} - 1)y^2/2\}\) lies
in \((0, 1]\) and decays in both tails, and the stayers' latent density
has strictly lighter tails than \(f_2\) on both sides. Yet for every
\(\varepsilon \ne 0\), \[
\Phi(\varepsilon) \;=\; \sup_y \frac{p\,q(y + \varepsilon)}{f_2(y)} \;=\; \exp\Bigl\{\frac{\varepsilon^2}{2(1 - \sigma^2)}\Bigr\} \;>\; 1 ,
\] the supremum being attained at \(y = -\varepsilon/(1 - \sigma^2)\):
the identified set is \(\{0\}\) (\(\sigma = p = .5\) gives a concrete
structure). The constraint binds at the centre, where \(\bar\pi = 1\),
although there is slack in both tails. The example also shows that
non-identification requires an actual feasible \(\varepsilon \ne 0\):
here \(\inf_{\varepsilon \ne 0}\Phi(\varepsilon) = 1\), but the infimum
is not attained away from zero.

\emph{Example (logistic selection).} Let \(f_2\) be Gaussian with
variance \(\sigma^2\) and let the true selection be logistic in the
outcome, \(\bar\pi(y) = \Lambda(\alpha + \beta y)\) with \(\beta > 0\)
(better-off respondents stay). For a location perturbation
\(c'(y) = c(y) + \Delta\) the constraint reads
\(\Lambda(\alpha + \beta(y+\Delta))\, f_2(y+\Delta) \le f_2(y)\) for all
\(y\). No \(\Delta < 0\) is feasible (the constraint fails in the upper
tail, where \(\bar\pi \to 1\)), and in the lower tail the log-ratio
behaves as \(y(\beta - \Delta/\sigma^2)\), so no
\(\Delta > \beta\sigma^2\) is feasible either. Because \(f_2\) is
unimodal, the set is an interval (Theorem~\ref{thm-bounds}),
\([0, \Delta^{*}]\) with \[
\Delta^{*} \;\le\; \beta\sigma^2 ,
\] and the endpoint depends on the retention level through \(\alpha\):
the interior maximum of the constraint ratio (attained where
\(1 - \Lambda = \Delta/\beta\sigma^2\)) pushes \(\Delta^{*}\) towards
\(0\) as retention rises (\(\alpha \to \infty\), where the set must
collapse, since \(p \to 1\) recovers Corollary~\ref{cor-noattrition})
and towards the envelope \(\beta\sigma^2\) as retention falls. The
envelope is attained once retention is low enough that two conditions
hold: the lower-tail limit of the constraint ratio,
\(\exp\{\alpha + \beta\Delta - \Delta^2/(2\sigma^2)\}\) evaluated at
\(\Delta = \beta\sigma^2\), is at most one, that is,
\(\alpha \le -\beta^2\sigma^2/2\); and the interior supremum is also at
most one. With \(\sigma^2 = \beta = 1\) and \(\alpha = -1\), retention
\(.30\), the supremum at \(\Delta = 1\) is \(e^{-1/2}\), so
\(\Delta^{*} = \beta\sigma^2\) exactly. In this model the ceiling on the
ambiguity is the product of selection strength (\(\beta\)) and outcome
dispersion (\(\sigma^2\)): with mild outcome selection
(\(\beta\sigma \approx 0.1\) in standardized units) at most a tenth of a
standard deviation of conditioning can be confused with selection,
whatever the retention rate. This is a continuous-outcome counterpart of
the worst-case bounds of Das, Toepoel, and van Soest
(\citeproc{ref-dasetal2011}{2011, 41}) on the population effect, which
rest on the bounded-outcome logic of Manski
(\citeproc{ref-manski1989}{1989, 345--46}, eqs. (5)--(6)): those bounds
have width equal to the attrition rate because they let the attriters'
unobserved answers range over the whole outcome space. Their bounds on
the survivor effect (\citeproc{ref-dasetal2011}{Das, Toepoel, and van
Soest 2011}, online Appendix 1, p.~26) already require the survivors'
implied unconditioned mass to fit under the refreshment distribution, as
Corollary~\ref{cor-meanbounds} does for continuous outcomes;
Corollary~\ref{cor-funnel} adds the map restriction, under which the
survivors' implied density, moved by the candidate shift, must fit under
\(f_2\) pointwise. The example is stated for location shifts
\(c(y) = y + h\); for a general map \(c\) the perturbation
\(c' = c + \Delta\) does not reduce to \(g(y) = y + \Delta\), and the
constraint must be evaluated on \(g = c^{-1} \circ c'\) directly.

\begin{corollary}[Attrition independent of the
outcome]\protect\hypertarget{cor-independent}{}\label{cor-independent}

Let the conditioning map be a location shift, \(c(y) = y + h\), and
suppose the true selection does not depend on the latent outcome,
\(\bar\pi(y) \equiv p\), so that the stayers' latent distribution equals
the population's, \(Q_c = F_2\). Define the \emph{shift-ratio function}
of the refreshment density, \[
R(\varepsilon) \;:=\; \operatorname*{ess\,sup}_y \frac{f_2(y + \varepsilon)}{f_2(y)} \;\in\; [1, \infty], \qquad R(0) = 1 .
\] Then the identified set of Theorem~\ref{thm-sharp} for the shift is
\[
\mathcal{H}_{ID} \;=\; h + \bigl\{\, \varepsilon : R(\varepsilon) \le 1/p \,\bigr\} .
\] In particular:

\begin{enumerate}
\def\labelenumi{(\roman{enumi})}
\item
  if \(\log f_2(y)/|y| \to -\infty\) as \(y \to -\infty\) and as
  \(y \to +\infty\) (tails thinner than exponential on both sides, as
  for the Gaussian), then \(R(\varepsilon) = \infty\) for every
  \(\varepsilon \ne 0\) and \(\mathcal{H}_{ID} = \{h\}\): the shift is
  point identified at every retention rate \(p < 1\), although attrition
  is present and its mechanism is unknown to the analyst; if the
  condition holds in the lower (upper) tail only, no \(\varepsilon > 0\)
  (\(\varepsilon < 0\)) belongs to the set;
\item
  if \(\log f_2(y)/|y| \to -\lambda\) as \(y \to \pm\infty\) for some
  \(\lambda > 0\), then \(R(\varepsilon) \ge e^{\lambda|\varepsilon|}\)
  and
  \(\mathcal{H}_{ID} \subseteq [\,h - \lambda^{-1}\log(1/p),\; h + \lambda^{-1}\log(1/p)\,]\);
  if in addition \(\log f_2\) is \(\lambda\)-Lipschitz (as for the
  Laplace density \(f_2(y) \propto e^{-\lambda|y|}\) and, with
  \(\lambda = 1\), the standard logistic density), then
  \(R(\varepsilon) = e^{\lambda|\varepsilon|}\) and the inclusion is an
  equality, an interval of width \(2\lambda^{-1}\log(1/p)\) that
  vanishes as \(p \to 1\) and grows without bound as \(p \to 0\);
\item
  if \(\log f_2\) is \(L\)-Lipschitz,
  \(\mathcal{H}_{ID} \supseteq [\,h - L^{-1}\log(1/p),\; h + L^{-1}\log(1/p)\,]\),
  and if \(f_2\) is unimodal, \(\mathcal{H}_{ID}\) is an interval
  (Theorem~\ref{thm-bounds}). For the standard Cauchy density,
  \(R(\varepsilon) = \bigl(\tfrac{|\varepsilon|}{2} + \sqrt{1 + \varepsilon^2/4}\bigr)^2\)
  and
  \(\mathcal{H}_{ID} = [\,h - (1-p)/\sqrt{p},\; h + (1-p)/\sqrt{p}\,]\).
  Tail behaviour alone does not determine \(\mathcal{H}_{ID}\): a
  smooth, strictly positive density with Cauchy tails, or with logistic
  tails, can have a disconnected identified set.
\end{enumerate}

\end{corollary}

\begin{proof}
With \(Q_c = F_2\) the reported density is \(q(y) = f_2(y - h)\). A
candidate \(c'(y) = y + h'\) implies the latent density
\(q_{c'}(y) = q(y + h') = f_2(y + \varepsilon)\) with
\(\varepsilon := h' - h\), and the constraint of
Theorem~\ref{thm-sharp}, \(p\,q_{c'} \le f_2\) a.e., reads
\(p f_2(y + \varepsilon) \le f_2(y)\) for a.e. \(y\),
i.e.~\(R(\varepsilon) \le 1/p\); sharpness is Theorem~\ref{thm-sharp}'s.
Parts (i) and (ii) rest on a telescoping bound. If
\(R(\varepsilon) = M < \infty\) for some \(\varepsilon > 0\), then
\(f_2(y) \ge f_2(y + \varepsilon)/M\) for a.e. \(y\); the set of \(y_0\)
for which some \(y_0 - n\varepsilon\), \(n \ge 1\), falls in the
exceptional null set is null, so for a.e. \(y_0\) iteration gives
\(\log f_2(y_0 - n\varepsilon) \ge \log f_2(y_0) - n\log M\) for every
\(n\), and hence
\(\liminf_{n\to\infty} \log f_2(y_0 - n\varepsilon)/(n\varepsilon) \ge -\varepsilon^{-1}\log M\).
(i) Under the lower-tail condition the left side is \(-\infty\), a
contradiction; so \(R(\varepsilon) = \infty\) for every
\(\varepsilon > 0\), and the case \(\varepsilon < 0\) uses the upper
tail in the same way. (ii) Under the tail condition the left side equals
\(-\lambda\), so \(\log M \ge \lambda\varepsilon\); symmetrically for
\(\varepsilon < 0\). If \(\log f_2\) is \(\lambda\)-Lipschitz,
\(\log f_2(y + \varepsilon) - \log f_2(y) \le \lambda|\varepsilon|\) for
every \(y\), so \(R(\varepsilon) = e^{\lambda|\varepsilon|}\). The
Laplace log-density \(-\lambda|y|\) is \(\lambda\)-Lipschitz with tail
rate \(\lambda\); the logistic log-density has derivative
\(-\tanh(y/2) \in (-1, 1)\) and tail rate one. (iii) The Lipschitz bound
gives \(R(\varepsilon) \le e^{L|\varepsilon|}\), and the interval
statement is Theorem~\ref{thm-bounds} with
\(q(\cdot) = f_2(\cdot - h)\). For \(f_2(y) \propto (1 + y^2)^{-1}\),
the ratio \(r(y) = (1 + y^2)/(1 + (y+\varepsilon)^2)\) is continuous and
tends to one at both tails; \(r'(y) = 0\) reduces to
\(y^2 + \varepsilon y - 1 = 0\), and at
\(y^{*} = -\tfrac{1}{2}(\varepsilon + \sqrt{\varepsilon^2 + 4})\) one
finds
\(r(y^{*}) = \bigl(\tfrac{\varepsilon}{2} + \sqrt{1 + \varepsilon^2/4}\bigr)^2 = e^{2\operatorname{asinh}(\varepsilon/2)}\)
for \(\varepsilon \ge 0\), so \(R(\varepsilon) \le 1/p\) iff
\(|\varepsilon| \le 2\sinh\{\tfrac{1}{2}\log(1/p)\} = (1-p)/\sqrt{p}\).
For the last sentence, let
\(f(y) \propto (1 + y^2)^{-1}\exp\{2e^{-y^2}\cos(20\pi y)\}\) and
\(p = 1/2\). The cosine factor has period \(0.1\), so for
\(\varepsilon = 0.1\) the log-ratio is at most
\(0.1 + 0.2\sqrt{2/e} < \log 2\) and \(0.1\) is feasible; but at
\(y = 0.05\) the log-ratio for \(\varepsilon = 0.05\) is
\(\log f(0.1) - \log f(0.05) \approx 3.97 > \log 2\), so \(0.05\) is
not. Replacing \((1 + y^2)^{-1}\) by the logistic density gives the same
conclusion with exponential tails, and there \(R(1) \ne e\) although
\(R(\varepsilon) \ge e^{|\varepsilon|}\) (Appendix, deterministic
checks).
\end{proof}

Under outcome-independent attrition, therefore, the identified set is a
sublevel set of the refreshment density's shift-ratio function. Point
identification follows from a property of the refreshment density that
can be estimated; it does not follow from knowledge that attrition is
outcome-independent. The corollary has three consequences. First, the
width of the set does not measure selection. For the Laplace, logistic
and Cauchy densities, and more generally whenever \(\log f_2\) is
Lipschitz, the set is nondegenerate even when attrition is completely
ignorable, whereas with tails thinner than exponential it is a point
even at low retention; a plug-in width should be interpreted accordingly
(Remark~\ref{rem-trim}). Second, under outcome-dependent attrition,
whether the set is degenerate depends on the whole function \(\Phi\) of
Corollary~\ref{cor-funnel} and not on tails alone: nondegeneracy means
that \(\Phi(\varepsilon) \le 1\) for some \(\varepsilon \ne 0\), which
is the hypothesis of Theorem~\ref{thm-tworef}(b) under the marginal
version of (A1). Third, no selection process makes the stayers' latent
density a Gaussian with the population's variance and a different mean.
The ratio \(p\exp(\delta y - \delta^2/2)\) of such a stayer density to
the population density is unbounded in the tail towards which the mean
moves, so no admissible selection process realizes it; under
outcome-independent attrition the Gaussian is therefore a case of point
identification.

\begin{refremark}[What a funnel measures]
Corollary~\ref{cor-funnel} gives the funnel's shape under a specified
selection; Corollary~\ref{cor-independent} gives it under no selection
on the outcome. Read together, they show that a perturbation is
admissible where the stayers' latent density, moved by it, still fits
under \(f_2\). The slack has two sources, which the data do not
distinguish: selection that keeps the stayers' latent density below
\(f_2\) wherever the moved density must fit, and a refreshment density
whose shape can absorb a shift. Only the second is visible from the
refreshment sample alone.

\label{rem-tp1}

\end{refremark}

\begin{refremark}[Point identification from the far tails and the effect
of a trimmed support]
Part (i) of Corollary~\ref{cor-independent} rests on the extreme tails:
for a Gaussian \(f_2\) with variance \(\sigma^2\) and a shift
\(\varepsilon > 0\), the constraint \(p f_2(y+\varepsilon) \le f_2(y)\)
fails only where
\(y \le -\{\sigma^2\log(1/p)/\varepsilon + \varepsilon/2\}\), that is,
at a distance of order \(\sigma^2\log(1/p)/\varepsilon\) into the tail.
Small shifts are excluded only far into the tails, where no finite
sample has observations. If the constraint is imposed only on
\(|y| \le Y\), as in the trimmed implementation considered here (which
trims the support at a fraction of the mode, as in Simulation 1), the
effective set for a shift under outcome-independent attrition is not
\(\{h\}\) but the interval
\(|\varepsilon| \lesssim \sigma^2\log(1/p)/Y\) (to first order in
\(\varepsilon\)), and a slack \(\kappa\) in the criterion,
\(p\,q_{c'} \le (1 + \kappa) f_2\), widens it to
\(\sigma^2\log\{(1+\kappa)/p\}/Y\). With \(Y = 2.15\sigma\) (the
10\%-of-mode trim), \(p = .8\), and \(\kappa = .10\), this gives a
half-width of \(.15\sigma\); the exact population version of the same
trimmed and relaxed criterion has width \(.31\sigma\), against mean
plug-in widths of \(.25\)--\(.26\) in Simulation 1 at \(\beta = 0\),
\(p = .8\), where the identified set is a point. To a good
approximation, therefore, the plug-in width reflects the trim and the
slack and not selection, and a positive numerical width is not
sufficient evidence that identification fails. For this reason
Section~\ref{sec-implementation} treats plug-in criteria only as
exploratory checks.

\label{rem-trim}

\end{refremark}

\subsection{Discrete outcomes}\label{sec-discrete}

The class \(\mathcal{M}\) is rich when the outcome is continuous and
nearly empty when it is discrete, so the theorem is used differently for
the two kinds of survey item. Let the outcome take values in a finite
ordered set \(\mathcal{Y} = \{y_1 < \dots < y_m\}\) and let conditioning
be a weakly increasing map \(c: \mathcal{Y} \to \mathcal{Y}\), so that
maps merging adjacent categories are allowed. Because a map may merge,
the stayers' latent sub-distribution is no longer determined by their
reports, and a structure must specify it; the identified set is defined
as before, as the set of maps for which some structure rationalizes the
observables.

\begin{corollary}[Discrete outcomes: mass
domination]\protect\hypertarget{cor-discrete}{}\label{cor-discrete}

Under refreshment validity, a weakly increasing \(c'\) belongs to the
identified set if and only if
\begin{equation}\phantomsection\label{eq-mass}{
p\,Q(\{a\}) \;\le\; F_2\bigl(c'^{-1}(\{a\})\bigr) \qquad \text{for every } a \in \mathcal{Y} :
}\end{equation} the retention-scaled mass the stayers report in each
category must fit under the population mass of the latent categories the
map sends there. Consequently: (i) the image of \(c'\) must contain the
set \(\mathcal{Y}_Q\) of categories the stayers report,
\(\mathcal{Y}_Q \subseteq c'(\mathcal{Y})\); categories in the image
that the stayers do not report impose no constraint; (ii) when the
stayers report every category, \(c'\) is onto and hence the identity,
and the identified set is \(\{\mathrm{id}\}\) if
\(p\,Q(\{a\}) \le F_2(\{a\})\) for every \(a\) and empty otherwise;
(iii) for a binary item the same holds whenever both values are
reported. When a category goes unreported, the identity can remain
feasible and non-identity maps can be feasible, even uniquely: with
\(F_2\) uniform on \(\{0, 1, 2\}\), \(p = 1\), and
\(Q = (2/3, 0, 1/3)\), the only feasible map sends \((0, 1, 2)\) to
\((0, 0, 2)\). In case (ii) the deterministic model has no identifying
content beyond a specification test: the inequality
\(\max_a p\,Q(\{a\})/F_2(\{a\}) \le 1\) is a sharp implication of ``no
conditioning'' that holds under every attrition process. Its violation
is evidence against the maintained model (refreshment validity,
comparable coding of the item in the panel and the refreshment sample,
and deterministic monotone measurement) and, if the first two are
credible, evidence of conditioning that no weakly increasing
deterministic map represents, such as a stochastic change of individual
answers (the model of Das, Toepoel, and van Soest
(\citeproc{ref-dasetal2011}{2011}) for binary items).

\end{corollary}

\begin{proof}
\emph{Necessity.} A rationalizing structure has a stayers' latent
sub-measure \(\nu\) on \(\mathcal{Y}\) (the second marginal of the
stayers' latent joint) with \(\nu \le F_2\), because selection
probabilities are at most one, and with \(c'_{\#}\nu = p\,Q\); hence
\(p\,Q(\{a\}) = \nu(c'^{-1}(\{a\})) \le F_2(c'^{-1}(\{a\}))\).
\emph{Sufficiency.} On each block \(c'^{-1}(\{a\})\) let \(\nu\) be
\(F_2\) restricted to the block and scaled by
\(p\,Q(\{a\})/F_2(c'^{-1}(\{a\})) \le 1\) (zero if \(p\,Q(\{a\}) = 0\));
then \(\nu \le F_2\) and \(c'_{\#}\nu = p\,Q\). Couple \(\nu\) with the
entry report by giving the stayers whose latent category \(b\) lies in
the block \(c'^{-1}(\{a\})\) the observed conditional law of \(Y_1\)
given the report \(a\): the joint stayer measure
\(\mu(dy_1 \times \{b\}) := G(dy_1 \mid Y_2 = a)\,\nu(\{b\})\) has first
marginal \(\sum_a G(dy_1 \mid Y_2 = a)\,p\,Q(\{a\}) = p\,G_1 \le F_1\)
and second marginal \(\nu \le F_2\), and Lemma~\ref{lem-completion},
which holds on any measurable space, completes the structure when
\(p < 1\); when \(p = 1\) the constraints force \(\nu = F_2\), and
\(\mu\) is itself the structure. (i): a category outside the image of
\(c'\) has empty preimage, so Equation~\ref{eq-mass} forces
\(p\,Q(\{a\}) = 0\) there; every reported category therefore lies in the
image, while an unreported category in the image satisfies
Equation~\ref{eq-mass} trivially. (ii): if every category is reported,
the image is \(\mathcal{Y}\), and a weakly increasing map of a finite
chain onto itself is strictly increasing, hence the identity. (iii) is
(ii) with \(m = 2\). In the example, Equation~\ref{eq-mass} at \(a = 0\)
requires the preimage of \(0\) to carry mass at least \(2/3\), hence to
be \(\{0, 1\}\), and at \(a = 2\) requires \(2\) in the preimage of
\(2\). The ``every attrition process'' clause is the necessity
direction, which uses only \(\nu \le F_2\).
\end{proof}

The corollary determines the scope of the deterministic model for the
two kinds of survey item. For continuous items, and for ordinal items
with enough categories to be treated as continuous, the class of maps is
rich and the funnel is the object of interest. For items with a handful
of categories, all of which the stayers use (most attitude and behaviour
items in panels of the size considered here), the only admissible weakly
increasing deterministic map is the identity, so Theorem~\ref{thm-sharp}
reduces to a test: the retention-scaled stayer mass in every category
must fit under the refreshment mass. Section~\ref{sec-implementation}
applies this criterion. If shifts in the use of extreme and middle
categories, such as those examined in Section~\ref{sec-implementation},
reflect conditioning, they are incompatible with weakly increasing
deterministic maps (such a map sending \(5\) to \(4\) would have to do
so for everyone, and the stayers still report \(5\)). Stochastic partial
merging of categories is one possible explanation, but the marginal
distributions alone do not distinguish it from nonmonotone deterministic
responses. The bounds of Das, Toepoel, and van Soest
(\citeproc{ref-dasetal2011}{2011}), applied to threshold events as they
suggest {[}p.~36{]}, provide an alternative analysis that does not
impose a deterministic monotone map. We do not undertake this analysis
or a joint treatment of ordered categories here. The check of
Section~\ref{sec-ademp} enumerates every weakly increasing map on four
categories and confirms Equation~\ref{eq-mass} against a
linear-programming feasibility test.

\section{Multiwave Panels: What the Longitudinal Structure
Adds}\label{sec-multiwave-theory}

Theorem~\ref{thm-sharp} concerns two waves. In this section we allow a
cohort to be observed at many waves, with attrition between them, with a
battery of items asked once at entry, and with response patterns that
may be interrupted. We ask whether this longitudinal structure (the
joint distribution of reports across waves, the timing of dropout, the
negative controls) narrows the identified set of the conditioning map at
a given wave. We show that, within one cohort and without restrictions
that link waves, it does not: the identified set of the map at a wave
with a refreshment sample is the two-wave set of Theorem~\ref{thm-sharp}
computed from that wave's respondents and its refreshment distribution,
and the maps at waves without a refreshment sample are unidentified. The
result supplies the multiwave construction that the necessity part of
Theorem~\ref{thm-tworef} requires, and it shows what the constructive
designs of Section~\ref{sec-constructive} must assume, namely a
restriction linking the selection at one wave to something observed at
another.

\subsection{Setting}\label{sec-mw-setting}

A cohort enters at wave \(e\) and is followed for \(s\) waves; index
waves by tenure \(j = 1, \dots, s\). The latent path is
\(Y^{*} = (Y^{*}_1, \dots, Y^{*}_s) \sim F\) on \(\mathbb{R}^s\); the
entry coordinate may be vector-valued, carrying any battery of items
asked once at entry, and nothing below changes if it is. Conditioning
maps are tenure-specific, \(c_1 = \mathrm{id}\) and
\(c_j \in \mathcal{M}\) for \(j \ge 2\), and reports are
\(Y_j = c_j(Y^{*}_j)\); the case in which the map depends on the number
of completed interviews rather than on tenure is treated in
Remark~\ref{rem-dose}. A \emph{response pattern} is \(r \in \{0,1\}^s\)
with \(r_1 = 1\); monotone attrition is the case
\(r = (1, \dots, 1, 0, \dots, 0)\), but interrupted patterns are
allowed. Selection is a Markov kernel
\(\pi(r \mid y^{*}) = \Pr(R = r \mid Y^{*} = y^{*})\), unrestricted: the
probability of any pattern may depend on the whole latent path,
including coordinates that are never reported. Write
\(p_r := \Pr(R = r)\) and \(p_j := \sum_{r: r_j = 1} p_r\) for the
response rate at tenure \(j\).

The observables are, for each pattern \(r\), the sub-measure
\(\lambda_r\) on \(\mathbb{R}^{|r|}\) of the reported coordinates
\((Y_j)_{j: r_j = 1}\) among respondents with that pattern (of mass
\(p_r\)); and, at every tenure \(j\) in a set \(J^{*}\) at which a fresh
cohort is drawn, the population latent marginal \(F_j\) (refreshment
validity and cohort comparability, (A1) below). The entry marginal
\(F_1 = \sum_r \mathrm{proj}_1 \lambda_r\) is observed automatically.
Let \(Q_j := p_j^{-1}\sum_{r: r_j=1} \mathrm{proj}_j \lambda_r\) be the
reported-outcome distribution of all respondents at tenure \(j\),
whatever their pattern, and \(Q_{j,c'} := (c'^{-1})_{\#} Q_j\) its
pullback through a candidate map. If \(p_j = 0\), nobody responds at
\(j\); \(Q_j\) is then undefined, and we read \(p_j Q_{j,c'}\) as the
zero measure, so that the constraints below are void at \(j\) and
\(c_j\) is unrestricted. A structure is
\(\theta = (F, \pi, c_2, \dots, c_s)\); it rationalizes the observables
if it reproduces every \(\lambda_r\) and every \(F_j\), \(j \in J^{*}\).

\begin{lemma}[Pattern
completion]\protect\hypertarget{lem-seqcompletion}{}\label{lem-seqcompletion}

Let \(\{\tilde\lambda_r\}\) be finite measures, \(\tilde\lambda_r\) on
the coordinates \(\{j : r_j = 1\}\), indexed by patterns with
\(r_1 = 1\) and with masses summing to one. Let
\(J^{*} \subseteq \{1, \dots, s\}\) and, for \(j \in J^{*}\), let
\(F_j\) be a probability measure on \(\mathbb{R}\) with
\(\sum_{r: r_j = 1}\mathrm{proj}_j \tilde\lambda_r \le F_j\). Then there
exist a probability measure \(F'\) on \(\mathbb{R}^s\) and a Markov
kernel \(\pi'(r \mid y^{*})\) such that (i) \(\mathrm{proj}_j F' = F_j\)
for every \(j \in J^{*}\) and
\(\mathrm{proj}_1 F' = \sum_r \mathrm{proj}_1\tilde\lambda_r\); and (ii)
for every pattern \(r\), the sub-measure of the coordinates
\(\{j : r_j = 1\}\) among those with \(R = r\) under \((F', \pi')\) is
exactly \(\tilde\lambda_r\).

\end{lemma}

\begin{proof}
For \(j \notin J^{*}\) fix any probability measure \(\kappa_j\) on
\(\mathbb{R}\). For \(j \in J^{*}\) let \(p_j\) be the total mass of
\(\sum_{r: r_j = 1}\tilde\lambda_r\); if \(p_j < 1\) set
\(L_j := F_j - \sum_{r: r_j = 1}\mathrm{proj}_j\tilde\lambda_r\), a
nonnegative measure of mass \(1 - p_j\) by hypothesis, and
\(\kappa_j := L_j/(1 - p_j)\); if \(p_j = 1\) nobody is missing at \(j\)
and \(\kappa_j\) is not needed. For each pattern \(r\) define the
measure on \(\mathbb{R}^s\) \[
\nu_r \;:=\; \tilde\lambda_r \otimes \bigotimes_{j : r_j = 0} \kappa_j ,
\] which gives the reported coordinates the law \(\tilde\lambda_r\) and
the unreported ones independent draws from the \(\kappa_j\). Put
\(F' := \sum_r \nu_r\), a probability measure because the masses sum to
one, and \(\pi'(r \mid y^{*}) := (d\nu_r/dF')(y^{*})\), which exists and
lies in \([0,1]\) because \(\nu_r \le F'\), and which sums over \(r\) to
\(dF'/dF' = 1\). The pattern-\(r\) sub-measure is then
\(\pi'(r \mid \cdot)\,F' = \nu_r\), whose projection onto the reported
coordinates is \(\tilde\lambda_r\): this is (ii). For (i),
\(\mathrm{proj}_j F' = \sum_{r: r_j = 1}\mathrm{proj}_j\tilde\lambda_r + \sum_{r: r_j = 0} p_r\,\kappa_j = (F_j - L_j) + (1 - p_j)\,L_j/(1-p_j) = F_j\)
for \(j \in J^{*}\), since \(\sum_{r: r_j = 0} p_r = 1 - p_j\); at
\(j = 1\) every pattern reports, so
\(\mathrm{proj}_1 F' = \sum_r \mathrm{proj}_1\tilde\lambda_r\).
\end{proof}

Under monotone attrition the kernel takes the sequential form used in
Section~\ref{sec-constructive}: writing
\(N_j := \sum_{r: r_j = 1}\nu_r\) for the sub-measure of those still
responding at tenure \(j\), one has
\(N_1 = F' \ge N_2 \ge \dots \ge N_s\), and
\(\pi'_j := dN_{j+1}/dN_j \in [0,1]\) is the probability of continuing
past \(j\) given the latent path; the survivors through \(j\) then have
latent sub-measure \(\prod_{u < j}\pi'_u\,F' = N_j\). The completion is
the multiwave form of Lemma~\ref{lem-completion}: everyone's unreported
coordinates are assigned independently of their reported ones, with the
leftover distribution at each refreshment wave, and every observed
restriction is preserved.

\begin{theorem}[No additional identification from within-cohort
histories under unrestricted cross-wave
dependence]\protect\hypertarget{thm-multiwave}{}\label{thm-multiwave}

Under refreshment validity and cohort comparability at the tenures in
\(J^{*}\), the identified set for the family of maps
\((c_j)_{j \in J^{*}}\) is the product \[
\prod_{j \in J^{*}} \bigl\{\, c' \in \mathcal{M} : p_j\, Q_{j, c'} \le F_j \,\bigr\},
\] each factor the two-wave set of Theorem~\ref{thm-sharp} computed from
the tenure-\(j\) respondents' reported distribution, the response rate
at \(j\), and the refreshment distribution at \(j\). The maps at tenures
outside \(J^{*}\) are unidentified, every element of \(\mathcal{M}\)
being compatible with the data; and the joint distribution of reports
across waves, the response patterns, and any battery of entry-wave items
impose no further restriction. In particular, for a panel with one
refreshment sample at tenure \(s\), the identified set of \(c_s\) is
exactly \(\{c' : p_s Q_{s,c'} \le F_s\}\), whatever the panel's length,
its attrition history, and its entry battery.

\end{theorem}

\begin{proof}
\emph{Necessity.} Let \((F', \pi', c')\) rationalize the observables and
fix \(j \in J^{*}\). The latent sub-measure of the tenure-\(j\)
respondents is \(\sum_{r: r_j = 1}\pi'(r \mid \cdot)\,F' \le F'\),
because the kernel's values are nonnegative and sum to one over \(r\);
its \(j\)-th marginal is therefore dominated by
\(\mathrm{proj}_j F' = F_j\). By rationalization that marginal is also
the pullback of the observed reported distribution,
\((c_j'^{-1})_{\#}(p_j Q_j) = p_j Q_{j,c'_j}\). Hence
\(p_j Q_{j, c'_j} \le F_j\).

\emph{Sufficiency.} Take candidates \(c'_j\), \(j \in J^{*}\), each
satisfying its constraint, arbitrary \(c'_j \in \mathcal{M}\) at the
other tenures, and \(c'_1 = \mathrm{id}\). Pull every observed pattern
measure back to the latent scale coordinate by coordinate,
\(\tilde\lambda_r := (\times_{j : r_j = 1} c_j'^{-1})_{\#}\lambda_r\).
For \(j \in J^{*}\),
\(\sum_{r: r_j = 1}\mathrm{proj}_j\tilde\lambda_r = (c_j'^{-1})_{\#}(p_j Q_j) = p_j Q_{j,c'_j} \le F_j\),
which is the hypothesis of Lemma~\ref{lem-seqcompletion}. The lemma
returns \((F', \pi')\); pushing each \(\nu_r\)'s reported coordinates
forward through the \(c'_j\) recovers \(\lambda_r\), and
\(\mathrm{proj}_j F' = F_j\) on \(J^{*}\). So \((F', \pi', c')\)
rationalizes the observables. The product structure follows because the
constraints, and the construction, involve each \(j \in J^{*}\)
separately, and the unrestricted tenures never enter a constraint.
\end{proof}

The theorem extends the first remark after Theorem~\ref{thm-sharp}
(that, for identifying \(c\), the stayers' joint distribution with their
baseline reports is summarized by one distribution) to the entire panel.
For identifying the map at a wave with a refreshment sample, that wave's
reported distribution, its response rate and its refreshment
distribution exhaust the information in a cohort's reports over any
number of waves, in the pattern of its dropouts and in any battery of
entry items. The longitudinal structure does not identify \(c\), but it
makes possible restrictions that use it. The transport restriction (A2)
of Section~\ref{sec-tworef} links the target item's selection at \(t\)
to the negative controls' selection at \(t\); the restrictions B2--B4 of
Section~\ref{sec-constructive} link the survivors' latent mean at \(t\)
to something observed, the fresh cohort's future or the continuing
cohort's past. The construction in the proof shifts the survivors'
latent distribution at \(t\) while holding every entry-wave quantity and
every earlier wave fixed, and therefore violates each of these
restrictions; this is why each of them identifies. The fixed-effects
results of the companion identification paper are of this kind: the
additive cell-mean structure there is a restriction across cohorts and
waves, as the present theorem requires, and its Theorem 1 and Corollary
1 show how much of the path such a restriction recovers.

\begin{corollary}[Repeated
tenures]\protect\hypertarget{cor-repeated}{}\label{cor-repeated}

In the tenure model, where \(c_s\) is common to all cohorts, suppose
tenure \(s\) is observed at two refreshment episodes (cohorts
\(e_1 < e_2\) compared with fresh cohorts at \(t_r = e_r + s - 1\),
\(r = 1, 2\)). Then the identified set of \(c_s\) is the intersection of
the two episodes' sets,
\(\bigcap_{r=1,2}\{c' : p^{(r)} Q^{(r)}_{c'} \le F^{(r)}\}\), with
\(p^{(r)}\), \(Q^{(r)}\), \(F^{(r)}\) the response rate, reported
distribution, and refreshment distribution of episode \(r\). It contains
\(c_s\), can be strictly smaller than either set, and is not in general
a point. Under the additional hypothesis that the selection functional
is the same at the two episodes, \(\beta_1(s) = \beta_2(s)\) in the
notation of Section~\ref{sec-momentsystem}, within the location family
\(c_s(y) = y + h\), and with the moments of the two episodes expressed
on a common scale (raw units or a fixed reference standard deviation),
the data carry the testable restriction \(C_1(s) = C_2(s)\): if it
holds, the identified set of the shift is unchanged by the hypothesis,
and if it fails, the hypothesis is rejected. On episode-specific scales,
or for general maps, the hypothesis can exclude maps that satisfy both
episodes' domination constraints.

\end{corollary}

\begin{proof}
The two cohorts are independent samples with separate latent processes
and selection kernels, so Theorem~\ref{thm-multiwave} applies to each
with the common candidate \(c'_s\), which is admissible iff it is
admissible at both episodes; \(c_s\) satisfies both constraints. That
the intersection is not a point in general: under outcome-independent
attrition at both episodes with Laplace refreshment tails of rate
\(\lambda\), Corollary~\ref{cor-independent} gives intervals of
half-widths \(\lambda^{-1}\log(1/p^{(r)})\) about \(h\), whose
intersection is the smaller interval. That it can be strictly smaller
than either set: with Gaussian refreshment distributions and logistic
selection increasing in the outcome at one episode and decreasing at the
other, the example following Corollary~\ref{cor-funnel} gives the
one-sided sets \([h, h + a]\) and \([h - b, h]\) with \(a, b > 0\),
whose intersection is \(\{h\}\). For the location clause, a candidate
\(h'\) implies survivor latent means
\(\mathbb{E}[Y^{(r)}_t \mid S] - h'\) at both episodes; on a common
scale, \(\beta_1(s) = \beta_2(s)\) reads \(C_1(s) = C_2(s)\) after
\(h'\) cancels, a restriction on the data alone. If instead the
episode-\(r\) moments are divided by different standard deviations
\(\sigma_r\), the restriction reads
\((C_1(s) - h')/\sigma_1 = (C_2(s) - h')/\sigma_2\) in raw units, which
determines \(h'\) when \(\sigma_1 \ne \sigma_2\). For general maps the
candidate does not cancel: with \(F^{(1)} = Q^{(1)} = N(0,1)\),
\(F^{(2)} = Q^{(2)} = N(1,1)\), retention \(.1\) at both episodes and
the true identity (selection differentials zero at both), the candidate
\(c'(y) = 2y\) has pullbacks \(N(\mu_r/2, 1/4)\), whose largest density
ratios against \(F^{(r)}\) are \(2\exp(\mu_r^2/6)\) with
\(\mu_r \in \{0, 1\}\), so it satisfies both constraints; but it implies
survivor latent means \(\mu_r/2\), hence selection differentials \(0\)
and \(-1/2\), and violates \(\beta_1(s) = \beta_2(s)\).
\end{proof}

The corollary is the formal basis of the recommendation in
Section~\ref{sec-design} that panels repeat a tenure across refreshment
episodes. Repeating a tenure intersects the two episodes' identified
sets and, in the location family, adds a direct test of stationary
selection, but it does not suffice for point identification, which
requires a restriction of the kind that Theorem~\ref{thm-tworef} and
Section~\ref{sec-constructive} supply.

\begin{refremark}[Interrupted participation and dose-specific maps]
If the map at a wave depends on the number \(k\) of interviews completed
before it and not on tenure (a dose model, as in interrupted rotation
designs), then respondents at tenure \(j\) with different patterns carry
different maps \(c_{j,k}\), and the constraint at \(j \in J^{*}\)
becomes \(\sum_k (c_{j,k}'^{-1})_{\#}(p_{j,k}\,Q_{j,k}) \le F_j\), where
\(p_{j,k}\) and \(Q_{j,k}\) are the response rate and reported
distribution of the tenure-\(j\) respondents with dose \(k\): a joint
constraint on the family \(\{c_{j,k}\}_k\) requiring the sum of the
dose-specific pullbacks to fit under \(F_j\). The identified set for the
family is then not a product over \(k\), but the proof is unchanged,
since Lemma~\ref{lem-seqcompletion} is stated for arbitrary patterns.
The companion identification paper's assumption of uninterrupted
participation (its M1) is the case in which every respondent at a tenure
has the same dose, and the present theorem shows why interrupted
patterns do not affect the identification question addressed here: the
construction re-encodes them and does not use them.

\label{rem-dose}

\end{refremark}

A deterministic check of the construction is reported in
Section~\ref{sec-ademp}. On a four-wave cohort with a finite latent
lattice, a binary entry-only negative control, all eight response
patterns, and a selection kernel that depends on the never-reported
coordinates, the completion of Lemma~\ref{lem-seqcompletion} is built
explicitly for every candidate map at the refreshment wave. The
candidates that satisfy the domination inequality reproduce every
pattern sub-measure, the entry marginal, and the refreshment marginal
within the script's tolerance of \(10^{-10}\) while moving the
survivors' latent mean by the candidate's perturbation; those that
violate it have a negative leftover.

\section{Sharp Bounds for Location-Shift Conditioning}\label{sec-bounds}

Before restricting the map, we record the bounds that
Theorem~\ref{thm-sharp} implies for the survivors' mean conditioning
effect when the map is unrestricted. These bounds are a known object,
and they are the benchmark against which the shape restriction should be
judged.

\begin{corollary}[Mean bounds for an unrestricted
map]\protect\hypertarget{cor-meanbounds}{}\label{cor-meanbounds}

Under the conditions of Theorem~\ref{thm-sharp}, suppose that \(p < 1\),
that \(F_2\) and \(Q\) are atomless with support \(\mathbb{R}\), and
that both have finite means. Let
\(m^{S} := \mathbb{E}[Y_2^{*} \mid S = 1]\) be the survivors' latent
mean and \(h^{S} := \mathbb{E}[Y_2 - Y_2^{*} \mid S = 1]\) their mean
conditioning effect. Write \(q_u := F_2^{-1}(u)\) and \[
m_L(p) := \mathbb{E}_{F_2}[Y \mid Y \le q_p], \qquad m_U(p) := \mathbb{E}_{F_2}[Y \mid Y \ge q_{1-p}]
\] for the means of the lower and upper \(p\)-fractions of the
refreshment distribution. Then the identified set for \(m^{S}\) is the
open interval \((m_L(p), m_U(p))\) and that for \(h^{S}\) is
\((\,\mathbb{E}[Y_2 \mid S] - m_U(p),\ \mathbb{E}[Y_2 \mid S] - m_L(p)\,)\):
every interior point is attained by a map in \(\mathcal{M}\), and the
endpoints, which are the sharp infimum and supremum, are not. The
closures are the trimmed-means intervals. The width
\(w(p) := m_U(p) - m_L(p)\) depends on the refreshment distribution and
the retention rate only. The same argument bounds the survivors' latent
distribution function pointwise:
\(\max\{0, (F_2(y) - 1 + p)/p\} \le \Pr(Y_2^{*} \le y \mid S = 1) \le \min\{1, F_2(y)/p\}\).

\end{corollary}

\begin{proof}
By Theorem~\ref{thm-sharp} the survivors' latent distribution
\(\nu := Q_{c'}\) ranges over the probability measures with
\(p\,\nu \le F_2\) that are pullbacks of \(Q\) through some
\(c' \in \mathcal{M}\). Any \(\nu\) with \(p\,\nu \le F_2\) has
\(\int y\,d\nu \le m_U(p)\): the density \(d\nu/dF_2\) is bounded by
\(1/p\) and integrates to one, and among such densities the mean is
maximized by the one equal to \(1/p\) on \([q_{1-p}, \infty)\) and zero
below, since moving mass from below \(q_{1-p}\) to above it, at density
at most \(1/p\), raises the mean; because \(F_2\) is atomless this
maximizer \(\nu^{*} := p^{-1}F_2|_{[q_{1-p},\infty)}\) is unique, and
symmetrically \(\int y\,d\nu \ge m_L(p)\) with equality only at the
lower analogue. Neither extremal measure is admissible: its support is a
half-line, whereas a pullback of \(Q\) through an increasing
homeomorphism of \(\mathbb{R}\) has support \(\mathbb{R}\) because \(Q\)
does. Conversely, with \(\eta \in (0, 1]\), the mixtures
\(\nu_\eta := (1-\eta)\,\nu^{*} + \eta F_2\) satisfy
\(p\,\nu_\eta \le (1-\eta)F_2 + \eta\,p F_2 \le F_2\), are atomless with
support \(\mathbb{R}\), and have means filling
\([\,\mathbb{E}_{F_2}Y,\ m_U(p))\); the lower analogue fills the rest of
the open interval. For each such \(\nu\) the increasing rearrangement
\(c' := Q^{-1}\circ F_\nu\) is a continuous, strictly increasing map of
\(\mathbb{R}\) onto \(\mathbb{R}\), because both distribution functions
are continuous and strictly increasing, so \(c' \in \mathcal{M}\) and
\((c'^{-1})_{\#}Q = \nu\); thus \(c' \in \mathcal{C}_{ID}\) by
Theorem~\ref{thm-sharp}, and
\(h^{S} = \mathbb{E}[Y_2 \mid S] - \int y\,d\nu\) because
\(\mathbb{E}[Y_2 \mid S]\) is observed and finite. The
distribution-function bounds are \(\nu(-\infty, y] \le F_2(y)/p\) and
\(\nu(y, \infty) \le (1 - F_2(y))/p\).
\end{proof}

The closure of the interval is the identification region of Horowitz and
Manski (\citeproc{ref-horowitzmanski1995}{1995}) {[}Proposition 1A,
p.~285; Corollary 4.1, eq. (14), p.~291{]} for the mean of one component
of a mixture with a known mixing proportion. Here the refreshment
distribution is the mixture, the survivors are the component of known
proportion \(p\), and the attriters' latent distribution is the free
component. Its width is also the width of the trimming bounds of Lee
(\citeproc{ref-lee2009}{2009}) at trimming proportion \(1 - p\); Lee
(\citeproc{ref-lee2009}{2009}, Appendix, Lemma 1, pp.~1097--1098) states
the mixture result and attributes its proof to the corollary of Horowitz
and Manski. The trimmed-mean interval is therefore an established
result. Corollary~\ref{cor-meanbounds} adds the statement for the map
class \(\mathcal{M}\), that interior points are attained and endpoints
are not; the interval then serves as the benchmark against which the
location restriction is judged. For a Gaussian \(F_2\) with standard
deviation \(\sigma\) the width is
\(w(p) = 2\sigma\,\phi(\Phi^{-1}(p))/p\): \(.70\sigma\) at \(p = .8\)
and \(1.29\sigma\) at \(p = .6\). A location shift \(h'\) places the
survivors' latent mean at \(\mathbb{E}[Y_2 \mid S] - h'\), so the
location set of the next theorem is contained in the closure of this
interval. Under the logistic-selection example of
Section~\ref{sec-corollaries} it has width at most \(\beta\sigma^2\)
(\(.3\sigma\) for \(\beta = .3\), \(\sigma = 1\)), and under
outcome-independent attrition with tails thinner than exponential it is
a point. In the design of Simulation 1 the location sets have widths
between \(0\) and \(.26\sigma\) at \(p = .8\), against \(.70\sigma\) for
an unrestricted map: most of the narrowing comes from the shape
restriction and not from the retention rate.

\begin{refremark}[The mean bounds do not depend on the conditioning
model]
The necessity half of Corollary~\ref{cor-meanbounds} uses only that the
survivors are a sub-population of known proportion \(p\) of a population
whose latent distribution the refreshment sample reveals; it never uses
that conditioning is a deterministic increasing map. The interval
\([m_L(p), m_U(p)]\) for the survivors' latent mean, and hence the
corresponding interval for \(h^{S}\), therefore bounds these quantities
under any conditioning mechanism (stochastic, nonmonotone,
item-specific), and in a model rich enough to carry an arbitrary
sub-population's latent distribution to the observed reported one, as
stochastic models are, the bounds are sharp and can be attained. With
atoms the extremal sub-populations may split an atom, and the endpoints
are the values of the linear programme \(\max\) (or \(\min\))
\(\int y\,d\nu\) over \(0 \le \nu \le F_2/p\) with
\(\nu(\mathbb{R}) = 1\). For a binary item with population mean
\(\pi_2\) this gives
\(m^{S} \in [\max\{0, (\pi_2 - 1 + p)/p\},\ \min\{1, \pi_2/p\}]\), an
interval of width \((1 - p)/p\) whenever \(1 - p \le \pi_2 \le p\). The
principal decomposition of Das, Toepoel, and van Soest
(\citeproc{ref-dasetal2011}{2011, 40--41}) concerns a population
conditioning effect, whose unobserved part is the attriters' conditioned
answers and whose sharp bounds have width \(1 - p\). Their online
Appendix 1 {[}p.~26{]} also gives the sharp bounds on the binary
survivor effect, whose unobserved part is the survivors' latent mean:
they are the observed survivor mean minus the interval for \(m^{S}\)
above. Our result for continuous outcomes relates these bounds to the
homeomorphism model and states when the endpoints are attained. In the
stated interior-probability range the survivor interval is wider by the
factor \(1/p\); near the boundaries its width is truncated and no
uniform ordering holds (at \(p = .8\) and \(\pi_2 = .01\) it is
\(.0125\), against \(.2\) for the population interval). The location set
of Theorem~\ref{thm-bounds} and the mass criterion of
Corollary~\ref{cor-discrete}, by contrast, do use the deterministic
model, and their additional power derives from it.

\label{rem-modelfree}

\end{refremark}

When \(\mathcal{M}\) is restricted to location shifts \(c(y) = y + h\),
the identified set becomes an interval and Theorem~\ref{thm-sharp}
yields sharp bounds on \(h\).

\begin{theorem}[Sharp
bounds]\protect\hypertarget{thm-bounds}{}\label{thm-bounds}

Under the conditions of Theorem~\ref{thm-sharp} with \(c(y) = y + h\),
and assuming the refreshment density \(f_2\) is unimodal with an upper
semicontinuous version (for example, continuous), \[
\mathcal{H}_{ID} \;=\; \Bigl\{\, h' : \; p\, q(y + h') \le f_2(y) \;\; \text{a.e.}\Bigr\} \;=\; \bigl[\,h_L,\; h_U\,\bigr] \;\ni\; h ,
\] where \(q\) is the density of \(Q\): the identified set is a closed
interval containing \(h\), and it is sharp, every
\(h' \in \mathcal{H}_{ID}\) being attained by the explicit completion of
Lemma~\ref{lem-completion}. Writing
\(\Psi(h') := \operatorname{ess\,sup}_y p\,q(y+h')/f_2(y)\), the
endpoints are the infimum and supremum of \(\{h' : \Psi(h') \le 1\}\).
The interval is nondegenerate if and only if \(\Psi(h') \le 1\) for some
\(h' \ne h\); Corollary~\ref{cor-funnel}(iii) gives a sufficient
condition for a two-sided neighbourhood of \(h\). It may be a single
point at interior retention: with \(f_2 = N(0,1)\), \(q = N(-1,1)\) and
\(p = .5\) it is \(\{-1\}\), and with \(f_2 = N(0,1)\),
\(q = N(0, \sigma^2)\) and \(p = \sigma < 1\) it is \(\{0\}\) although
the stayers' latent density has strictly lighter tails than \(f_2\) on
both sides (the example following Corollary~\ref{cor-funnel}). Lighter
tails of the stayers' latent density are therefore not sufficient for a
nondegenerate interval. The interval is contained in the closure of the
unrestricted-map interval of Corollary~\ref{cor-meanbounds}.

\end{theorem}

\begin{proof}
From Theorem~\ref{thm-sharp}, the location family makes \(Q_{c'}\) a
shift of \(Q\) and the domination constraint becomes the displayed
density inequality. Rewriting it with \(z = y + h'\), the constraint is
equivalent to \[
f_2(z - h') \;\ge\; p\, q(z) \quad \text{for a.e. } z .
\] \emph{Interval.} If \(h'_1 < h'_2\) are feasible and
\(h'_1 < h' < h'_2\), then for a.e. \(z\) both
\(f_2(z - h'_1) \ge p\,q(z)\) and \(f_2(z - h'_2) \ge p\,q(z)\) hold;
\(z - h'\) lies between \(z - h'_2\) and \(z - h'_1\), and a unimodal
density at a point between two others is at least the smaller of its
values there, so \(f_2(z - h') \ge p\,q(z)\). \emph{Closed.} If feasible
\(h'_n \to h'\), then outside the countable union of the exceptional
null sets \(f_2(z - h'_n) \ge p\,q(z)\) for every \(n\), and upper
semicontinuity gives
\(f_2(z - h') \ge \limsup_n f_2(z - h'_n) \ge p\,q(z)\). The set is
nonempty because it contains the true \(h\). (Log-concavity of \(f_2\)
is sufficient a fortiori; no condition on \(q\) is needed.) Sharpness of
every point in it is Lemma~\ref{lem-completion}. The characterization of
the endpoints and of nondegeneracy restates the definition, and the
Gaussian examples are computed in the example following
Corollary~\ref{cor-funnel} and in Section~\ref{sec-ademp}. Containment
follows because a shift \(h'\) fixes the survivors' latent mean at
\(\mathbb{E}[Y_2 \mid S] - h'\), which must lie in \([m_L(p), m_U(p)]\)
by the necessity half of Corollary~\ref{cor-meanbounds}.
\end{proof}

\begin{remark}[Relation to Lee bounds]
Lee (\citeproc{ref-lee2009}{2009}) bounds the average effect of a
training programme on wages for those whose wages would be observed
under either treatment assignment, by trimming the observed outcome
distribution of the treated group from either tail by the proportion
\(p_0 = [\Pr(S=1 \mid D=1) - \Pr(S=1 \mid D=0)]/\Pr(S=1 \mid D=1)\),
under random assignment and the monotonicity restriction \(S_1 \ge S_0\)
(\citeproc{ref-lee2009}{Lee 2009}, Proposition 1a, p.~1083). In his
application the observation indicator is employment; the method is
stated for sample selection in general, and his conclusion notes its
application to survey nonresponse and attrition
(\citeproc{ref-lee2009}{Lee 2009, 1097}). Corollary~\ref{cor-meanbounds}
trims the refreshment distribution at the retention rate to bound the
survivors' latent mean. The two results share the same geometry: a
sub-measure of known mass must fit under a fixed measure, and the
extremal fits lie in the tails. They differ in the estimand and in the
assumptions that produce the mixture. Here the unidentified component is
the survivors' latent-outcome distribution, whose mixture with the
attriters' latent distribution is the refreshment distribution. In Lee
(\citeproc{ref-lee2009}{2009}) the unidentified component is the outcome
distribution of treated units who would be observed under either
assignment, and the mixture follows from independence and monotonicity.
The containment statement of Theorem~\ref{thm-bounds} shows what the
location restriction adds to the unrestricted-map interval.
\end{remark}

\subsection{Shrinking the funnel with covariates}\label{sec-shrink}

Covariates observed for the whole cohort and in the refreshment sample
shrink the identified set through stratification, and the
covariate-standardized designs of Section~\ref{sec-constructive} exploit
this.

\begin{proposition}[Covariate
stratification]\protect\hypertarget{prp-covariates}{}\label{prp-covariates}

Let \(X\) be a discrete entry covariate observed for the whole cohort
and in the refreshment sample, let structures carry \(X\) (so that \(F\)
is a law on \((X, Y_1, Y_2^{*})\) and \(\pi\) may depend on \(X\)), and
suppose the conditioning map does not vary with \(X\). Write \(p(x)\),
\(Q(\cdot \mid x)\), and \(F_2(\cdot \mid x)\) for the retention rate,
the stayers' reported-outcome distribution, and the refreshment
distribution within \(X = x\). Then the identified set is \[
\mathcal{C}_{ID}(X) \;=\; \bigcap_{x}\bigl\{\, c' \in \mathcal{M} : p(x)\,Q_{c'}(\cdot \mid x) \le F_2(\cdot \mid x) \,\bigr\} \;\subseteq\; \mathcal{C}_{ID},
\] it is sharp, and in the location family
\(\mathcal{H}_{ID}(X) = \bigcap_x \mathcal{H}_{ID}(x)\). The inclusion
is strict if and only if some map satisfying the pooled constraint
violates the constraint of some stratum; this can happen when the strata
differ only in their retention rates. If the map is allowed to vary with
\(X\), the identified set is the product of the stratum sets.

\end{proposition}

\begin{proof}
Necessity within each stratum is Theorem~\ref{thm-sharp} applied
conditionally on \(X = x\): the stratum's observables are
\((F_1(\cdot \mid x), F_2(\cdot \mid x), G(\cdot \mid x), p(x))\), and a
structure that rationalizes the pooled data rationalizes each stratum's.
Sufficiency: for \(c'\) in the intersection, complete each stratum by
Lemma~\ref{lem-completion} and glue,
\(F' := \sum_x \Pr(X = x)\,F'(\cdot \mid x)\) and
\(\pi'(y_1, y_2, x) := \pi'_x(y_1, y_2)\). The inclusion: the stayers'
sub-measure is the mixture
\(p\,Q = \sum_x \Pr(X = x)\,p(x)\,Q(\cdot \mid x)\) and the pullback
commutes with mixing, so summing the stratum inequalities with weights
\(\Pr(X = x)\) gives
\(p\,Q_{c'} \le \sum_x \Pr(X = x)\,F_2(\cdot \mid x) = F_2\). The
strictness statement restates the definition of the two sets, and the
retention-only case is the example below. The product statement is
immediate.
\end{proof}

Whether the covariate identifies the map depends on an assumption that
the data cannot check. If selection is known to be independent of the
latent outcome within strata (missing at random given \(X\)), then the
stayers' latent distribution within each stratum equals
\(F_2(\cdot \mid x)\), and Corollary~\ref{cor-noattrition} applied
within a stratum identifies the map by quantile alignment on the
stratum's latent-outcome support (the map being common, on the union of
these supports). The data cannot reveal that this is so: the observables
generated under selection on \(X\) alone are also generated by
outcome-dependent selection within strata, and without imposing the
assumption the covariate shrinks the identified set only through the
intersection. A covariate can tighten the set even when it predicts
retention and nothing else. Take two equally prevalent strata with
identical Laplace distributions of rate one,
\(F_2(\cdot \mid x) = Q(\cdot \mid x)\), the true map the identity, and
retention \(.8\) and \(.2\). The pooled retention rate is \(.5\) and the
pooled location set is \([-\log 2, \log 2]\)
(Corollary~\ref{cor-independent}(ii)); the stratum sets are
\([-\log 1.25, \log 1.25]\) and \([-\log 5, \log 5]\), whose
intersection \([-\log 1.25, \log 1.25]\) is strictly smaller, because
shifts with \(\log 1.25 < |\varepsilon| \le \log 2\) satisfy the pooled
constraint but violate the high-retention stratum's.

\section{Two Refreshments: Point Identification of the
Path}\label{sec-tworef}

A single refreshment sample bounds a single contrast. A panel with
several refreshment samples observes several tenure configurations in
common periods, and whether a contrast is identified then depends on the
rank of a linear system, which we now set out.

\subsection{The moment system}\label{sec-momentsystem}

Entry cohorts \(e_1 < e_2 < e_3\) (JLPS: 2007, 2011, 2019), tenure
\(s = t - e + 1\). A \emph{refreshment episode} \(r\) is a wave at which
a new cohort enters alongside surviving older cohorts: episode 1
(\(t = 2011\)) generates tenures \(\{5, 1\}\), episode 2 (\(t = 2019\))
generates \(\{13, 9, 1\}\). Outcomes are standardized item by item, each
on one scale fixed across episodes and cohorts (for example, divided by
a single reference standard deviation); for a target item, write the
survivor mean at episode \(r\) and tenure \(s\) as \[
m_r(s) \;=\; \mu_r \;+\; g(e) \;+\; \tau(s) \;+\; \beta_r(s),
\] where \(\mu_r\) is the period effect, \(g(e)\) the entry-cohort
effect, \(\tau(s)\) the conditioning path with the entry normalization
\(\tau(1) = 0\), and \(\beta_r(s)\) the \emph{selection functional}: the
standardized difference between the surviving and full entry cohort at
tenure \(s\), with \(\beta_r(1) = 0\) because a refreshment cohort at
entry is an unselected draw. Within an episode, period effects cancel in
contrasts and cohort is a function of tenure, so the estimable moments
are the \emph{entry contrasts} \[
C_r(s) \;=\; m_r(s) - m_r(1) \;=\; \Delta g_r(s) \;+\; \tau(s) \;+\; \beta_r(s),
\] one per experienced tenure per episode: \(C_1(5)\), \(C_2(9)\),
\(C_2(13)\). Negative-control items (time-invariant facts reported once
at entry, so that \(\tau^{N} \equiv 0\) by construction, and reported
before any panel selection, so that their survivor means at \(t\) equal
the entry values of survivors) yield the parallel moments \[
N_r(s) \;=\; \Delta g^{N}_r(s) \;+\; \beta^{N}_r(s).
\]

Two transport assumptions link the systems, both adapted from the
recovery conditions of the companion identification paper (its
sampling-equivalence assumption M2 and the transport condition (iv) of
its Proposition 2, which concerns cohort effects and not selection), and
both partially testable here:

\textbf{(A1) Cohort overlap and comparability.} Each refreshment targets
the same birth cohorts as the incumbents on a common support (in the
Japanese panel, the age-matched refreshment of 2011), and on that
support entrants and incumbents are drawn from the same population up to
sampling, so that \(\Delta g_r(s) = \Delta g^{N}_r(s) = 0\) for durable
items. Overlap of birth-year support is necessary for the second clause
but does not imply it: comparability is a separate assumption about
sampling frames, nonresponse at entry, mode, and the coding of items.
(Off support, \(\Delta g\) is bounded rather than zero and the
conclusions become interval-valued; see
Corollary~\ref{cor-generational}.) \emph{Information set.} Comparability
can be asserted for the marginal laws of the target and of each negative
control separately, or jointly for the target and a discrete vector
\(X\) of time-invariant entry variables observed both for the incumbent
cohorts at their entry and in the refreshment sample (the
negative-control battery, possibly with other discrete entry
covariates). We call these the \emph{marginal} and the \emph{joint}
versions of (A1). Under the joint version every cohort and the
refreshment sample share one population law of \((X, Y^{*})\) on the
common support. The entry distribution of \(X\) is then the same in all
of them, which is an observable implication of (A1), and the refreshment
sample reveals the conditional law \(F^{(r)}(\cdot \mid x)\) of the
latent target given \(X = x\) and not only its margin. The incumbents'
reports and response rates are observed within each stratum of \(X\), so
Proposition~\ref{prp-covariates} applies at every episode. The two
versions give the same result when the loading of (A2) is known and can
give different identified sets when it is only bounded
(Theorem~\ref{thm-tworef}(a)).

\textbf{(A2) Selection transport.} For target item \(j\),
\(\beta^{j}_r(s) = \Gamma_j\, \beta^{N}_r(s)\) with loading \(\Gamma_j\)
known or bounded: the item-level selection differential is the common
survivor-composition shift measured on the negative controls, scaled by
the item's loading on it. (\(\Gamma_j = 1\) for all standardized items
is a leading case; a bounded \(\Gamma_j \in [\Gamma_L, \Gamma_U]\) gives
intervals.) (A2) restricts the joint selection of the target item and
the battery; it is not implied by (A1), and a small negative-control
contrast is not sufficient evidence that the target item's selection
bias is small.

For episode \(r\), experienced tenure \(s\) and a value \(x\) of \(X\)
with positive probability, let \(p^{(r)}(s \mid x)\) and
\(Q^{(r)}_s(\cdot \mid x)\) be the response rate and the reported
distribution of the cohort's respondents at tenure \(s\) within
\(X = x\). These are all respondents at that tenure, whatever their
response pattern, as in Theorem~\ref{thm-multiwave}; under monotone
attrition they are the survivors of the moment system. Let
\(\mathcal{H}_r(s \mid x)\) be the location set of
Theorem~\ref{thm-bounds} computed from \(p^{(r)}(s \mid x)\),
\(Q^{(r)}_s(\cdot \mid x)\) and \(F^{(r)}(\cdot \mid x)\) and expressed
in the units of the moment system. A stratum with
\(p^{(r)}(s \mid x) = 0\) imposes no constraint. Write
\(\mathcal{H}_r(s \mid X) := \bigcap_x \mathcal{H}_r(s \mid x)\). Under
the marginal version of (A1), \(X\) is empty and
\(\mathcal{H}_r(s \mid X) = \mathcal{H}_r(s)\), the set computed from
the pooled response rate, reported distribution and refreshment
distribution.

\subsection{The theorem}\label{sec-thm3}

\begin{theorem}[Point identification with two refreshments: the role of
the negative-control
battery]\protect\hypertarget{thm-tworef}{}\label{thm-tworef}

Under (A1)--(A2), in a design with three entry cohorts and two
age-matched refreshment episodes that generate experienced-cohort
contrasts at the tenure pairs \((r, s) \in \{(1,5), (2,9), (2,13)\}\)
(the episode structure of the Japanese panel, whose second refreshment
is however generational; see Corollary~\ref{cor-generational} for that
case), and, for the set statements in (a) and the attainability
statement in (b), with conditioning at each tenure \(s\) given by the
same location map \(c_s(y) = y + \tau(s)\) for every incumbent, whatever
the value of \(X\), on the fixed common scale of the target:

\begin{enumerate}
\def\labelenumi{(\alph{enumi})}
\item
  \emph{(Sufficiency.)} If every episode that contributes an
  experienced-cohort contrast carries a negative-control battery and the
  loading \(\Gamma\) is known, then \(\tau(s)\) is point-identified at
  every tenure \(s\) at which an experienced cohort is observed at a
  refreshment episode: \[
  \tau(s) \;=\; C_r(s) \;-\; \Gamma\, N_r(s), \qquad (r, s) \in \{(1,5), (2,9), (2,13)\}.
  \] No stationarity of selection across episodes is required: each
  episode's selection functional is measured within that episode. If the
  loading is known only to lie in an interval,
  \(\Gamma \in [\Gamma_L, \Gamma_U]\), then, because the loading is
  common to all pairs \((r, s)\), the pairs restrict it together: the
  identified set for the loading is \[
  \Gamma^{*} \;:=\; \bigl\{\, \Gamma \in [\Gamma_L, \Gamma_U] :\; C_r(s) - \Gamma N_r(s) \in \mathcal{H}_r(s \mid X) \text{ for every observed } (r, s) \,\bigr\},
  \] with \(X\) the entry variables over which (A1) is asserted and
  \(\mathcal{H}_r(s \mid X)\) the location set defined in
  Section~\ref{sec-momentsystem} (if a tenure \(s\) were observed at two
  episodes, \(\Gamma\) would also have to make
  \(C_r(s) - \Gamma N_r(s)\) equal across them, and the relevant set
  would be the intersection of Corollary~\ref{cor-repeated}). The
  identified set for the vector \((\tau(s))_s\) is
  \(\{(C_r(s) - \Gamma N_r(s))_{(r,s)} : \Gamma \in \Gamma^{*}\}\).
  Because \(\mathcal{H}_r(s \mid X) \subseteq \mathcal{H}_r(s)\)
  (Proposition~\ref{prp-covariates}), the set computed from the margins
  alone is sharp under the marginal version and only an outer bound
  under the joint version, where it can be strictly larger. Intersecting
  each pair's interval
  \(\{C_r(s) - \Gamma N_r(s) : \Gamma \in [\Gamma_L, \Gamma_U]\}\), of
  width \((\Gamma_U - \Gamma_L)\,|N_r(s)|\), with its own
  \(\mathcal{H}_r(s \mid X)\) gives a further outer bound, which can
  also be strictly larger. The set is a point when \(N_r(s) = 0\) at
  every pair or when the loading is known. The subtraction identity uses
  only the moment decomposition; the set statements also use the
  location map.
\item
  \emph{(Necessity.)} Without a negative-control battery at episode
  \(r\) (and without stationarity), \(\tau(s)\) and \(\beta_r(s)\) enter
  the observable moments of that episode only through their sum: for any
  \(\epsilon\), setting \(\tilde\tau(s) = \tau(s) + \epsilon\) and
  \(\tilde\beta_r(s) = \beta_r(s) - \epsilon\) reproduces all of them.
  The perturbation is attainable by an admissible selection process
  whenever the corresponding location shift \(h + \epsilon\) lies in the
  location set \(\mathcal{H}_r(s \mid X)\) at every episode at which
  tenure \(s\) is observed and none of those episodes carries a battery
  used under (A2); here \(X\) collects the entry variables over which
  (A1) is asserted, and is empty under its marginal version. If those
  sets contain \(h + \epsilon_0\) for some \(\epsilon_0 \ne 0\) and are
  intervals (as when every refreshment density \(f^{(r)}(\cdot \mid x)\)
  is unimodal), every \(\epsilon\) between \(0\) and \(\epsilon_0\) is
  attainable; a one-sided set permits perturbations of one sign only.
  For every attainable \(\epsilon \ne 0\), \(\tau(s)\) is not
  point-identified, and the non-identification is not a moments-only
  artefact: by Theorem~\ref{thm-multiwave} the perturbed structure can
  be chosen to reproduce the entire observed distribution of every
  cohort involved (every response pattern's joint distribution of
  reports across all waves, the entry-wave battery, and the retention
  pattern) together with the refreshment distributions, so that (A1) is
  preserved in the version asserted and no functional of the data
  distinguishes \(\tau(s)\) from \(\tau(s) + \epsilon\). The condition
  does not hold automatically. Under attrition independent of the latent
  outcome and a refreshment density with tails thinner than exponential,
  the location set is \(\{h\}\) (Corollary~\ref{cor-independent}(i));
  under outcome-dependent attrition it can also be \(\{h\}\) (the
  example following Corollary~\ref{cor-funnel}). Necessity therefore
  requires a nondegenerate location set, not merely outcome-dependent
  attrition. Cross-episode stationarity (\(\beta_1 = \beta_2 = \beta\))
  does not substitute when the episodes' tenure sets are disjoint (as in
  the Japanese panel, \(\{5\} \cap \{9, 13\} = \emptyset\)), because the
  common function \(\beta(\cdot)\) is never evaluated twice at the same
  argument, so the constraint binds nothing; when a tenure is repeated
  it adds, in the location family, a testable restriction without
  shrinking the identified set (Corollary~\ref{cor-repeated}).
  Restrictions that link cohorts or waves beyond (A1)--(A2), such as an
  additive cell-mean structure across cohorts, must be checked against
  the perturbation separately.
\item
  \emph{(Overidentification.)} With a \(K\)-item negative-control
  battery whose loadings are known, each pair \((r, s)\) yields \(K\)
  estimates of the single factor \(\beta^{N}_r(s)\), giving \(K - 1\)
  testable restrictions per pair (\(3(K-1)\) in the configuration of
  (a); \(K = 23\) in the Japanese panel), the battery-dispersion placebo
  test. Under the additional hypothesis of stationarity in tenure
  (\(\beta_r(s) = \beta(s)\) for a common \(\beta\)), no further
  restriction is testable in that tenure configuration (by the
  disjointness in (b)); a design with a repeated tenure across episodes
  (for example, refreshments eight waves apart in a panel observed at
  tenure 9 in both) adds one direct stationarity test per repeated
  tenure.
\end{enumerate}

\end{theorem}

\begin{proof}
\leavevmode

\begin{enumerate}
\def\labelenumi{(\alph{enumi})}
\item
  Substitute (A1) into the moment equations:
  \(C_r(s) = \tau(s) + \Gamma \beta^{N}_r(s)\) and
  \(N_r(s) = \beta^{N}_r(s)\). The system, stacked over
  \((r,s) \in \{(1,5),(2,9),(2,13)\}\), is triangular: each \(N_r(s)\)
  identifies its \(\beta^{N}_r(s)\), and each \(C_r(s)\) then identifies
  its \(\tau(s)\) by subtraction; the coefficient matrix is the identity
  on \((\beta^{N}_1(5), \beta^{N}_2(9), \beta^{N}_2(13))\) stacked with
  a unit-diagonal block on \((\tau(5), \tau(9), \tau(13))\). With
  \(\Gamma\) unknown, every \(\Gamma\) in \(\Gamma^{*}\) is attained:
  fix it and set \(\tilde\tau(s) = C_r(s) - \Gamma N_r(s)\), which lies
  in \(\mathcal{H}_r(s \mid x)\) for every \(x\) by definition of
  \(\Gamma^{*}\). For each cohort and each stratum \(X = x\),
  Theorem~\ref{thm-multiwave} applied within the stratum, with \(J^{*}\)
  the tenures at which that cohort meets a refreshment sample, the
  refreshment distributions \(F^{(r)}(\cdot \mid x)\), the entry
  coordinate carrying the battery, and at each such tenure the location
  map shifted to \(\tilde\tau(s)\) (the same map in every stratum),
  gives a structure that reproduces the cohort's observables within the
  stratum. Mixing the strata with the cohort's observed entry
  distribution of \(X\) (complete conditionally, then mix, as in the
  proof of Proposition~\ref{prp-covariates}) reproduces all of the
  cohort's observables, including their joint distribution with \(X\).
  Under the joint version of (A1) that entry distribution is the
  refreshment sample's, so the mixed population law of \((X, Y^{*})\) at
  each episode is the refreshment's joint law: (A1) holds, the
  incumbents' population mean of the target equals the entrants' mean
  \(m_r(1)\), and the survivors' latent mean at each such tenure is
  \(\mathbb{E}[Y \mid S] - \tilde\tau(s)\); hence
  \(\tilde\beta_r(s) = C_r(s) - \tilde\tau(s) = \Gamma N_r(s)\), and
  (A2) holds with loading \(\Gamma\). (If the entry distributions of
  \(X\) differ, the joint version of (A1) is refuted by the data and the
  identified set is empty.) Cohorts are sampled independently, so the
  per-cohort structures combine. Conversely, a \(\Gamma\) outside
  \(\Gamma^{*}\) implies, at some pair and some stratum, a shift outside
  that stratum's location set, which Theorem~\ref{thm-bounds} applied
  within the stratum excludes. The two outer bounds:
  \(\mathcal{H}_r(s \mid X) \subseteq \mathcal{H}_r(s)\) by the
  inclusion in Proposition~\ref{prp-covariates}, whose proof applies
  because the strata's sub-measures add up to the pooled ones, and the
  per-pair bound drops the requirement that one \(\Gamma\) serve all
  pairs; blocks 12 and 11 of the deterministic checks in
  Section~\ref{sec-ademp} give cases in which each is strictly larger.
\item
  The observable moments of episode \(r\) involve \((\tau, \beta_r)\)
  only through \(C_r(s) = \tau(s) + \beta_r(s)\) (given (A1)), so the
  map \((\tau, \beta_r) \mapsto\) moments factors through the sum, and
  the \(\epsilon\)-perturbation leaves the sum invariant. If the
  location shift \(h + \epsilon\) lies in \(\mathcal{H}_r(s \mid X)\),
  that is, in \(\mathcal{H}_r(s \mid x)\) for every \(x\), at every
  episode observing tenure \(s\), then for each cohort observed at
  tenure \(s\) at a refreshment episode, Theorem~\ref{thm-multiwave}
  applied within each stratum of \(X\) and mixed as in (a), with
  \(J^{*}\) its refreshment tenures, the candidate map at \(s\) shifted
  by \(\epsilon\) and the true maps at its other refreshment tenures,
  and with the entry coordinate carrying the negative-control battery,
  produces a structure that reproduces every observed pattern
  sub-measure, the entry marginal (hence every negative-control moment
  and every entry-wave selection differential), and the refreshment
  marginals; under it the survivors' latent mean at that episode moves
  by \(-\epsilon\), so \(\tilde\beta_r(s) = \beta_r(s) - \epsilon\) and
  \(\tilde\tau(s) = \tau(s) + \epsilon\) while \(C_r(s)\) and \(N_r(s)\)
  are unchanged. Under unimodality the stratum location sets are
  intervals containing \(h\) (Theorem~\ref{thm-bounds}), and so is their
  intersection, which gives the statement about \(\epsilon\) between
  \(0\) and \(\epsilon_0\). For the stationarity clause: the hypothesis
  \(\beta_1 = \beta_2\) constrains the function \(\beta(\cdot)\) only at
  arguments where both episodes evaluate it; with disjoint tenure sets
  the constraint set is vacuous, so the perturbation above again
  reproduces all moments; with a repeated tenure,
  Corollary~\ref{cor-repeated}.
\item
  Under (A2) the \(K\) negative-control moments at \((r,s)\) are
  \(N^{(k)}_r(s) = \lambda_k \beta^{N}_r(s)\), \(k = 1, \dots, K\), with
  known loadings \(\lambda_k\): a one-factor system with \(K\)
  observations and one unknown, hence \(K - 1\) restrictions per
  \((r,s)\), testable as dispersion of the battery around its common
  factor. The final clause follows from the argument in (b). \(\square\)
\end{enumerate}

\end{proof}

The distinction between the marginal and the joint version of (A1) in
(a) can matter only when the battery predicts retention
(\(N_r(s) \ne 0\) at some pair), and it can already arise at a single
pair. Let a binary battery item \(X = \pm 1\) take both values with
probability one half, let the target \(Y^{*}\) be independent of \(X\)
with a Laplace distribution of scale \(b = 1/\sqrt 2\) (both variables
have variance one), let the map be the identity, and let retention be
\(.8\) at \(X = 1\) and \(.2\) at \(X = -1\), independently of \(Y^{*}\)
within strata. Then \(p = .5\), \(C = 0\), \(N = .6\), and the true
loading is zero. With \(\Gamma \in [0, 1]\), the set computed from the
margins admits \(\Gamma \le b\log 2/.6 = .817\): at \(\Gamma = .5\) the
implied shift \(-.3\) satisfies the pooled constraint,
\(.5\,e^{.3/b} = .764 \le 1\), but violates that of the high-retention
stratum, \(.8\,e^{.3/b} = 1.223\). Under the joint version of (A1) the
identified set is \([0, b\log 1.25/.6] = [0, .263]\) (block 12 of the
deterministic checks in Section~\ref{sec-ademp}). A completion that uses
only the margins reproduces both margins but not their association: in
the high-retention stratum it would need a latent law of the target
other than the refreshment's. The joint version is the natural reading
of (A1) when the battery consists of time-invariant facts asked of both
samples; the marginal version is appropriate when only the target's
margin is claimed to be comparable. Conditioning on a 23-item battery is
impractical, but any coarsening \(\varphi(X)\) gives a valid set between
the two, because
\(\mathcal{H}_r(s \mid X) \subseteq \mathcal{H}_r(s \mid \varphi(X)) \subseteq \mathcal{H}_r(s)\);
the set computed from \(\varphi(X)\) is sharp when (A1) is asserted over
\(\varphi(X)\) only.

\begin{corollary}[Generational
refreshment]\protect\hypertarget{cor-generational}{}\label{cor-generational}

Suppose episode \(r\)'s refreshment cohort shares no birth-cohort
support with the incumbents, so (A1) fails between entrants and
incumbents, but two or more incumbent cohorts at episode \(r\) do
overlap. Then: (i) level contrasts against the entrants identify only
the composite \(\Delta g_r(s) + \tau(s) + \beta_r(s)\); after
negative-control correction under (A2) with loading \(\Gamma\),
\(C_r(s) - \Gamma N_r(s) = \tau(s) + \{\Delta g_r(s) - \Gamma\,\Delta g^{N}_r(s)\}\),
so \(\tau(s)\) is interval-identified, with width \(2M\), under a bound
\(|\Delta g_r(s) - \Gamma\,\Delta g^{N}_r(s)| \le M\) on the residual
cohort effect after transport; a bound on \(|\Delta g_r(s)|\) alone does
not suffice when \(\Delta g^{N}_r(s)\) is unrestricted; (ii) the
increment between incumbent tenures, \(\tau(s) - \tau(s')\) for
incumbent cohorts at tenures \(s, s'\), remains point-identified on the
incumbents' common support under (A1) between the two incumbent cohorts
and (A2) with a known loading for their contrast (a bounded loading
gives an interval), by applying Theorem~\ref{thm-tworef}'s argument to
the incumbent pair, with the negative-control contrast between the same
two cohorts supplying \(\beta^{N}_r(s) - \beta^{N}_r(s')\).

\end{corollary}

\begin{proof}
\leavevmode

\begin{enumerate}
\def\labelenumi{(\roman{enumi})}
\tightlist
\item
  With disjoint support, \(\Delta g_r(s)\) cannot be zeroed by (A1) and
  enters every entrant contrast, and \(\Delta g^{N}_r(s)\) enters every
  negative-control contrast; substituting
  \(\beta_r(s) = \Gamma\beta^{N}_r(s)\) gives the display, and the bound
  on the residual gives the interval. If only \(|\Delta g_r(s)| \le M\)
  is assumed, \(\Gamma\Delta g^{N}_r(s)\) is free and the corrected
  contrast is unrestricted. (ii): restrict to the two incumbent cohorts
  (entered at \(e_1, e_2\), observed at episode \(r\) with tenures
  \(s > s'\)) on their overlapping birth years; under (A1) within this
  pair the target contrast is
  \([\tau(s) - \tau(s')] + [\beta_r(s) - \beta_r(s')]\) and the
  same-pair negative-control contrast is
  \(\beta^{N}_r(s) - \beta^{N}_r(s')\); subtract under (A2). \(\square\)
\end{enumerate}

\end{proof}

Refreshments in practice are of two kinds. An \emph{age-matched}
refreshment replenishes the same birth cohorts (attrition repair); a
\emph{generational} refreshment extends the panel to a new cohort
(coverage repair). Chadi (\citeproc{ref-chadi2021}{2021}) distinguishes
a different pair of entrant types (entrants induced by the data
collector and ``natural'' entrants into a household panel) and uses
their coexistence to separate attrition bias from
participation-experience effects; the distinction here is between
entrants whose birth years overlap the incumbents' and entrants from a
new birth cohort. The Japanese panel's 2011 episode is of the first kind
and its 2019 episode of the second: the 2019 entrants' birth years
(1987--1998) do not intersect the incumbents' (1966--1986). Under
(A1)--(A2) the panel therefore identifies \(\tau(5)\) (episode 1,
levels) and \(\tau(13) - \tau(9)\) (episode 2, increments), while
\(\tau(13)\)'s level rests on a bound on the residual cohort effect. For
a panel planning a refreshment, the choice between the two kinds is
therefore also a choice of what becomes identified.

Part (a) shows that point identification comes from the negative-control
battery together with (A2), and not from the second refreshment itself.
The second refreshment extends the set of tenures at which \(\tau\) is
recovered (from \(\{5\}\) to \(\{5, 9, 13\}\) when the second
refreshment is age-matched, and to the increment \(\tau(13) - \tau(9)\)
when it is generational, as in the Japanese panel) and adds
overidentifying restrictions at the newly observed pairs. Part (b) shows
that stationary selection and two refreshments do not identify the path
by differencing. In designs whose refreshment episodes observe disjoint
tenure sets the stationarity restriction binds nothing; in designs that
repeat a tenure it adds a testable restriction and does not shrink the
set. This is the selection-side counterpart of the non-identification
result in the companion identification paper (its Theorem 1): there,
tenure--period--cohort collinearity; here, the sum structure
\(\tau + \beta_r\). Part (c) shows that under (A2) a placebo battery is
an identifying instrument and not only a robustness check. The 23-item
battery of the Japanese panel (Section~\ref{sec-implementation},
Table~\ref{tbl-ncbattery}) is of this kind; its item contrasts at the
2019 episode include one rejection by a per-item test, an item on a
telephone in the childhood home. One reading is that telephone
availability varied with the era of childhood more finely than the
three-year birth bands absorb. Another is a difference in how the entry
questionnaires worded the item, which (A1) assumes away. The dispersion
test of (c) is designed to expose an item of the first kind, which is
also a reason to prefer facts whose distribution does not move with the
era of childhood.

\section{Constructive Designs: Restrictions That Identify the Survivors'
Mean Shift}\label{sec-constructive}

Theorem~\ref{thm-sharp} shows what a refreshment sample identifies
without restrictions on attrition, and Theorem~\ref{thm-tworef} shows
what a second refreshment and a negative-control battery add. Applied
work lies between these two cases: a single refreshment sample is
combined with a restriction on the attrition process, often stated
informally, as in the survival-matching design of Halpern-Manners,
Warren, and Torche (\citeproc{ref-halpernmanners2017}{2017}). In this
section we state such designs as restrictions within the theory. Each
design is a restriction on the selection functional that identifies the
survivors' mean shift, and each is biased in a characteristic way when
its restriction fails. Because the designs rest on different
restrictions, their disagreements can be used to check whether the
restrictions are compatible. Table~\ref{tbl-prior}, at the end of the
section, collects the restrictions stated in prior comparisons of
experienced and inexperienced respondents and their counterparts here.
The designs are stated in the tenure notation used throughout, with the
dose \(k\) of prior interviews written as \(k = s - 1\).

\subsection{Multi-wave setting and the selection
functional}\label{sec-multiwave}

Specialize the multiwave setting of Section~\ref{sec-mw-setting} to the
cohort structure of the moment system (Section~\ref{sec-momentsystem})
under monotone attrition. A continuing cohort enters at wave \(e\) and
is compared, at wave \(t > e\), with a fresh cohort entering at \(t\);
the continuing cohort's tenure is \(s = t - e + 1\) (dose
\(k = s - 1\)). Write \(R_{iw} \in \{0,1\}\) for a completed interview
at wave \(w\) and \(S_{i,a:b} = \prod_{w=a}^{b} R_{iw}\) for survival
through waves \(a, \dots, b\). The decision governing \(R_{i,w+1}\) is
taken at wave \(w\), after the interview at \(w\), and may depend on the
transient component \(\varepsilon_{iw}\) of that wave's latent response.
Hence \(S_{i,e:t} = 1\) reflects decisions taken at waves
\(e, \dots, t-1\) and, under the independent-innovation model (M) of
Proposition~\ref{prp-fail}, is not selected on \(\varepsilon_{it}\). A
fresh cohort's \(S_{i,t:t+k} = 1\) reflects decisions at
\(t, \dots, t+k-1\) and is selected on \(\varepsilon_{it}\). This is the
timing implemented in the simulation code and the package. The latent
(unconditioned) response at \(t\) is \(Y^{*}_{it}\); the continuing
survivors report \(Y_{it} = c(Y^{*}_{it})\), the fresh cohort reports
\(Y^{*}_{it}\). We write
\(h^{S} := \mathbb{E}[Y_{t} - Y^{*}_{t} \mid S_{e:t} = 1]\) for the
survivors' mean conditioning effect, which is defined for any
conditioning mechanism and equals \(h\) in the location-shift case
\(c(y) = y + h\).

Every design must identify the \emph{selection functional}
\begin{equation}\phantomsection\label{eq-selfun}{
\delta_e(w) \;:=\; \mathbb{E}\big[Y^{*}_{w} \mid S_{e:t} = 1,\ \text{cohort } e\big] \;-\; \mathbb{E}\big[Y^{*}_{w} \mid \text{cohort } e\big], \qquad w \in \{e, \dots, t\},
}\end{equation} the difference between eventual survivors and the whole
entry cohort on the latent response at wave \(w\). At \(w = t\) it is
the raw-scale version of \(\beta_r(s)\) in
Section~\ref{sec-momentsystem}; at \(w = e\) it is observable, because
every entrant reports at entry and entry reports are unconditioned. The
naive contrast between continuing survivors and the fresh cohort
decomposes, under refreshment validity, as
\begin{equation}\phantomsection\label{eq-naive}{
C \;:=\; \mathbb{E}[Y_t \mid S_{e:t}=1] - \mathbb{E}[Y^{*}_t \mid \text{fresh}] \;=\; h^{S} + \delta_e(t),
}\end{equation} so that \(h^{S}\) is identified if and only if
\(\delta_e(t)\) is. By Theorem~\ref{thm-sharp}, the data at \(t\) bound
\(\delta_e(t)\) but do not identify it, and by
Theorem~\ref{thm-multiwave} the cohort's own earlier waves do not
identify it either. Every design below therefore rests on a restriction
that links \(\delta_e(t)\) to something observed.

We use the following restrictions. \textbf{B1} is refreshment validity
together with cohort overlap (A1): the fresh cohort's latent
distribution at \(t\) equals the continuing cohort's population latent
distribution at \(t\) on a common support.

\begin{itemize}
\tightlist
\item
  \textbf{B2 (survival-matched attrition exchangeability).}
  \(\mathbb{E}[Y^{*}_t \mid S_{e:t}=1,\ \text{cont}] = \mathbb{E}[Y^{*}_t \mid S_{t:t+k}=1,\ \text{fresh}]\):
  fresh respondents who go on to survive the same number of waves as the
  continuing survivors have the same latent mean.
\item
  \textbf{B3 (stationary, symmetric attrition).}
  \(\mathbb{E}[Y^{*}_t \mid S_{e:t+1}=1,\ \text{cont}] = \mathbb{E}[Y^{*}_t \mid S_{t:t+k+1}=1,\ \text{fresh}]\):
  when both arms are selected by exactly \(k+1\) response decisions, one
  of them taken at \(t\), the selected latent means coincide. The
  stationary case of model (M) in Proposition~\ref{prp-fail}(i) (a
  common loading, traits equally distributed in the two cohorts,
  innovations independent and identically distributed across waves and
  independent of the traits, and a common continuation probability
  \(\pi(A_i, \varepsilon_{iw})\) at every wave of both cohorts with
  decisions conditionally independent across waves) provides sufficient
  primitive conditions for B3.
\item
  \textbf{B4 (time-invariant selection differential).}
  \(\delta_e(t) = \delta_e(e)\): the survivors' advantage on the latent
  response is the same at the comparison wave as at entry. This is (A2)
  of Section~\ref{sec-momentsystem} with the negative control replaced
  by the item's own entry-wave report and loading \(\Gamma = 1\). It is
  a parallel-trends restriction on the unconditioned mean between
  survivors and the full cohort. For a binary outcome, it is Assumption
  3\(^{\mathrm{Alt}}\) of Das, Toepoel, and van Soest
  (\citeproc{ref-dasetal2011}{2011}, online Appendix 1, p.~27).
\end{itemize}

B2--B4 are restrictions on means of the latent outcome; they identify
the mean shift \(h^{S}\) (or \(h^{S+}\)), and they identify neither the
conditioning map \(c\) itself nor the selection process. The
distributional versions of survival matching and symmetric matching
require the corresponding restrictions to hold in distribution, which is
stronger than B2 or B3 and which we do not test. In those versions the
whole latent distribution of the matched fresh arm is taken as that of
the survivors, so that \(c\) is recovered on that distribution's support
by the quantile alignment of Corollary~\ref{cor-noattrition}. No mean
restriction determines a selection function: many selection processes
are compatible with B2, B3, or B4, and for this reason the designs are
compared against the identified set of Theorem~\ref{thm-sharp} and not
against any particular selection model.

B2--B4 are alternative restrictions and are not nested. B2 needs \(k\)
waves of follow-up on the fresh cohort; B3 needs \(k+1\); B4 needs none,
but needs the item at the cohort's entry wave with the same coding.

\subsection{Three designs and what each identifies}\label{sec-designs}

\begin{proposition}[Survival
matching]\protect\hypertarget{prp-sm}{}\label{prp-sm}

Under B1 and B2, \(h^{S}\) is point identified:
\(h^{S} = \mathbb{E}[Y_t \mid S_{e:t}=1,\ \text{cont}] - \mathbb{E}[Y^{*}_t \mid S_{t:t+k}=1,\ \text{fresh}]\).
If B2 holds in distribution (the two conditional distributions of
\(Y^{*}_t\) coincide), the map \(c\) is identified on the survivor
latent-outcome support in the model augmented by distributional B2,
\(c = Q^{-1} \circ F_{Y^{*}_t \mid S_{t:t+k}=1,\ \text{fresh}}\) almost
everywhere there and hence, by continuity, on all of it; the
unrestricted set \(\mathcal{C}_{ID}\) of Theorem~\ref{thm-sharp} is
unchanged, distributional B2 being an additional restriction that
selects, when it is compatible with the data, the members of
\(\mathcal{C}_{ID}\) that equal this map on that support.

\end{proposition}

\begin{proposition}[Symmetric survival
matching]\protect\hypertarget{prp-ssm}{}\label{prp-ssm}

Under B1 and B3, the mean shift among survivors who also respond at
\(t+1\), \(h^{S+} := \mathbb{E}[Y_t - Y^{*}_t \mid S_{e:t+1}=1]\), is
point identified:
\(h^{S+} = \mathbb{E}[Y_t \mid S_{e:t+1}=1,\ \text{cont}] - \mathbb{E}[Y^{*}_t \mid S_{t:t+k+1}=1,\ \text{fresh}]\);
with B3 in distribution, \(c\) is identified on the latent-outcome
support of that subpopulation in the model augmented by distributional
B3. The two survivor estimands differ by a composition term,
\(h^{S} = r\,h^{S+} + (1-r)\,h^{S-}\) with
\(r = \Pr(R_{t+1}=1 \mid S_{e:t}=1)\) and \(h^{S-}\) the shift among
survivors who do not respond at \(t+1\).

\end{proposition}

\begin{proposition}[Entry-wave correction (the Das--Toepoel--van Soest
functional)]\protect\hypertarget{prp-ec}{}\label{prp-ec}

Under B1 and B4, \(h^{S}\) is point identified from the continuing
cohort's own entry wave and the fresh cohort's first wave:
\begin{equation}\phantomsection\label{eq-ec}{
h^{S} \;=\; \mathbb{E}[Y_t \mid S_{e:t}=1,\ \text{cont}] \;-\; \underbrace{\big(\mathbb{E}[Y_e \mid S_{e:t}=1,\ \text{cont}] - \mathbb{E}[Y_e \mid \text{cont}]\big)}_{\delta_e(e)} \;-\; \mathbb{E}[Y^{*}_t \mid \text{fresh}] .
}\end{equation} In the location-shift submodel augmented by B4, the
identified set for \(h\) is
\(\mathcal{H}_{ID} \cap \{C - \delta_e(e)\}\), with \(\mathcal{H}_{ID}\)
the unchanged set of Theorem~\ref{thm-bounds}: the point
\(C - \delta_e(e)\) when the restrictions are compatible with the
observations and empty otherwise. The mean identity itself needs no
location-shift structure. The design uses no follow-up of the fresh
cohort, is feasible for every \(s \le t\), and admits a
covariate-standardized version. Writing \(\delta_e(\cdot \mid x)\) for
the selection functional within \(X = x\), entry covariates observed in
both cohorts, and \(m^{S}_e\), \(m_e\) and \(m^{*}_t\) for the
regression functions on \(x\) of the entry outcome among the survivors,
of the entry outcome among all entrants and of the fresh cohort's
outcome, the conditional restriction
\(\delta_e(t \mid x) = \delta_e(e \mid x)\) identifies
\(h^{S} = \mathbb{E}[Y_t \mid S] - \mathbb{E}[m^{S}_e(X) - m_e(X) \mid S] - \mathbb{E}[m^{*}_t(X) \mid S]\),
every regression prediction being integrated over the survivors'
covariate distribution. The restriction bears on the terms within each
\(x\); it identifies the shift only when the two predictions are
integrated over one and the same distribution, and a version that
averages the entry term over a subset of the survivors (those with an
observed entry report, or those eligible at entry) and the fresh
prediction over all of them is not identified by it alone
(Section~\ref{sec-implementation}). The observable combination in
Equation~\ref{eq-ec},
\(\mathbb{E}[Y_t \mid S] - \mathbb{E}[Y_e \mid S] + \mathbb{E}[Y_e] - \mathbb{E}[Y^{*}_t \mid \text{fresh}]\),
is the one Das, Toepoel, and van Soest
(\citeproc{ref-dasetal2011}{2011}) use for a binary item under their
Assumption 3 {[}pp.~43--44{]} and, for the survivor effect, under their
Assumption 3\(^{\mathrm{Alt}}\) {[}online Appendix 1, pp.~27--28{]}; the
proposition extends the survivor-effect result to general outcome means
and does not propose a new estimator. Das, Toepoel, and van Soest
(\citeproc{ref-dasetal2011}{2011}) distinguish two targets, the
population and the survivor effect: their main-text Assumption 3
concerns conditioned outcomes, whereas Assumption 3\(^{\mathrm{Alt}}\)
concerns unconditioned outcomes and identifies the binary survivor
effect, and the identifying expression under Assumption
3\(^{\mathrm{Alt}}\) {[}online Appendix 1, p.~28{]} is the same
observable functional as under Assumption 3 {[}p.~44{]}. B4 extends
Assumption 3\(^{\mathrm{Alt}}\) to general outcome means. Under a
homogeneous shift the two targets coincide; under heterogeneous effects
they do not, and a design can satisfy one restriction while violating
the other. Simulation 3 (Appendix) illustrates the difference for a
continuous outcome.

\end{proposition}

\begin{proof}
All three follow from Equation~\ref{eq-naive}. Under B1 the fresh
cohort's mean of \(Y^{*}_t\) (or of its survival-matched subset) is an
unbiased proxy for the corresponding continuing-cohort latent mean
whenever the restriction identifies the survivors' latent mean: under B2
it is the fresh \(k\)-wave survivors' mean, under B3 the fresh
\((k+1)\)-wave survivors' mean (with the continuing arm restricted to
\(S_{e:t+1}=1\)), under B4 the population mean plus \(\delta_e(e)\),
which is observed because entry reports are unconditioned and available
for every entrant. Subtracting the proxy from the survivors' reported
mean leaves \(h^{S}\) (or \(h^{S+}\)). The distributional statements
replace means by distributions and apply
Corollary~\ref{cor-noattrition}: once the latent distribution of the
relevant survivors is known, \(c\) is, on that distribution's support,
the increasing rearrangement transporting it to those survivors'
reported distribution. The composition identity in
Proposition~\ref{prp-ssm} is the law of total expectation over
\(R_{t+1}\) among survivors. In the location-shift submodel
\(h^{S} = h\), so under B4 the identified value of \(h\) is
\(C - \delta_e(e)\), which lies in \(\mathcal{H}_{ID}\) if and only if
B4 is compatible with the data.
\end{proof}

The three propositions rest on one argument: each restriction names the
observed quantity from which the survivors' latent mean is obtained, and
so identifies the relevant survivor mean without determining the
selection function uniquely. Theorem~\ref{thm-tworef} is the same idea
with a battery of items whose \(h\) is zero by construction in place of
the item's own entry wave; the entry-wave correction is its single-item
version, available when the item has a comparable entry-wave measure,
but without overidentifying restrictions.

The comparison with Das, Toepoel, and van Soest
(\citeproc{ref-dasetal2011}{2011}) concerns the model, the targets, the
restrictions and the identified objects. Their conditioning model lets
each binary answer change at reinterview, a stochastic change of
answers; ours is a deterministic monotone map, although the mean
statements of Propositions \ref{prp-sm}--\ref{prp-ec}, unlike their
distributional and location-shift clauses, do not use it. Their main
text defines the population conditioning effect PC and, as an
alternative (\citeproc{ref-dasetal2011}{Das, Toepoel, and van Soest
2011, 41}), the survivors' effect PC\(^{\mathrm{Alt}}\) together with
the attrition bias AB\(^{\mathrm{Alt}}\) of unconditioned answers; for a
binary item these are \(h^{S}\) and \(\delta_e(t)\). Their Assumption 3
{[}p.~43{]} fixes the attrition bias of conditioned answers over time,
and their Assumption 3\(^{\mathrm{Alt}}\) {[}online Appendix 1, p.~27{]}
fixes AB\(^{\mathrm{Alt}}\) and is B4 for a binary item. Both
assumptions identify their target through the observable functional of
Equation~\ref{eq-ec} {[}p.~44; online Appendix 1, p.~28{]}, and without
them the sharp binary bounds are those on PC {[}p.~41{]} and on
PC\(^{\mathrm{Alt}}\) {[}online Appendix 1, p.~26{]}, the latter being
the binary case of the survivor bounds in Remark~\ref{rem-modelfree}.
The survivor interpretation of the entry-wave functional for binary
items is therefore theirs. They also note that their approach extends to
non-binary outcomes through the binary events \(\{Y > y\}\) {[}pp.~36,
52{]}; imposed on every such event, Assumption 3\(^{\mathrm{Alt}}\)
implies B4 whenever \(Y_e\) and \(Y^{*}_t\) have finite means (integrate
the event-wise differences over \(y\)), whereas B4 restricts only the
mean. This section adds the formulation for general outcome means (B4)
in the multiwave notation, the symmetric comparison, the bias identities
of Proposition~\ref{prp-fail}, and the diagnostics of
Section~\ref{sec-diagnostics}.

\begin{refremark}[Transporting the retention function]
Das, Toepoel, and van Soest (\citeproc{ref-dasetal2011}{2011, 43},
Assumption 4) also consider a stationarity restriction on the attrition
probability and not on the attrition bias: attrition depends on the
first- and second-wave answers in the same way. Its version for
unconditioned answers {[}online Appendix 1, pp.~27--28, Assumption
4\(^{\mathrm{Alt}}\){]}, written for the two-wave setting of
Theorem~\ref{thm-sharp} and a general outcome, is
\begin{equation}\phantomsection\label{eq-retention}{
\Pr(S = 1 \mid Y_2^{*} = y) \;=\; \pi_1(y) := \Pr(S = 1 \mid Y_1 = y) \qquad \text{for } F_2\text{-almost every } y,
}\end{equation} where \(\pi_1 = d(p\,G_1)/dF_1\) is the
outcome-conditional retention function at entry,
\(\Pr(S = 1 \mid Y_1 = y)\), not the bivariate selection function
\(\pi(y_1, y_2)\) of Section~\ref{sec-setup}. It is identified because
the first wave is observed for the whole cohort, but only \(F_1\)-almost
surely, so Equation~\ref{eq-retention} has content only where the
second-wave latent distribution puts mass on values that the first wave
takes: under \(F_2 \ll F_1\) (or an explicit rule specifying \(\pi_1\)
on every \(F_1\)-null set to which \(F_2\) assigns positive mass), the
measure \(H(dy) := \pi_1(y)\,F_2(dy)/p\) is identified. The restriction
then says that \(H\) is the stayers' latent distribution, so \(H\) must
have total mass one, which is a check on the restriction. By the
quantile alignment of Corollary~\ref{cor-noattrition} applied to \(H\)
and \(Q\), \(c\) is identified on the support of \(H\) as the increasing
rearrangement of \(H\) onto \(Q\). On that support the restriction
selects the member of \(\mathcal{C}_{ID}\) whose rationalizing retention
function, as a function of the latent outcome, equals the one observed
at entry, provided such a member exists (the normalization of \(H\) does
not guarantee a map in \(\mathcal{M}\) compatible with the observables);
outside it the map is restricted only by membership in \(\mathcal{M}\).
Without the support condition even the survivors' mean shift is not
identified. Let \(F_1\) be uniform on \([0, 1]\), \(F_2\) and \(Q\)
uniform on \([2, 3]\), \(p = 1/2\), and let \(Y_1\) be independent of
the second-wave variables among stayers. Retention equal to \(1/2\)
throughout with \(c\) the identity, and retention equal to one on
\([2, 2.5]\) and zero elsewhere with \(c(y) = 2y - 2\), produce the same
observables and the same entry function \(\pi_1 = 1/2\) on \([0, 1]\);
their survivor mean shifts are \(0\) and \(.25\). The restriction
constrains the outcome-conditional retention function and not only a
mean, and it is not implied by B4. Under (M) of
Proposition~\ref{prp-fail} it holds when selection acts on traits alone
and the latent distribution does not change between the waves
(\(\mu_1 = \mu_2\), \(\lambda_1 = \lambda_2\)), because
\((A_i, Y^{*}_{i1})\) and \((A_i, Y^{*}_{i2})\) then have the same law
and survival depends on \(A_i\) alone. It can fail under state-dependent
attrition, as B2 and B4 do, because the response decision taken at the
first wave depends on that wave's innovation, though not necessarily: a
dependence on the innovation that averages out within each value of the
outcome (as when its sign is set by a trait with mean zero) leaves both
functions equal. And, unlike B4, it can also fail when only the period
mean changes (\(\mu_1 \ne \mu_2\)), which shifts the second function
relative to the first, unless retention does not depend on the outcome
at all.

\label{rem-retention}

\end{refremark}

\subsection{Failure modes}\label{sec-failure}

The three designs rest on different restrictions, so they fail under
different violations. Official statistics show that restricting the
comparison to respondents present at every wave does not suffice to
remove the pattern. Among Current Population Survey respondents who
reported labour force status in all eight interviews, the slope of the
multiplicative rotation-group index of the unemployment rate over the
eight months in sample is \(-1.23\) in 1982--1993 and \(-1.22\) in
1994--2014 (\citeproc{ref-krueger2017}{Krueger, Mas, and Niu 2017},
Table 5, p.~263). The bias identities below are exact; the directions
require a model of how traits, transient states and response decisions
are related.

\begin{proposition}[Biases under
violation]\protect\hypertarget{prp-fail}{}\label{prp-fail}

Let
\(\Delta_{SM} := \mathbb{E}[Y^{*}_t \mid S_{e:t}=1,\ \text{cont}] - \mathbb{E}[Y^{*}_t \mid S_{t:t+k}=1,\ \text{fresh}]\),
\(\Delta_{SSM} := \mathbb{E}[Y^{*}_t \mid S_{e:t+1}=1,\ \text{cont}] - \mathbb{E}[Y^{*}_t \mid S_{t:t+k+1}=1,\ \text{fresh}]\),
and \(\eta := \delta_e(t) - \delta_e(e)\). Then, under B1, \[
\hat h_{SM} \to h^{S} + \Delta_{SM}, \qquad \hat h_{EC} \to h^{S} + \eta, \qquad \hat h_{SSM} \to h^{S+} + \Delta_{SSM}.
\] Suppose in addition model (M): the latent response is
\(Y^{*}_{iw} = \mu_w + \lambda_w' A_i + \varepsilon_{iw}\), with a
vector \(A_i\) of person-level traits distributed identically in the two
cohorts and transient innovations \(\varepsilon_{iw}\) that are
independent and identically distributed across persons and waves, have
mean zero, and are independent of \(A_i\); and the response decision
taken at wave \(w\) continues participation with probability
\(\pi_w(A_i, \varepsilon_{iw})\), the decisions being independent across
waves given \((A_i, \varepsilon_{i\cdot})\). Then:

\begin{enumerate}
\def\labelenumi{(\roman{enumi})}
\item
  \emph{(State-dependent, stationary attrition.)} If
  \(\lambda_w \equiv \lambda\) and \(\pi_w \equiv \pi\) at every wave of
  both cohorts, then \(\Delta_{SSM} = 0\) and \[
  \Delta_{SM} \;=\; \eta \;=\; -\,\mathbb{E}[\varepsilon_t \mid S_{t:t+k} = 1,\ \text{fresh}] ;
  \] if \(\pi(a, \cdot)\) is nondecreasing for every \(a\) (continuation
  more likely at higher transient values), \(\Delta_{SM} = \eta \le 0\),
  and if it is nonincreasing, \(\Delta_{SM} = \eta \ge 0\). Survival
  matching and the entry-wave correction carry the same bias, and
  symmetric matching is unbiased for \(h^{S+}\).
\item
  \emph{(Non-stationary trait selection.)} If
  \(\lambda_w \equiv \lambda\) and \(\pi_w(a, \varepsilon) = \pi_w(a)\)
  does not depend on \(\varepsilon\), then \(\eta = 0\), while
  \(\Delta_{SM} = \lambda'\{\mathbb{E}[A \mid S_{e:t}=1, \text{cont}] - \mathbb{E}[A \mid S_{t:t+k}=1, \text{fresh}]\}\)
  and
  \(\Delta_{SSM} = \lambda'\{\mathbb{E}[A \mid S_{e:t+1}=1, \text{cont}] - \mathbb{E}[A \mid S_{t:t+k+1}=1, \text{fresh}]\}\):
  the entry-wave correction is unbiased, and survival and symmetric
  matching are biased whenever these trait contrasts are nonzero, that
  is, whenever selection makes the survivors' mean of \(\lambda'A\)
  differ between the continuing cohort's selection waves and the fresh
  cohort's.
\item
  \emph{(Drift in the trait loading.)} If
  \(\pi_w(a, \varepsilon) = \pi(a)\) at every wave of both cohorts but
  \(\lambda_t \ne \lambda_e\), then \(\Delta_{SM} = \Delta_{SSM} = 0\)
  and
  \(\eta = (\lambda_t - \lambda_e)'\{\mathbb{E}[A \mid S_{e:t}=1, \text{cont}] - \mathbb{E}[A \mid \text{cont}]\}\):
  only the entry-wave correction is biased; this is the failure of the
  parallel-trends content of B4.
\end{enumerate}

Outside (M) the directions are not implied. With symmetric innovations
and continuation depending on \(\varepsilon_t^2\), selection leaves the
mean of \(\varepsilon_t\) unchanged among survivors; and with serially
dependent innovations, survivors selected on earlier innovations are
also selected on the current one.

\end{proposition}

\begin{proof}
The identities: substitute the proxies of Propositions
\ref{prp-sm}--\ref{prp-ec} into Equation~\ref{eq-naive}; the SM proxy
misses the continuing survivors' latent mean by \(\Delta_{SM}\), the SSM
proxy misses the corresponding mean of the survivors through \(t+1\) by
\(\Delta_{SSM}\), and the EC proxy misses by
\(\delta_e(t) - \delta_e(e)\). Under (M) the common term \(\mu_t\)
cancels from \(\Delta_{SM}\), \(\Delta_{SSM}\) and \(\eta\). The
continuing survivors \(S_{e:t} = 1\) have passed the \(k\) decisions
taken at waves \(e, \dots, t-1\), and the fresh survivors
\(S_{t:t+k} = 1\) the \(k\) decisions taken at \(t, \dots, t+k-1\).

\begin{enumerate}
\def\labelenumi{(\roman{enumi})}
\item
  With stationary \(\pi\) and \(\lambda\) and identically distributed
  innovations, the joint law of \((A, \text{decisions})\) is the same
  for the \(k\) decisions of the two arms, so
  \(\mathbb{E}[A \mid S_{e:t}, \text{cont}] = \mathbb{E}[A \mid S_{t:t+k}, \text{fresh}]\).
  The innovation \(\varepsilon_t\) is independent of \(A\) and of the
  continuing cohort's decisions before \(t\), so
  \(\mathbb{E}[\varepsilon_t \mid S_{e:t}, \text{cont}] = 0\), whereas
  the fresh cohort's first decision is taken at \(t\) and depends on
  \(\varepsilon_t\); hence
  \(\Delta_{SM} = -\mathbb{E}[\varepsilon_t \mid S_{t:t+k}, \text{fresh}]\).
  For the entry-wave correction,
  \(\delta_e(t) = \lambda'\{\mathbb{E}[A \mid S_{e:t}] - \mathbb{E}A\}\)
  because \(\varepsilon_t\) is independent of the decisions before
  \(t\), while
  \(\delta_e(e) = \lambda'\{\mathbb{E}[A \mid S_{e:t}] - \mathbb{E}A\} + \mathbb{E}[\varepsilon_e \mid S_{e:t}]\)
  because the first decision of the continuing cohort is taken at \(e\)
  and depends on \(\varepsilon_e\): the survivors are selected on the
  entry innovation and not on the current one. So
  \(\eta = -\mathbb{E}[\varepsilon_e \mid S_{e:t}, \text{cont}]\), and
  since the continuing cohort's survivors pass \(k\) decisions, the
  first depending on \(\varepsilon_e\), exactly as the fresh cohort's
  pass \(k\) decisions, the first depending on \(\varepsilon_t\), the
  two expectations are equal: \(\eta = \Delta_{SM}\). For the sign,
  condition on \(A = a\): the other decisions are independent of the
  relevant innovation given \(a\), so
  \(\mathbb{E}[\varepsilon_t\, 1\{S\} \mid a]\) has the sign of
  \(\mathrm{Cov}(\varepsilon_t, \pi(a, \varepsilon_t))\), which is
  nonnegative when \(\pi(a, \cdot)\) is nondecreasing, by the covariance
  inequality for monotone functions of one variable, and nonpositive
  when it is nonincreasing. For \(\Delta_{SSM}\), both arms pass
  \(k + 1\) decisions, exactly one of them taken at \(t\) and depending
  on \(\varepsilon_t\); because the innovations are exchangeable across
  waves and \(\pi\) is stationary, the position of that decision among
  the \(k+1\) does not change the joint law of \((A, \varepsilon_t)\)
  given survival, so the two selected latent means coincide.
\item
  Decisions do not depend on innovations, so
  \(\mathbb{E}[\varepsilon_w \mid S] = 0\) for every \(w\) and both
  arms,
  \(\delta_e(t) = \delta_e(e) = \lambda'\{\mathbb{E}[A \mid S_{e:t}] - \mathbb{E}A\}\),
  and \(\eta = 0\); the expressions for \(\Delta_{SM}\) and
  \(\Delta_{SSM}\) are the trait parts of the definitions.
\item
  Both arms of SM and of SSM pass the same number of decisions under the
  same \(\pi\), so their trait laws coincide and the common loading
  \(\lambda_t\) at the comparison wave cancels; the innovations do not
  enter the decisions; and
  \(\delta_e(t) - \delta_e(e) = (\lambda_t - \lambda_e)'\{\mathbb{E}[A \mid S_{e:t}] - \mathbb{E}A\}\).
  The two examples in the last paragraph are immediate. \(\square\)
\end{enumerate}

\end{proof}

Under (M), violations (i) and (ii) of Proposition~\ref{prp-fail} bias
different subsets of the three estimators: stationary state dependence
biases survival matching and the entry-wave correction equally and
leaves symmetric matching unbiased; non-stationary trait selection
biases the two matching designs and leaves the entry-wave correction
unbiased. Section~\ref{sec-diagnostics} uses their disagreements as
compatibility checks. Outside (M) the directions are not implied;
Simulation 3 illustrates them in one parametrization.

Part (i) bears directly on the survival-matching comparisons used in
practice. In the comparisons of adjacent rotation groups of the Current
Population Survey (\citeproc{ref-halpernmanners2012}{Halpern-Manners and
Warren 2012, 1506--8}; \citeproc{ref-warren2012}{Warren and
Halpern-Manners 2012, 517}), of rotation groups one year apart in its
December Food Security Supplement (\citeproc{ref-warren2024}{Warren,
Himmelstern, and Halpern-Manners 2024, 198}), and of the 2006 and 2008
cohorts of the General Social Survey
(\citeproc{ref-halpernmanners2017}{Halpern-Manners, Warren, and Torche
2017, 108--9}), \(k = 1\). The fresh arm's only selection decision is
then the one taken after the comparison-wave interview, whereas the
continuing arm's was taken a wave earlier. Restricting both arms to
respondents with the same number of completed waves therefore does not
suffice to remove attrition bias. Williams and Mallows
(\citeproc{ref-williamsmallows1970}{1970, 1339}) made this point for the
comparison of a cohort's first and second interviews among respondents
present at both, and it qualifies the view that such a restriction
eliminates attrition error ``without having to impose any further
assumptions'' (\citeproc{ref-bach2018}{Bach 2018, 14}). Equal attrition
mechanisms at corresponding waves, the assumption that Okubo
(\citeproc{ref-okubo2024}{2024}) adds to survival matching, do not
remove the bias either: under (M) that assumption is the stationarity of
part (i), and the bias remains because the comparison wave closes the
continuing survivors' selection window and opens the fresh survivors'.
The comparison requires B2 and, under stationary state dependence, the
symmetric version B3.

\begingroup\footnotesize\setstretch{1.05}
\begin{longtable}[]{@{}
  >{\raggedright\arraybackslash}p{(\linewidth - 6\tabcolsep) * \real{0.1591}}
  >{\raggedright\arraybackslash}p{(\linewidth - 6\tabcolsep) * \real{0.3864}}
  >{\raggedright\arraybackslash}p{(\linewidth - 6\tabcolsep) * \real{0.2045}}
  >{\raggedright\arraybackslash}p{(\linewidth - 6\tabcolsep) * \real{0.2500}}@{}}
\caption{Restrictions used in prior comparisons of experienced and
inexperienced respondents, and their counterparts
here}\label{tbl-prior}\tabularnewline
\toprule\noalign{}
\begin{minipage}[b]{\linewidth}\raggedright
Comparison in prior work
\end{minipage} & \begin{minipage}[b]{\linewidth}\raggedright
Restriction, as stated there
\end{minipage} & \begin{minipage}[b]{\linewidth}\raggedright
Counterpart here
\end{minipage} & \begin{minipage}[b]{\linewidth}\raggedright
What the counterpart shows
\end{minipage} \\
\midrule\noalign{}
\endfirsthead
\toprule\noalign{}
\begin{minipage}[b]{\linewidth}\raggedright
Comparison in prior work
\end{minipage} & \begin{minipage}[b]{\linewidth}\raggedright
Restriction, as stated there
\end{minipage} & \begin{minipage}[b]{\linewidth}\raggedright
Counterpart here
\end{minipage} & \begin{minipage}[b]{\linewidth}\raggedright
What the counterpart shows
\end{minipage} \\
\midrule\noalign{}
\endhead
\bottomrule\noalign{}
\endlastfoot
Continuing vs.~fresh respondents, unadjusted & Attrition completely at
random, or at random given the first answer, ``typically'' assumed in
the conditioning literature (\citeproc{ref-dasetal2011}{Das, Toepoel,
and van Soest 2011, 42--43}); valid only without attrition or with
attrition orthogonal to the outcome (\citeproc{ref-warren2024}{Warren,
Himmelstern, and Halpern-Manners 2024, 197}) & \(\delta_e(t) = 0\) in
Equation~\ref{eq-naive} & Otherwise the contrast adds \(\delta_e(t)\),
which differential nonresponse alone can generate
(\citeproc{ref-williamsmallows1970}{Williams and Mallows 1970,
1342--43}); Theorem~\ref{thm-sharp} bounds the split \\
Worst-case bounds & None (\citeproc{ref-dasetal2011}{Das, Toepoel, and
van Soest 2011, 41}, and online Appendix 1, p.~26) &
Theorem~\ref{thm-sharp} with an unrestricted map
(Corollary~\ref{cor-meanbounds}) & The map restriction narrows the set
(Corollary~\ref{cor-funnel}) \\
Weighting or covariate adjustment & Attrition random given observed
characteristics (\citeproc{ref-halpernmanners2017}{Halpern-Manners,
Warren, and Torche 2017, 108}; \citeproc{ref-kraemer2024}{Kraemer et al.
2024, 45}; \citeproc{ref-struminskaya2016}{Struminskaya 2016,
100--101}); adjusts observed differences only
(\citeproc{ref-warren2012}{Warren and Halpern-Manners 2012, 515};
\citeproc{ref-halpernmanners2014}{Halpern-Manners, Warren, and Torche
2014, n. 1}, p.~567); untestable (\citeproc{ref-bach2018}{Bach 2018,
15}) & Missing at random given covariates: the current selection bias
vanishes within cells, \(\delta_e(t \mid x) = 0\), one sufficient
identifying restriction for such adjustments; distinct from the
conditional entry-wave restriction
\(\delta_e(t \mid x) = \delta_e(e \mid x)\) (Proposition~\ref{prp-ec}),
which allows both biases to be nonzero and equal & Covariate balance
alone establishes neither restriction; without maintained conditional
ignorability, the identification problem of Theorem~\ref{thm-sharp}
remains within covariate cells (Section~\ref{sec-shrink}) \\
Both arms restricted to the same number of completed waves & \(k = 1\)
(\citeproc{ref-halpernmanners2012}{Halpern-Manners and Warren 2012};
\citeproc{ref-halpernmanners2017}{Halpern-Manners, Warren, and Torche
2017}; \citeproc{ref-warren2024}{Warren, Himmelstern, and
Halpern-Manners 2024}, one annual interval of the December Food Security
Supplement), \(k = 3\) (\citeproc{ref-eckmanbach2021}{Eckman and Bach
2021}), \(k = 4\) (\citeproc{ref-okubo2024}{Okubo 2024}). Attrition
``cannot explain'' the contrast
(\citeproc{ref-halpernmanners2012}{Halpern-Manners and Warren 2012,
1507}); equal ``propensity to persist''
(\citeproc{ref-warren2012}{Warren and Halpern-Manners 2012, 517};
\citeproc{ref-warren2024}{Warren, Himmelstern, and Halpern-Manners 2024,
198}); cohorts ``equated on both observed and unobserved
characteristics'' (\citeproc{ref-halpernmanners2017}{Halpern-Manners,
Warren, and Torche 2017, 109}); no differences other than exposure,
tested on observables (\citeproc{ref-eckmanbach2021}{Eckman and Bach
2021, 64--65}); no assumption about the form of attrition needed
(\citeproc{ref-bach2018}{Bach 2018, 14}) & B2 (Proposition~\ref{prp-sm})
& Proposition~\ref{prp-fail} (i): biased under stationary
state-dependent attrition; (ii): biased under non-stationary trait
selection \\
The same, with equal attrition mechanisms at corresponding waves & Okubo
(\citeproc{ref-okubo2024}{2024, sec. 2.2}) & B2 with stationary \(\pi\)
in (M) & Proposition~\ref{prp-fail} (i): the bias remains; B3
(Proposition~\ref{prp-ssm}) removes it \\
Both arms restricted to a balanced panel of two entry cohorts & The two
halves of a panel, which entered a year apart, are both restricted to
respondents present at all five waves, because attrition would bias a
comparison of respondents with \(N+1\) and \(N\) participations in the
same year (\citeproc{ref-vanlandeghem2019}{Van Landeghem 2019, 88}) & B3
in the form ``the same number of response decisions in both arms, one of
them at the comparison wave'', which that design achieves at every
comparison year but its last & Unbiased in the stationary case of (M);
Proposition~\ref{prp-fail} (ii): biased under non-stationary trait
selection \\
Stationary attrition bias & Das, Toepoel, and van Soest
(\citeproc{ref-dasetal2011}{2011}) {[}p.~43, Assumption 3; online
Appendix 1, p.~27, Assumption 3\(^{\mathrm{Alt}}\){]}; applied to the
refreshment samples of the German and Swiss household panels by Van
Landeghem (\citeproc{ref-vanlandeghem2014}{2014, 237}), who also
restricts the fresh cohort to respondents who remain for three further
years & B4 (Proposition~\ref{prp-ec}) & Proposition~\ref{prp-fail}
(iii): biased under loading drift; (i): same bias as survival
matching \\
Stationary attrition probability & Das, Toepoel, and van Soest
(\citeproc{ref-dasetal2011}{2011}) {[}p.~43, Assumption 4; online
Appendix 1, pp.~27--28, Assumption 4\(^{\mathrm{Alt}}\){]} & Retention
transport (Remark~\ref{rem-retention}) & Under \(F_2 \ll F_1\),
identifies \(c\) on the support of the transported stayers'
distribution; can fail under state dependence and under a shift of the
latent mean \\
Randomized prior exposure within a panel; attrition checked or modelled
& Attrition independent of the outcome and of assignment given
covariates (\citeproc{ref-torche2012}{Torche et al. 2012, 897}), with a
sensitivity analysis that imputes every attritor as a yes or every
attritor as a no (\citeproc{ref-torche2012}{Torche et al. 2012, n. 1},
p.~914); attrition unrelated to assignment
(\citeproc{ref-halpernmanners2014}{Halpern-Manners, Warren, and Torche
2014, 574}; \citeproc{ref-kraemer2025}{Kraemer et al. 2025, 16--17}) &
Both arms are survivors & Equal attrition rates do not by themselves
make the survivors comparable
(\citeproc{ref-halpernmanners2017}{Halpern-Manners, Warren, and Torche
2017, n. 14}, p.~119; \citeproc{ref-ghanem2026}{Ghanem, Hirshleifer, and
Ortiz-Becerra 2026}, Prop. 3, p.~708); they suffice under random
assignment with monotone selection (\citeproc{ref-lee2009}{Lee 2009},
Remark 2, p.~1084; \citeproc{ref-ghanem2026}{Ghanem, Hirshleifer, and
Ortiz-Becerra 2026}, Prop. 3(ii)) \\
Refreshment identification of attrition & No panel conditioning Das,
Toepoel, and van Soest (\citeproc{ref-dasetal2011}{2011}) & The identity
map & One member of \(\mathcal{C}_{ID}\) \\
\end{longtable}
\endgroup

\subsection{Increments, generational refreshments, and cohort-level
nuisances}\label{sec-increments}

Two further constructions transfer directly. First, when two continuing
cohorts drawn from the same birth cohorts are observed at a common wave
\(t\) with tenures \(s_A > s_B\), applying Proposition~\ref{prp-ec} to
each with its own entry wave identifies the increment
\(h^{S_A}(s_A) - h^{S_B}(s_B)\) without any fresh cohort at \(t\), under
B4 for each cohort and comparability of the two cohorts' latent
distributions at \(t\). This is the constructive version of
Corollary~\ref{cor-generational}(ii), and it is the quantity that
remains identifiable under a generational refreshment, when B1 fails
between entrants and incumbents. Second, suppose the two cohorts differ
in a fixed attribute of the survey process that shifts every report
additively by \(\kappa\); in the JLPS, the 2011 cohort's questionnaires
are returned by mail while the 2007 cohort's are collected by
interviewer visit. Then no within-cohort adjustment removes \(\kappa\),
every level contrast identifies its estimand plus \(\kappa\), and the
diagnostics of Section~\ref{sec-diagnostics}, which compare designs
within the same pair of cohorts, are unaffected. The combination
\(\vartheta := \{h^{S}(s) + \kappa\} - \{[h^{S_A}(s_A) - h^{S_B}(s_B)] + \kappa\}\)
is free of \(\kappa\) but constrains only the difference of two
conditioning contrasts: it equals \(h^{S}(s)\) when the path has
saturated between \(s_B\) and \(s_A\) and zero under a time-homogeneous
linear path with \(s_A - s_B = s - 1\); it does not identify either
contrast separately. We record this because it is the situation of the
JLPS and because \(\vartheta\) is easily misread as a level.

\section{Diagnostics for the Attrition Process}\label{sec-diagnostics}

Because Proposition~\ref{prp-sm} and Proposition~\ref{prp-ec} identify
the same quantity under different restrictions, their difference is an
overidentification statistic, and by Proposition~\ref{prp-fail} its
expectation is a combination of the bias terms; the same holds for
Proposition~\ref{prp-ssm} against Proposition~\ref{prp-sm}.

\begin{proposition}[Two directional
diagnostics]\protect\hypertarget{prp-tests}{}\label{prp-tests}

Define \(\mathcal{T}_{\mathrm{NS}} := \hat h_{SM} - \hat h_{EC}\) and
\(\mathcal{T}_{\mathrm{SD}} := \hat h_{SSM} - \hat h_{SM}\). Add

\begin{itemize}
\tightlist
\item
  \textbf{B5 (survivor-effect homogeneity).} \(h^{S+} = h^{S}\): the
  mean conditioning shift is the same among survivors of \(k\) and of
  \(k+1\) further waves.
\end{itemize}

Under B1--B5 both statistics have expectation zero. B5 is needed because
\(\mathcal{T}_{\mathrm{SD}}\) compares estimands defined on different
survivor sets: without it,
\(\mathbb{E}[\mathcal{T}_{\mathrm{SD}}] = (h^{S+} - h^{S}) + \Delta_{SSM} - \Delta_{SM}\)
even when B1--B4 hold, and a rejection can reflect a change in effect
composition induced by the additional survival requirement rather than
any failure of the attrition restrictions. The statistics are therefore
compatibility checks among the maintained restrictions, not unique
classifiers of which one fails. Under model (M) of
Proposition~\ref{prp-fail}: with stationary state-dependent attrition,
\(\mathbb{E}[\mathcal{T}_{\mathrm{NS}}] = \Delta_{SM} - \eta = 0\) and
\(\mathbb{E}[\mathcal{T}_{\mathrm{SD}}] = (h^{S+} - h^{S}) - \Delta_{SM}\);
with non-stationary trait selection,
\(\mathbb{E}[\mathcal{T}_{\mathrm{NS}}] = \Delta_{SM}\) and
\(\mathbb{E}[\mathcal{T}_{\mathrm{SD}}] = (h^{S+} - h^{S}) + \Delta_{SSM} - \Delta_{SM}\),
which need not vanish. In (M), therefore, \(\mathcal{T}_{\mathrm{NS}}\)
has no power against stationary state dependence, and
\(\mathcal{T}_{\mathrm{SD}}\) can respond to non-stationarity as well as
to state dependence. Each statistic is a difference of subset means
within two independent cohorts; its variance is estimated from influence
functions that include the estimated group shares
(Section~\ref{sec-variance}), or by a person-level bootstrap within
cohort that recomputes every arm and every share.

\end{proposition}

\begin{proof}
Immediate from Proposition~\ref{prp-fail}: \(\mathcal{T}_{\mathrm{NS}}\)
estimates \(\Delta_{SM} - \eta\) and \(\mathcal{T}_{\mathrm{SD}}\)
estimates \((h^{S+} - h^{S}) + \Delta_{SSM} - \Delta_{SM}\); substitute
the values of \(\Delta_{SM}, \Delta_{SSM}, \eta\) under (M)(i) and
(M)(ii).
\end{proof}

A rejection indicates incompatibility among the maintained restrictions
B1--B5; it does not uniquely identify the source of the failure, and it
does not say which design is wrong. A mechanism that biases SM and EC
alike, such as stationary state dependence in (M), leaves
\(\mathcal{T}_{\mathrm{NS}}\) unaffected. The diagnostics therefore
complement, and do not replace, the overidentification tests of
Theorem~\ref{thm-tworef}(c), which use a battery of items instead of a
pair of designs.

\emph{Monte Carlo evidence.} Simulation 3 (Section~\ref{sec-ademp})
calibrates the designs to the size of the Japanese panel's 2011 episode
(\(k = 4\), \(n_{\text{old}} = 4{,}800\), \(n_{\text{new}} = 960\),
\(R = 1{,}000\) replications per regime) under six attrition regimes
generated within model (M). The constructive estimators have biases
below \(0.003\) in absolute value (outcome units) and 95\% coverage
between \(0.94\) and \(0.96\) under the regimes their restrictions
allow. Survival matching is biased under non-stationarity (\(-0.081\),
coverage \(0.65\)) and under state dependence (\(+0.104\), \(0.53\));
the entry-wave correction under state dependence (\(+0.103\), \(0.36\))
but not under non-stationarity (\(-0.002\), \(0.95\)); symmetric
matching under non-stationarity (\(-0.068\), \(0.76\)) but not under
state dependence (\(+0.002\), \(0.95\)), as
Proposition~\ref{prp-fail}(i)--(ii) implies for these processes. The
diagnostics reject at rates between \(0.039\) and \(0.056\) under the
three regimes in which B1--B4 hold; \(\mathcal{T}_{\mathrm{NS}}\)
rejects at \(0.62\) under non-stationarity and at \(0.040\) under state
dependence, as (M) implies, and \(\mathcal{T}_{\mathrm{SD}}\) at
\(0.99\) under state dependence and at \(0.08\) under non-stationarity.
Inverse-probability weighting on entry covariates, the adjustment that
most users would apply first, is unbiased only under selection on
observables; it is biased by \(-0.11\) to \(-0.31\) under every regime
with selection on the unobserved trait. The constructive designs
restrict the timing of selection; they do not require selection to be on
observables.

\section{Design Theory: What a Refreshment Schedule
Identifies}\label{sec-design}

The results so far take the panel's schedule as given. In this section
we treat the schedule as the object of choice. A design is the fielded
support \(\mathcal{S} \subset \{(e,t): t \ge e\}\) (which entry cohorts
are interviewed at which times) together with an assignment of
instrument and mode to cells. Under the additive cell-mean structure of
the companion identification paper (its Assumption M4),
\(\mu(e,t) = \alpha(t) + g(e) + \tau(s)\) with \(\tau(1) = 0\) and
\(g(e_1) = 0\), stacking the fielded cells gives a linear system
\(\mu = X_{\mathcal{S}}\,\theta\) with \(\theta = (\alpha, g, \tau)\)
and \(X_{\mathcal{S}}\) the design's incidence matrix. Here \(\mu(e,t)\)
is the mean report of cohort \(e\) at \(t\) in the population, or among
survivors after a selection correction such as (A2); with attrition, raw
survivor means add selection terms \(\beta\) that vary by cohort and
period and break the additive structure, which is the subject of
Section~\ref{sec-tworef}. Two parameter vectors generate the same cell
means iff their difference lies in
\(\mathcal{K} := \ker X_{\mathcal{S}}\); write \(\mathcal{K}_\tau\) for
the set of tenure components of its elements. Let \(d\), the
\emph{stride} of the schedule, be the greatest common divisor of the
spacings between entry cohorts. We use the \emph{increment graph} of the
companion identification paper: its vertices are the observed
increments, the tenures \(u\) at which some cohort is observed at
tenures \(u\) and \(u + 1\) in consecutive periods, and \(u\) and \(u'\)
are joined when two different cohorts are observed in the same two
consecutive periods with those increments. Condition C\(_d\) holds when
every observed tenure is reached from tenure one through observed
increments and the connected components of the increment graph are
exactly the residue classes of the observed increments modulo \(d\).

\begin{theorem}[What a design
identifies]\protect\hypertarget{thm-rank}{}\label{thm-rank}

Let \(\mathcal{S}\) be a fielded support with cell means following the
additive structure and the normalizations \(\tau(1) = 0\),
\(g(e_1) = 0\).

\begin{enumerate}
\def\labelenumi{(\roman{enumi})}
\item
  \emph{(Matrix form.)} A linear functional \(\lambda'\tau\) is
  identified from cell means if and only if \(\lambda\) annihilates
  \(\mathcal{K}_\tau\). If, in addition, known linear restrictions
  \(A\tau = a\) are imposed (anchors), the admissible directions are
  \(\mathcal{K}_\tau \cap \ker A\): \(\lambda'\tau\) is identified iff
  \(\lambda\) annihilates it, and \(\tau\) is point identified on the
  observed tenures iff \(\mathcal{K}_\tau \cap \ker A = \{0\}\). The
  rank of \(X_{\mathcal{S}}\), a basis of \(\mathcal{K}_\tau\), and the
  identification of any proposed functional can be computed before
  fielding.
\item
  \emph{(Directions that are never identified.)} For every
  \(m \in \mathbb{R}\) and every \(d\)-periodic \(\rho\) with
  \(\rho(1) = 0\), the vector \(h_\tau(s) = m(s-1) + \rho(s)\),
  \(h_g(e) = m(e - e_1)\),
  \(h_\alpha(t) = -m(t - e_1) - \rho(t - e_1 + 1)\) lies in
  \(\mathcal{K}\). Hence no design identifies a level or the linear
  trend of \(\tau\) from cell means alone, and when \(d > 1\) no
  functional whose weights fail to sum to zero within some residue class
  modulo \(d\) other than that of \(s = 1\) is identified; in
  particular, no ordinary second difference is.
\item
  \emph{(When those are all.)} If \(\mathcal{S}\) satisfies C\(_d\), the
  tenure components of the directions in (ii) span \(\mathcal{K}_\tau\),
  and \(\lambda'\tau\) is identified iff \(\sum_s \lambda_s(s - 1) = 0\)
  and \(\sum_{s \equiv r} \lambda_s = 0\) for every residue
  \(r \not\equiv 1 \pmod d\); the centred lag-\(d\) second differences
  \(\Delta_d^2\tau(s) = \tau(s+d) - 2\tau(s) + \tau(s-d)\) are then
  identified, and when \(d = 1\) every ordinary second difference is.
  Without C\(_d\), \(\mathcal{K}_\tau\) can be larger even when
  \(d = 1\) and all observed tenures are linked through tenures observed
  in common periods: with entries at waves 1, 3 and 4 observed through
  wave 4, the seven cells leave nine parameters with rank seven, and
  besides the affine direction the kernel contains
  \(h_\tau = (0, 0, 1, 1)\) on tenures \(1\)--\(4\),
  \(h_\alpha = (0, 0, -1, -1)\) on periods \(1\)--\(4\), and
  \(h_g = (0, 1, 1)\) on cohorts \(1, 3, 4\), which moves
  \(\tau(3) - 2\tau(2) + \tau(1)\). For staggered trapezoids
  \(\{(e,t): e \le t \le T\}\), the dimension of \(\mathcal{K}_\tau\)
  equals the number of components of the increment graph; C\(_d\) holds
  whenever the last cohort is observed for
  \(w := T - e_K \ge (e_2 - e_1) - d\) periods after its entry; with
  three cohorts the dimension is exactly \(\max\{d, (e_2 - e_1) - w\}\);
  and with more cohorts that expression is an upper bound.
\item
  \emph{(Mode assignment.)} A mode change applied uniformly to all cells
  at a time \(t^{*}\) is absorbed in \(\alpha(t^{*})\) and costs no
  identification. A mode assignment that is a function of tenure (for
  example, migrating long-tenure respondents to a cheaper mode first)
  adds a term \(m(s)\) to the cell mean that loads on the same
  coordinates as \(\tau(s)\); no functional separating \(m\) from
  \(\tau\) is identified, and the mode effect is confounded with the
  conditioning path itself.
\end{enumerate}

\end{theorem}

\begin{proof}
\leavevmode

\begin{enumerate}
\def\labelenumi{(\roman{enumi})}
\tightlist
\item
  Two parameter vectors generate the same cell means iff their
  difference lies in \(\mathcal{K}\); with anchors, the admissible
  differences are the elements of \(\mathcal{K}\) whose tenure
  components satisfy \(A h_\tau = 0\), and since \(A\) acts on \(\tau\)
  alone their tenure projection is \(\mathcal{K}_\tau \cap \ker A\). A
  linear functional is constant on an affine family iff it annihilates
  the family's direction space. (ii) At every fielded cell,
  \(h_\alpha(t) + h_g(e) + h_\tau(s) = m[-(t - e_1) + (e - e_1) + (s - 1)] + [\rho(t - e + 1) - \rho(t - e_1 + 1)] = 0\),
  because \(s - 1 = t - e\) and \(e \equiv e_1 \pmod d\). The functional
  statements follow by taking \(\rho = 0\), or \(m = 0\) with \(\rho\)
  the indicator of one residue class \(r \not\equiv 1\); for \(d > 1\)
  the weights \((1, -2, 1)\) of an ordinary second difference at
  \(s - 1, s, s + 1\) fall into more than one residue class, and at
  least one class other than that of tenure one has a nonzero sum. (iii)
  The spanning statement is Theorem 1(c) of the companion identification
  paper, and the trapezoid statements are the component count, the
  sufficient follow-up condition, the three-cohort formula and the upper
  bound of its Lemma 1; the functional characterization is (i) with
  \(\mathcal{K}_\tau\) spanned by \(s - 1\) and the residue-class
  indicators. In the example the increment graph has the components
  \(\{1, 3\}\) and \(\{2\}\), and the rank and the null vector are
  verified by exact arithmetic (Section~\ref{sec-ademp}). (iv) A uniform
  switch adds a constant to all cells at \(t^{*}\), which is a change in
  \(\alpha(t^{*})\). A tenure-dependent assignment adds \(m(s_{it})\), a
  function of the same coordinate as \(\tau(s_{it})\); the columns of
  \(X_{\mathcal{S}}\) that carry \(\tau(s)\) and \(m(s)\) are identical,
  so their sum is the only identified object. \(\square\)
\end{enumerate}

\end{proof}

A proposed schedule determines the matrix \(X_{\mathcal{S}}\), so its
identified space can be computed before any interview is conducted.
Anchors enter the same calculation. A negative-control battery at a
refreshment episode anchors the level of \(\tau\) at that episode's
tenure (Theorem~\ref{thm-tworef}). Under C\(_d\), a level anchor at a
non-entry tenure \(s > 1\) congruent to one modulo \(d\) removes the
affine direction and none of the periodic ones (at \(s = 1\) it would
restate the normalization \(\tau(1) = 0\)). A plateau restriction over
\(d + 1\) consecutive tenures removes all \(d\) directions, and an
external estimate of the linear part of \(\tau\) removes the affine
direction only. The Japanese panel of Section~\ref{sec-implementation}
(entries at waves 1, 5 and 13, observed through wave 19) has 41 cells,
39 free parameters and rank 35; it satisfies C\(_4\), so the identified
set of the path is four-dimensional. It identifies the eleven centred
lag-4 second differences \(\Delta_4^2\tau(s)\), \(s = 5, \dots, 15\),
and cross-class contrasts such as \(\Delta^2\tau(2) - \Delta^2\tau(6)\)
(a 14-dimensional space of functionals), but no ordinary second
difference. Level anchors at tenures 5, 9 and 13, all congruent to one
modulo four, which Theorem~\ref{thm-tworef}(a) would supply with an
age-matched second refreshment, leave the three periodic directions. So
do the anchors the panel actually supplies, a level at tenure 5 and the
increment \(\tau(13) - \tau(9)\). A plateau over five consecutive
tenures would leave none (Section~\ref{sec-ademp}).

\subsection{The stride of a refreshment schedule}\label{sec-stride}

For staggered trapezoids, Lemma 1 of the companion identification paper
gives a sufficient follow-up condition for the identified set of the
path to have the minimal dimension \(d\), and the exact dimension with
three cohorts. The proposition states them as three design rules.

\begin{proposition}[Refreshment spacing as an identification
decision]\protect\hypertarget{prp-stride}{}\label{prp-stride}

Let a panel enter cohorts at waves \(e_1 < \dots < e_K\) and be observed
at every wave through \(T\), with \(d\) the stride and \(w = T - e_K\).

\begin{enumerate}
\def\labelenumi{(\roman{enumi})}
\item
  \emph{(The identified set without anchors.)} Without anchors, the
  identified set of the path contains \(\tau + \{m(s-1) + \rho(s)\}\),
  \(\rho\) \(d\)-periodic with \(\rho(1) = 0\): a schedule with stride
  \(d > 1\) cannot distinguish, from cell means, a path from one that
  differs from it by a \(d\)-periodic pattern. Under C\(_d\) this is the
  whole identified set, and the identified curvature is spanned by the
  centred lag-\(d\) second differences and contrasts across residue
  classes.
\item
  \emph{(Coprime spacing.)} Adding a cohort at a wave \(e'\) with
  \(\gcd(d, e' - e_1) = 1\) makes the stride one. If in addition the
  last cohort of the enlarged schedule is followed for at least
  \((e_2 - e_1) - 1\) periods after its entry, where \(e_2\) is the
  second-earliest entry after the addition, the identified set is the
  affine line and every second difference of \(\tau\) is identified.
\item
  \emph{(Follow-up.)} With three cohorts, the identified set has
  dimension \(\max\{d, (e_2 - e_1) - w\}\): each additional period of
  follow-up of the last cohort removes one dimension until the floor
  \(d\) is reached. With more cohorts, \(w \ge (e_2 - e_1) - d\) remains
  sufficient for dimension \(d\) and \(\max\{d, (e_2 - e_1) - w\}\) is
  an upper bound, but shorter follow-up of the last cohort can suffice:
  with entries at waves 1, 5, 6 and 9 and final wave 9, the last cohort
  has no follow-up and the identified set is nevertheless the affine
  line.
\end{enumerate}

\end{proposition}

\begin{proof}
\leavevmode

\begin{enumerate}
\def\labelenumi{(\roman{enumi})}
\tightlist
\item
  is Theorem~\ref{thm-rank}(ii)--(iii). (ii): the stride after the
  addition is \(\gcd(d, e' - e_1)\), because every \(e_k - e_1\) is a
  multiple of \(d\); the dimension statement is the sufficient follow-up
  condition of the companion identification paper's Lemma 1 with stride
  one. (iii) is the three-cohort formula and the upper bound of that
  lemma; the four-cohort example is verified by exact rank computation
  (Section~\ref{sec-ademp}), which also shows that among the four- and
  five-cohort trapezoids with first entry at wave 1, later entries in
  \(\{2, \dots, 10\}\) and up to eight periods of follow-up, 176 attain
  dimension \(d\) with \(w < (e_2 - e_1) - d\). \(\square\)
\end{enumerate}

\end{proof}

Applied to the panel of Section~\ref{sec-implementation}: its entries at
waves 1, 5 and 13 give \(d = 4\). A fourth refreshment can restore
stride one only if it enters at a wave \(e'\) with
\(\gcd(4, e' - 1) = 1\), that is, at an even wave: entry at wave 20
gives \(\gcd(4, 19) = 1\) and entry at wave 22 gives
\(\gcd(4, 21) = 1\), whereas entry at wave 21 (\(\gcd(4, 20) = 4\))
leaves the stride at four and entry at wave 19 (\(\gcd(4, 18) = 2\))
halves it. Stride one is necessary for identifying ordinary curvature
but not sufficient: with the entry at wave 20 the identified set has
dimension 4, 3, 2 and 1 for final waves 20, 21, 22 and 23, so three
periods of follow-up after the new entry are needed; here the sufficient
condition of Proposition~\ref{prp-stride}(ii),
\(w \ge (5 - 1) - 1 = 3\), is also necessary. Refreshment timing is
usually decided on budgetary and coverage grounds. By the proposition,
it also determines what is identified. The same arithmetic applies to
any refreshed panel: a panel refreshed every two years on an annual
interview calendar has \(d = 2\) and identifies curvature only at lag
two, and a panel refreshed at irregular intervals whose gaps have no
common divisor can identify ordinary curvature once the follow-up is
sufficient. Rotation designs with interrupted participation, such as the
4--8--4 pattern of the US Current Population Survey, are treated
separately in the companion identification paper.
Table~\ref{tbl-strides} collects a few cases; every entry is an exact
rank computation.

\begin{table}[htbp]
\caption{Strides and identified-set dimensions of common refreshment
schedules (staggered trapezoids observed at every wave through \(T\)),
by exact rank computation. The last column applies when the dimension
equals \(d\).}\label{tbl-strides}
\begin{tabular}{@{}
  >{\raggedright\arraybackslash}p{(\linewidth - 8\tabcolsep) * \real{0.2600}}
  >{\raggedright\arraybackslash}p{(\linewidth - 8\tabcolsep) * \real{0.0700}}
  >{\raggedright\arraybackslash}p{(\linewidth - 8\tabcolsep) * \real{0.2800}}
  >{\raggedright\arraybackslash}p{(\linewidth - 8\tabcolsep) * \real{0.1500}}
  >{\raggedright\arraybackslash}p{(\linewidth - 8\tabcolsep) * \real{0.2400}}@{}}
\toprule\noalign{}
\begin{minipage}[b]{\linewidth}\raggedright
Schedule (entry waves)
\end{minipage} & \begin{minipage}[b]{\linewidth}\raggedright
stride \(d\)
\end{minipage} & \begin{minipage}[b]{\linewidth}\raggedright
dimension of the identified set
\end{minipage} & \begin{minipage}[b]{\linewidth}\raggedright
follow-up needed for dimension \(d\)
\end{minipage} & \begin{minipage}[b]{\linewidth}\raggedright
identified curvature (dimension \(d\))
\end{minipage} \\
\midrule\noalign{}
annual refreshment: \(1, 2, 3, \dots\) & 1 & 1 & none & all
\(\Delta^2\tau\) \\
biennial refreshment: \(1, 3, 5, \dots\) & 2 & 2 & none
(\(e_2 - e_1 = d\)) & \(\Delta_2^2\tau\) and cross-class contrasts \\
Japanese panel: \(1, 5, 13\) & 4 & 4 & none (\(e_2 - e_1 = d\)) &
\(\Delta_4^2\tau\) and cross-class contrasts \\
Japanese panel with a fourth entry at wave 20 & 1 & 4, 3, 2, 1 for
\(T = 20, 21, 22, 23\); 1 thereafter & \(T \ge 23\) & all
\(\Delta^2\tau\) once \(T \ge 23\) \\
Japanese panel with a fourth entry at wave 21 & 4 & 4 & none &
\(\Delta_4^2\tau\) and cross-class contrasts \\
\(1, 7, 10\) & 3 & \({\max\{3,\ 6 - (T - 10)\}}\) & \(T \ge 13\) &
\(\Delta_3^2\tau\) and cross-class contrasts once \(T \ge 13\) \\
a single refreshment at wave \(e_2\) & \(e_2 - 1\) & \(e_2 - 1\) & none
& \(\Delta_{e_2-1}^2\tau\) and cross-class contrasts \\
\bottomrule\noalign{}
\end{tabular}
\end{table}

In the last row, a panel with a single refreshment sample has stride
equal to the refreshment lag, so it identifies curvature only at that
lag, by comparing the entrants with the incumbents at tenures \(1\),
\(1 + (e_2 - 1)\) and \(1 + 2(e_2 - 1)\). Most of the applied literature
uses this design, and its practice of reporting level contrasts under
normalizations instead of shapes is consistent with the stride.

\begin{theorem}[The balance--information
trade-off]\protect\hypertarget{thm-balance}{}\label{thm-balance}

Suppose aggregates are published as tenure-balanced averages: at every
period the published statistic is
\(\bar y_t = \sum_s w_s\, \mu(t - s + 1, t)\) with weights
\(w_s \ge 0\), \(\sum_s w_s = 1\), constant across \(t\) (one-level
rotation balanced on time in sample). Then, under the additive structure
with cohort effects absorbed into the period path, \[
\bar y_t \;=\; \alpha(t) + \bar\tau, \qquad \bar\tau := \sum_s w_s\,\tau(s),
\] so that (i) the tenure component of the published level is the
constant \(\bar\tau\): it does not vary over time, so changes in the
published series are free of time-in-sample bias while its level carries
the bias \(\bar\tau\) in every period (balancing holds the bias constant
but does not eliminate it); and (ii) the published series depends on
\(\tau\) only through \(\bar\tau\), which is not separately identified
from the level of \(\alpha\), so no nonconstant functional of \(\tau\),
in particular no normalization-free functional (differences, second
differences, saturation tests), is identified from the published series
alone. The information is lost because only the balanced aggregate is
released, not because the design is balanced: the rotation-group
components from which the aggregate is formed identify what
Theorem~\ref{thm-rank} says the corresponding staggered support
identifies.

\end{theorem}

\begin{proof}
Substituting the additive structure into the balanced average gives
\(\bar y_t = \alpha(t) + \sum_s w_s g(t - s + 1) + \sum_s w_s \tau(s)\).
With cohort effects absorbed into the period path (or under sampling
equivalence \(g \equiv 0\)), the middle term is a function of \(t\)
alone and joins \(\alpha\). The map \(\tau \mapsto (\bar y_t)_t\) then
depends on \(\tau\) only through the scalar \(\bar\tau\), which is (i);
(ii) follows because a functional of \(\tau\) is identified from the
published series only if it is a function of the series, hence of
\(\bar\tau\), and because \(\alpha(t) + \bar\tau\) does not separate its
two summands.
\end{proof}

Part (i) restates, in the present notation, a property the rotation
literature established. Bailar (\citeproc{ref-bailar1975}{1975, 26}, eq.
(4.2)) shows that under a constant rotation group bias the level
estimates are biased while month-to-month changes are not, and concludes
that ``neither estimator is `best' for every statistic''
(\citeproc{ref-bailar1975}{Bailar 1975, 30}). Park, Kim, and Choi
(\citeproc{ref-parkkimchoi2001}{2001, 1488}, eq. (6)) establish the
constant level bias and its cancellation in changes for balanced
one-level rotation designs, and their component-level model
(\citeproc{ref-parkkimchoi2001}{Park, Kim, and Choi 2001, 1489}, eq.
(7)) supports estimating rotation group bias from the rotation-group
components. Part (ii) adds an explicit identification statement for the
case in which only the fixed-weight aggregate is released (the series
then depends on the time-in-sample path only through its weighted
average, which is confounded with the period level), together with its
connection with the rank calculation of Theorem~\ref{thm-rank}. It is
the design-side counterpart of the absorption result in the companion
identification paper (its Theorem 3): there, the orthogonality that
makes two-way fixed-effects estimates invariant to affine drift also
makes the affine drift unidentifiable from them. Releasing
disclosure-approved rotation-group summaries alongside the balanced
aggregate can preserve information that the aggregate loses. We leave
open the allocation problem of how to distribute a fixed budget of fresh
interviews across injection times to maximize the information on the
drift coefficient; numerical evaluation of the design information matrix
on lattice designs suggests that two injections at maximal temporal
spread dominate, but we have no proof and state no theorem.

\section{Implementation: JLPS}\label{sec-implementation}

The Japanese Life Course Panel Surveys (JLPS) provide the three-cohort
configuration of Theorem~\ref{thm-tworef}, with a negative-control
battery of 23 time-invariant childhood-circumstance items asked once at
entry in every cohort. The empirical contrasts below illustrate the
methods under the stated identifying restrictions; the accompanying
diagnostics assess compatibility among restrictions and do not validate
any one identifying assumption. The birth years of the 2011 entrants lie
within those of the incumbents (1966--1986): the birth-year supports
overlap, and the analysis additionally assumes cohort comparability as
specified in (A1), equal latent distributions on the common support,
which overlap does not establish and which differences in sampling
frames, nonresponse at entry, mode (Section~\ref{sec-increments}) or
coding could violate. The battery consists of the 23 items whose
question stems appear in all three cohorts' entry questionnaires (the
map from each item to its variable in the three entry questionnaires is
built from the data provider's variable labels; it is not redistributed
with the replication materials and is available from the author to
licensed users). We assume, and do not test, equality of response
categories and of administration across the entry questionnaires. The
2019 episode is the generational case of
Corollary~\ref{cor-generational}: entrants born 1987--1998 against
incumbents born 1966--1986, an empty overlap, verified in the data.

The estimand implemented at that episode is therefore the incumbent
increment: the 2007 cohort (tenure 13, \(n = 2{,}638\) at wave 13)
against the 2011 cohort (tenure 9, \(n = 619\)) on their common
birth-year support, with the same-pair negative-control contrast
\(\hat\beta^{N}_2(13) - \hat\beta^{N}_2(9)\) as the correction under
(A2). Across the 23 items, stratified by three-year birth band and sex
(Table~\ref{tbl-ncbattery}), the median absolute standardized contrast
is \(.033\) and the inverse-variance pooled value is \(-.035\). One item
in 23 is rejected at a Benjamini--Hochberg \(q < .10\) (the item
recording a telephone, mobile telephones included, in the childhood
home, whose availability varied with the era of childhood), and nine are
equivalent by two one-sided tests with margin \(.10\) at the 5\% level,
without multiplicity adjustment. The standard error \(.0095\) of that
pooled value, computed as \((\sum_k 1/\mathrm{se}_k^2)^{-1/2}\), treats
the 23 item contrasts as independent; because the items are answered by
the same respondents, we also computed a person-level bootstrap within
cohort (\(2{,}000\) replications, every item contrast, stratum weight
and pooling weight recomputed), which keeps their covariance. The
bootstrap standard error is \(.0186\), 2.0 times the independence value
(a variance ratio, or design effect, of \(3.84\); the mean correlation
between item contrasts across replications is \(.116\)), and the 95\%
percentile interval is \([-.072, .001]\), which includes zero. Without
the rejected item the pooled value is \(-.028\) (bootstrap standard
error \(.0185\)). Under (A2) with a common loading \(\Gamma\), increment
estimates that skip the correction would be biased by \(\Gamma\) times
the survivor-composition drift, estimated at about \(.035\,\Gamma\)
standard deviations in magnitude (\(.017\) to \(.052\) over
\(\Gamma \in [.5, 1.5]\)). Because the interval for the pooled value
includes zero, the battery does not establish that the survivor
composition has drifted; it limits the size of the drift at the 95\%
level, and at the far end of the interval the implied bias would be
\(.072\,\Gamma\) (\(.036\) to \(.108\) over the same range of
\(\Gamma\)). A small negative-control contrast measures the drift in
survivor composition on the battery; it removes the target's selection
bias only to the extent that (A2) and its loading hold.

In practice, Theorem~\ref{thm-sharp} yields one empirical check per
item. For a continuous item, \(\sup_y \hat p\, \hat q(y)/\hat f_2(y)\)
over a trimmed support, with \(\hat q\) the stayers' reported density,
is a plug-in estimate of \(\Psi(0)\), the constraint ratio of
Theorem~\ref{thm-bounds} for the candidate of no conditioning (equal to
\(\operatorname{ess\,sup}\bar\pi\) when the true shift is zero), and the
caveat of Remark~\ref{rem-trim} applies to it. At the population level,
\(\Psi(0) > 1\) excludes the no-conditioning candidate. Its sample
analogue is an exploratory diagnostic until sampling uncertainty has
been calibrated. At the population level, likewise, a value of
\(\Psi(0)\) well below one, together with a log-density that does not
vary too fast, implies that shifts near zero belong to the location set,
because \(\Psi(h') \le \Psi(0)\,e^{L|h'|}\) when \(\log f_2\) is
\(L\)-Lipschitz, the argument of Corollary~\ref{cor-funnel}(iii); a
value near one does not imply a tight set, because \(\Psi(0)\) does not
restrict \(\Psi\) elsewhere. For a discrete item, the mass ratio
\(\max_a \hat p\,\hat Q(\{a\})/\hat F_2(\{a\})\) of
Corollary~\ref{cor-discrete} needs neither trimming nor slack, and its
population version is exact; its sample analogue for the Japanese panel
is reported below.

The constructive designs of Section~\ref{sec-constructive} have been
applied to the same panel at the 2011 episode, with the fresh cohort's
follow-up supplying \(k = 4\) and \(k + 1 = 5\) matched waves and the
2007 cohort's entry wave supplying the entry-wave correction
(computations in the R package \texttt{panelcond}; aggregate outputs
only). The arms are 2,797 continuing survivors of five consecutive waves
against 963 fresh entrants, of whom 574 survive four further waves and
540 five. Item-level inference is by a joint person bootstrap with 500
replications that resamples persons within cohort and recomputes every
arm and every group share, so it does not rely on the analytic formula
of Section~\ref{sec-variance}.

The item universe is declared variable by variable in a specification
table checked against the 2011 and 2007 questionnaires
(\texttt{kit/R/15\_item\_scale.csv}). Of the 540 wave-5 variables, 53
are excluded (open-ended and after-coded classifications, dates and the
components of clock times and durations, duplicate recodes, and a
spouse-income bracket whose top code means that there is no spouse).
Each of the others is given a scale, with codes outside the scale
(``other'', ``not fixed'', ``public sector'', ``no parent at the time'')
set to not applicable. The fifteen nominal single-choice questions,
among them party identification, occupation, work status, housing
tenure, marital status and relationship status, enter as 90 zero/one
indicators, one per category, and the four clock times and the duration
of the current relationship are rebuilt from their components by one
rule at both waves. This gives 523 columns: 284 binary, 102 ordered and
41 continuous items, the 90 indicators and 6 derived items.

The counts below cover the 490 columns for which a mean contrast between
the arms is meaningful. We exclude the other 33 from every count and
every multiplicity adjustment: 31, in the marriage-history block,
because their item-nonresponse rate differs between continuing
respondents and entrants by more than fifteen percentage points, one as
a follow-up to such a question, and one sensitivity variant of the
bedtime item. The flag is a screening rule on the observed nonresponse
rates; the questionnaire filters behind the difference have not been
verified against the instruments. The detection counts below are not
sensitive to the rule. Moving the threshold to ten or to twenty points
leaves the flagged set, and with it every count, unchanged. Dropping the
routing exclusion altogether returns the 32 columns to the counts and
changes the flagged numbers by at most one (31, 21, 14, 29 and 22 in the
order below), with the same three items flagged by all five designs.

The entry-wave correction exists for the 320 columns whose 2007
counterpart shows the same codes after the recodes of the table (a party
listed in 2007 but not in 2011 and one listed in 2011 but not in 2007,
and a 2007 after-code for agriculture, are folded into ``other'' at both
waves). Of these, 16 are among the routing-flagged columns and 35 have a
2007 version that differs in reference period, options or format; their
entry-wave estimators are computed but stay outside the counts, and
admitting the latter would raise the two entry-wave counts below by two
each (\texttt{kit/KIT\_README.md}). The remaining 269 columns enter the
main entry-wave analysis. Excluding the sensitivity variant leaves 268
columns in the detection counts: 181 columns whose 2007 question has the
same wording, response options and eligibility criteria, allowing minor
differences of layout (classes A and B), and 87 asked in 2007 only of a
subgroup (class C).

Two entry-wave corrections are computed, and their respondent sets are
stated because the restrictions are stated on them. Let \(T\) be the
survivors with a substantive answer at the comparison wave (the
continuing arm for the entry-wave corrections and survival matching,
symmetric matching additionally requiring a response at \(t + 1\); for
an item asked only of a subgroup, its current members \(G_t\) by
construction). Let \(E\) be the entrants of the continuing cohort who
were eligible for the item at entry, \(G_e\), and gave a substantive
entry answer; here \(G_e\) is the set reached by the entry
questionnaire's routing: for a class-C item, the 2007 subgroup named in
the specification table (the employed, employees, the married,
respondents with a partner or parents); for an item asked of everyone,
the whole cohort. Let \(P = T \cap E\), and let \(F\) be the fresh
entrants with a substantive answer. The unadjusted correction is
\(\bar Y_T - (\bar Y^{e}_P - \bar Y^{e}_E) - \bar Y^{*}_F\), with
\(\bar Y^{e}\) the mean entry answer: the sample form of
Equation~\ref{eq-ec} with the answerers in place of the cohorts. Its
target is \(h^{S}\) for the survivors in \(T\), and it equals that
target when the fresh answerers stand for \(G_t\) and the entry
answerers for \(G_e\) and for \(T \cap G_e\) (representative item
completion, the requirement also placed on the mass diagnostic below),
and
\(\mathbb{E}[Y^{*}_t \mid T] - \mathbb{E}[Y^{*}_t \mid G_t] = \mathbb{E}[Y_e \mid T \cap G_e] - \mathbb{E}[Y_e \mid G_e]\).
When the eligible population is unchanged across waves, the bridge
reduces to B4 with the response indicator included in the conditioning
set. Classes A and B describe questionnaire comparability and are not
sufficient for unchanged eligibility: membership in groups such as the
non-working or the unmarried can change even under an identical routing
rule. The eligible population changes between the waves for every
class-C item by construction, since its 2007 filter differs from the
2011 one, and for any item routed on a status that can change, such as
the follow-up asked of those not working or the reasons for remaining
single. In such cases a person who marries or stops working between the
waves contributes to one difference and not to the other. The condition
is then a transport of the selection bias across the two eligible
populations; B4 within one fixed population does not imply it, and
matching wording and coding do not establish it. Two artificial
populations with no conditioning illustrate this. In one, eligibility at
the comparison wave is the outcome itself; in the other, the same
routing rule is applied at both waves but the status it is applied to
changes. Both give a correction of \(-.5\) with B4 intact. Item
completion that depends on the answer biases the two classes alike:
artificial populations with no conditioning give \(-.5\) when the entry
answer is missing according to its value and \(+.17\) when the fresh
answer is (\texttt{kit/R/15\_test\_ec\_adj.R}).

The covariate-standardized correction fits three least-squares
regressions on the entry covariates (sex, birth year and education): of
the entry answer among \(P\) and among \(E\), and of the fresh answer
among \(F\). It is
\(\bar Y_T - \overline{(m_P - m_E)}_T - \overline{m_F}_T\), every
prediction averaged over the covariate distribution of \(T\): the
unadjusted correction with its two comparison terms replaced by
regression predictions integrated over one population, the sample form
of the standardized version in Proposition~\ref{prp-ec}. Sufficient
conditions for identification are the conditional bridge
\(\delta_e(t \mid x) = \delta_e(e \mid x)\) and representative item
completion within covariate values, with the response indicators
included in the conditioning sets. The regression implementation further
assumes correctly specified conditional means. Where the covariate
support of \(T\) extends beyond that of \(P\), the entry regression is
extrapolated linearly, and its interpretation at those values rests on
the maintained extrapolation model. An earlier version of this analysis
averaged the entry residual over \(P\) and the fresh prediction over
\(T\), two covariate distributions whenever some survivors lack an entry
answer or, for class C, were not eligible at entry. Under that rule
equal conditional selection terms do not give a zero bias (a
sixteen-type population with no conditioning and equal conditional terms
gives \(1/12\)), and the standardized counts below supersede those of
that version. The class-C estimates are reported under the stated
transport restriction, and the entry-wave counts and the diagnostic
\(\mathcal{T}_{\mathrm{NS}}\) are given both for classes A--C and for
classes A and B alone (\texttt{kit/KIT\_README.md}, §3c). The A--B
analysis excludes the documented class-C questionnaire differences, but
the stated bridge and item-completion conditions remain necessary for
any included item whose eligible population changes; it is a sensitivity
comparison, not a fixed-population design.

The counts are exploratory: the Benjamini--Hochberg adjustment is
applied within the stated families, and the screening that defines the
families is not sufficient for false-discovery control.

The designs disagree. At a Benjamini--Hochberg-adjusted threshold of
\(q < .10\), the naive contrast flags 30 of 490 columns; survival
matching and symmetric matching flag 21 and 13 of 489, one item lacking
variation in the matched arms. The unadjusted and covariate-standardized
entry-wave corrections flag 29 and 22 of the 268 columns in the main
entry-wave analysis. Restricting the analysis to the 181 columns of
classes A and B and recomputing the adjustment within that family gives
28 and 19 flags; in the full 268-column family the class-C columns
account for 5 and 3 of the 29 and 22, contributions that are not
additive to the separately adjusted A--B counts. On the common set of
268 columns the five counts are 18, 16, 9, 29 and 22. Three items are
flagged by all five designs: subjective social position, owner-occupancy
of a detached house, and anxiety about married life as a reason for
remaining single. The last is an item asked of the unmarried, whose
eligible population changes between the waves, so that its two
entry-wave flags rest on the transport restriction stated above.

Using bootstrap standard errors, \(\mathcal{T}_{\mathrm{NS}}\) rejects
at the 5\% level for 39 of the 268 columns on which it exists (14.6\%;
it needs the entry wave; 24 of the 181 columns of classes A and B,
13.3\%), \(\mathcal{T}_{\mathrm{SD}}\) for 43 of 489 columns (8.8\%),
and \(\mathcal{T}_{\mathrm{SD}}\) for 73 of 409 item-nonresponse
indicators (17.8\%). The last share overstates the breadth of the
evidence: the items of a question grid are answered or skipped together,
and 50 of the 73 rejections come from four grids. Given B5, the
rejections are read as incompatibilities among the maintained
restrictions rather than as counts of failures of a named kind.

We also compare two person-level response-style composites. The first is
the share of a respondent's answers that fall in an extreme category,
over a list of 46 rating-scale items: agreement, satisfaction and
evaluation scales with verbal anchors, listed by variable name with
their number of categories in the archive
(\texttt{kit/R/15\_style\_items.csv}). The second is the share that fall
in the neutral middle category, over the 28 of these items that have
one. All five designs are computed on the 42 items of the list (25 with
a neutral middle) that were asked with the same codes at entry, so that
the entry-wave and comparison-wave composites, and the five designs, use
one battery. Frequency scales, quantity bands, classifications and
nominal codes are outside it. Extreme-category use falls by .22 to .31
standard deviations and midpoint use rises by .19 to .21 across the five
designs. Both contrasts retain their direction, and broadly similar
magnitudes, across the alternative item sets examined
(\texttt{kit/KIT\_README.md}). Over all 46 rating items, for the three
same-wave designs, extreme-category use falls by .23 to .32 and midpoint
use rises by .20 to .23. Over the 17 agreement items, which are asked of
every respondent, the fall is .20 to .28 and the rise .15 to .17. (The
composites are shares of the items a respondent answered; the
job-characteristics items and the job and marriage satisfaction items
reach only respondents with a job or a spouse, whose shares differ
between the arms.) Over the 21 five-point agreement and satisfaction
scales, the fall is .21 to .29 and the rise .15 to .17, the smallest
midpoint contrasts. With the 24 frequency scales that were asked with
the same codes at entry added to the common battery, extreme-category
use falls by .22 to .29; midpoint use is unchanged, since the frequency
scales have no neutral middle. The rule of an earlier version of this
analysis, which admitted every item with four to seven consecutively
numbered codes, among them frequency scales, occupational rank, smoking
and drinking, and education, gives on the present universe a fall of .21
to .26 and a rise of .20 to .22.

For employment the designs differ: after the entry-wave correction,
continuing respondents are more often employed than entrants, by 4.9
percentage points, a difference flagged by that correction (\(q = .03\),
Benjamini--Hochberg over the 268 columns in the main entry-wave
analysis) and by survival matching (\(q = .05\)), but not after
covariate standardization (\(q = .14\)) or without correction
(\(q = .38\)).

The sample analogue of the mass-domination criterion of
Corollary~\ref{cor-discrete},
\(\max_a \hat p\,\hat Q(\{a\})/\hat F_2(\{a\})\), uses no trimming and
no slack, and it is a complete-case quantity. Here \(\hat Q\) is the
distribution of the stayers' substantive answers and \(\hat F_2\) that
of the fresh entrants who gave one. Its categories are the union of
those observed among the stayers and in the fresh cohort; a codebook
category observed in neither contributes \(0 \le 0\) and is not a
category here. For an item asked only of a subgroup, such as respondents
with a spouse, the comparison is within the subgroup, so \(\hat p\) is
the share of the continuing cohort that survived and answered, divided
by the share of fresh entrants who reached the question (one for an item
asked of everyone). Reading these quantities through
Corollary~\ref{cor-discrete} requires two conditions beyond refreshment
validity. The reach rate must stand for eligibility, which requires the
same eligible share in the two cohorts and routing that can be
distinguished from ordinary nonresponse; in these data an item
nonresponse recorded as missing, and not with a no-answer code, cannot
be distinguished from routing. For this diagnostic, incumbent survivors
without a substantive item response are treated as unobserved and
excluded from the numerator defining \(\hat p\). And the fresh entrants
who answer must be representative of the eligible subgroup,
\(\mathcal{L}(Y^{*} \mid G, A_f = 1) = \mathcal{L}(Y^{*} \mid G)\),
writing \(G\) for eligibility and \(A_f\) for a substantive fresh
answer; refreshment validity concerns the sample, not its answerers, and
does not imply this. Das, Toepoel, and van Soest
(\citeproc{ref-dasetal2011}{2011, 37}) make the corresponding assumption
explicit for their design: no nonresponse in the refreshment sample, or
nonresponse missing at random. Without it, fresh item nonresponse alone
can produce a ratio above one: if half the population has outcome one,
survival is .8 and independent of the outcome, every survivor answers,
and among fresh entrants everyone with outcome zero answers but only a
fifth of those with outcome one, then the complete-case ratio at outcome
one is \(.8 \times .5 / (1/6) = 2.4\) with no conditioning at all.

A category reported by stayers but by no fresh entrant makes the sample
ratio infinite. Such a category is excluded from the finite maxima but
recorded. Among the 485 columns of the universe with two to nine
observed categories (the zero/one indicators of the nominal questions
included, every one of them assessed with its question's reach), 10 have
such a category, one each, none of them an indicator. Nine of the ten
are items asked only of a subgroup, six of them on how a respondent with
a fiancé(e) or partner met that person; two belong to the
marriage-history block flagged above, so that for them the comparison
may not be between like groups. The finite ratio exceeds one for 11
columns. The largest are two indicators of nominal questions whose
category has fewer than ten fresh entrants, support for one minor party
(\(1.955\)) and a spouse working as a family employee (\(1.581\)),
followed by two rare events, whether the respondent's mother died in the
past year (\(1.369\)) and whether the respondent expects to have taken
over the family business in ten years' time (\(1.270\)). For 9 of the 11
the exceedance occurs only in categories with fewer than ten fresh
entrants; the two items that exceed one in a category with at least ten
are how often the spouse prepares meals (\(1.192\)) and whether the
respondent would ask siblings for help in finding work (\(1.020\)).

We also assess compatibility with the identity map while allowing
fresh-sample missing outcomes to occupy any category. Let \(l_a\) be the
stayers who answered \(a\) and \(r_a\) the fresh entrants who answered
\(a\), both per eligible member, and let \(r_M\) be the fresh entrants
who reached the item but gave no substantive answer. A common latent
distribution consistent with no conditioning then exists if and only if
\(\sum_a \max\{l_a - r_a, 0\} \le r_M\) (the eligible populations and
the categories being taken as common); in the example the required mass
is \(.3\) and the available mass \(.4\).

On these data the condition holds for all 21 flagged columns (the 10
with a positive/zero category and the 11 finite exceedances), the two
supported exceedances included. The fresh missing mass is small, with a
median of 1.8\% over the 485 columns, and the mass needed to cover the
stayers under the identity map is smaller still, at most \(.0065\) among
the flagged columns. The reach rates of the two arms differ by less than
ten percentage points over the 485 columns, the largest gaps being on
items asked of respondents in owner-occupied housing or of the
unmarried; since attrition may select on the routing variable, this gap
is descriptive and does not test the equal-eligibility requirement. The
condition is permissive, since it lets the missing mass fall wherever it
helps, so passing it shows that a flag is within what fresh item
nonresponse alone could produce, not that it is an artefact of it. For
an item with nominal codes the inequality is the implication of no
conditioning (the identity map), without the monotone-map interpretation
of Corollary~\ref{cor-discrete}. At the population level, no
conditioning implies the inequality for the distribution of all eligible
fresh entrants under every attrition process. A violation is therefore
evidence against the maintained model that no attrition process removes
(Corollary~\ref{cor-discrete}); that model comprises refreshment
validity, comparable coding, deterministic monotone measurement for an
ordered item whose categories the stayers all report, and, for the
complete-case ratio, the two item-response requirements above. We report
these ratios as descriptive complete-case diagnostics. Fresh-sample item
nonresponse can generate a flag even without conditioning; interpreting
the ratios through Corollary~\ref{cor-discrete} additionally requires
representative item completion within comparable eligible subgroups, and
sampling calibration alone does not establish that requirement. Nor does
an infinite or large sample ratio show that a population mass is zero or
that the population inequality fails. The criterion is a maximum of
estimated ratios whose sampling error is not calibrated here. Rare
categories inflate it, since the denominator is the fresh cohort's share
of the category, estimated from at most 963 entrants and from far fewer
for subgroup items, whose \(\hat p\) also carries the error of the
estimated subgroup share. No multiplicity adjustment is applied across
items. A one-sided confidence bound on the maximal ratio, with
multiplicity control across items, is the appropriate inferential
version and is left to future work.

\section{Discussion}\label{discussion}

We have shown what a refreshment sample can and cannot do. On its own it
cannot separate conditioning from attrition selection when the location
set is nondegenerate: Theorem~\ref{thm-sharp} gives the exact residual
ambiguity, which, on the latent-outcome support, vanishes without
attrition and, under outcome-independent attrition, whenever the
refreshment density has tails thinner than exponential
(Corollary~\ref{cor-independent}). Within a cohort and without
restrictions linking waves, the panel's own history does not narrow it
(Theorem~\ref{thm-multiwave}), so, within the maintained model,
additional identifying restrictions must exclude otherwise feasible
completions. A refreshment sample can, however, bound the ambiguity by
an observable density-ratio criterion: the Horowitz--Manski interval for
an unrestricted map (Corollary~\ref{cor-meanbounds}), the location set
of Theorem~\ref{thm-bounds}, and the mass inequality for an item whose
categories are all reported (Corollary~\ref{cor-discrete}). Covariates
shrink the ambiguity (Proposition~\ref{prp-covariates}), as do negative
controls under a transport restriction. With a negative-control battery
at each refreshment and the transport restriction (A2) with a known
loading, the ambiguity is removed at the tenures observed at refreshment
episodes. The no-conditioning restriction that the refreshment-sample
attrition literature maintains (\citeproc{ref-hirano2001}{Hirano et al.
2001}; \citeproc{ref-deng2013}{Deng et al. 2013};
\citeproc{ref-franguridihahn2026}{Franguridi et al. 2026, 4}) selects
one point of a set whose extent the data reveal, and whose plug-in
width, as Remark~\ref{rem-trim} shows, depends on where the support is
trimmed.

Franguridi and Kosenkova (\citeproc{ref-franguridi2026}{2026, n. 1},
p.~5) observe that without restrictions on the attrition process the
observed marginals leave the joint distribution of the two waves largely
unrestricted, so that bounds on their structural parameter are in most
cases uninformative (the stayers' observed joint law does restrict it,
since an admissible joint must dominate \(p\) times that law). The
object here is different, the conditioning map and not the joint
distribution, and Theorem~\ref{thm-sharp} shows that the refreshment
marginal does restrict it, to the extent measured by the density-ratio
criterion of Corollary~\ref{cor-funnel}. The same point applies to the
test of selection on observables of Franguridi and Kapteyn
(\citeproc{ref-franguridikapteyn2026}{2026, 3, 13}, and Appendix B,
p.~20), which compares the inverse-probability-weighted second-period
distribution of the stayers with that of the refreshment sample and
rejects when the two differ. Conditioning can make them differ even when
selection is ignorable, so a rejection need not indicate nonignorable
selection. Conversely, under credible selection on observables and the
maintained common monotone-map model, the test is a test of no
conditioning on the latent-outcome support (a distributional comparison
need not detect a stochastic or nonmonotone change of responses that
preserves the marginal distribution). Their application rejects for a
set of variables that includes a cognition score corrected beforehand
for practice effects (\citeproc{ref-franguridikapteyn2026}{Franguridi
and Kapteyn 2026, 12} and Appendix C, pp.~25--26), a conditioning
correction applied to one variable before the test.

Warren and Halpern-Manners (\citeproc{ref-warren2012}{2012, 521--22})
suspect that, unlike attrition, conditioning leaves the data
``irredeemably biased'' once it has occurred. The results here qualify
this view: with a refreshment sample the conditioning map is
set-identified without further assumptions, and each design restriction
selects a member of the set.

The results have three consequences for practice.

Das, Toepoel, and van Soest (\citeproc{ref-dasetal2011}{2011}, main text
and online Appendix 1) show that the same entry-wave functional
identifies population or survivor conditioning effects under different
stationarity assumptions. The general-mean formulation here (B4 for the
survivor effect) preserves that distinction, which is important when
conditioning effects are heterogeneous: the two readings coincide under
homogeneous effects and differ by the composition term
\(\mathbb{E}[\tau_i \mid S] - \mathbb{E}[\tau_i]\) otherwise (Simulation
3), so reports of entry-wave-corrected estimates should state which
target they intend.

The diagnostics are compatibility checks among the maintained
restrictions B1--B5, not classifiers. When non-stationary and
state-dependent attrition are both present, the estimator with the
smallest bias in Simulation 3 is survival matching, whose restriction
both failures violate, because the two failures bias it in opposite
directions; and under stationary state dependence in model (M),
\(\mathcal{T}_{\mathrm{NS}}\) has no power at all. A rejection is
evidence that the set of restrictions is inconsistent, to be followed by
the overidentification tests of Theorem~\ref{thm-tworef}(c).

The identified set of the conditioning path on a staggered trapezoid has
dimension at least the stride of the refreshment schedule, and more when
the follow-up of the last cohort is short. On irregular supports it can
be larger or smaller, and the rank computation of
Theorem~\ref{thm-rank}, not a closed-form condition, is the general
check.

The deterministic monotone-map model serves continuous items. For items
whose categories are all reported it is a specification test
(Corollary~\ref{cor-discrete}), and shifts in the use of extreme and
middle categories of the kind examined in
Section~\ref{sec-implementation}, if they reflect conditioning, are
incompatible with the deterministic weakly increasing-map model;
stochastic partial merging of categories is one possible explanation.
Neither the application of the binary analysis of Das, Toepoel, and van
Soest (\citeproc{ref-dasetal2011}{2011}) to each event \(\{Y > a\}\),
which they suggest for non-binary outcomes {[}p.~36{]}, nor a joint
treatment of ordered categories is attempted here.

For panels with refreshment designs we recommend the following. Report
the identified set, with the trim and slack of any plug-in version,
instead of assuming it away. Report survivor and population targets
separately when effects may be heterogeneous. Design panels with at
least two refreshments, an entry-wave negative-control battery and,
where possible, a repeated tenure across refreshment episodes, which in
the location family adds a direct test of stationary selection. Compute
the rank of a proposed schedule before fielding it. Choose refreshment
waves whose spacings have no common divisor, and follow the last
refreshment cohort long enough (for the panel of
Section~\ref{sec-implementation}, an entry at an even wave followed for
at least three waves). Release disclosure-approved rotation-group
summaries alongside balanced aggregates. Each of these is an
identification decision, although it is usually taken on other grounds.

\section{Appendix: Simulation Evidence (ADEMP)}\label{sec-ademp}

Three simulation studies, reported in the ADEMP format, examine the
finite-sample content of the results, and three deterministic scripts
verify the worked examples, the counterexamples, and the design-rank
statements. Code and outputs are in the \texttt{sims/} folder of the
replication archive: the scripts \texttt{sim1\_bounds.R},
\texttt{sim2\_multicohort.R} and \texttt{sim3\_designs.R}, and the
checks \texttt{check\_funnel\_examples.R},
\texttt{check\_multiwave\_completion.R} and
\texttt{check\_design\_rank.py}. No confidential data are involved.
\texttt{check\_manuscript\_values.R}, also in \texttt{sims/}, asserts
every simulation value quoted in the text and tables against a named
cell of the output files, with the estimand and units of each. The
data-generating processes and estimands below are those implemented in
the archived scripts.

\textbf{Simulation 1 (Theorem~\ref{thm-sharp},
Corollary~\ref{cor-funnel}, Theorem~\ref{thm-bounds}: a relaxed plug-in
procedure).} \emph{Data-generating mechanism}: bivariate Gaussian
\((Y_1, Y_2^{*})\) with unit variances and \(\rho = .5\); logistic
outcome selection \(\pi = \Lambda(a + \beta Y_2^{*})\) with
\(\beta \in \{0, .3, .6\}\) and \(a\) tuned to retention
\(p \in \{.6, .8\}\); location conditioning \(h \in \{0, .2\}\);
\(n = 4{,}000\) plus a refreshment sample of \(2{,}000\); 200
replications per cell (in two cells one replication produced an empty
plug-in set and is excluded from the set summaries). \emph{Estimand}:
\(h\). \emph{Methods}: a relaxed plug-in set
\(\{h' : \sup_{y} \hat p\, \hat q(y + h')/\hat f_2(y) \le 1 + \kappa\}\)
with kernel densities, the supremum taken over the support trimmed at
\(10\%\) of the refreshment density's mode, slack \(\kappa = .10\), and
candidates \(h' \in [-1, 1.5]\) on a grid of step \(.01\) that contains
the true values; the naive comparator is quantile alignment ignoring
attrition. As population benchmarks the script computes, without
sampling, the identified set of Theorem~\ref{thm-bounds} for the same
process, \(h + [0, \Delta^{*}]\) (a point when \(\beta = 0\)), and the
population analogue of the relaxed procedure (the same trim and slack
applied to the true densities). \emph{Performance}: The naive
estimator's bias is essentially zero at \(\beta = 0\) (at most \(.004\)
in absolute value), \(.056\)--\(.058\) at \(\beta = .3\) and
\(.114\)--\(.116\) at \(\beta = .6\) when \(p = .8\), and
\(.116\)--\(.118\) and \(.218\)--\(.220\) when \(p = .6\): raising
retention from \(.6\) to \(.8\) halves it. The population identified
sets have widths \(0\), \(.145\) and \(.262\) at \(p = .8\) and \(0\),
\(.260\) and \(.489\) at \(p = .6\) for \(\beta = 0, .3, .6\). The
relaxed plug-in sets are wider (mean widths \(.25\)--\(.32\) at
\(p = .8\) and \(.58\)--\(.68\) at \(p = .6\)), and positive at
\(\beta = 0\), where the identified set is a point. Their widths track
those of the population analogue of the relaxed procedure
(\(.31\)--\(.38\) at \(p = .8\), \(.61\)--\(.71\) at \(p = .6\)), so
they are properties of the trim and the slack, not of the identified set
(Remark~\ref{rem-trim}), and Figure~\ref{fig-funnel} reports them as
such. Coverage of the true \(h\) by the relaxed plug-in set, counting an
endpoint equal to the truth as covering it, is \(.930\)--\(1\) at
\(p = .6\) and \(.760\)--\(.965\) at \(p = .8\), falling with \(\beta\);
a uniformly valid version requires inference-aware calibration of the
slack, in the spirit of intersection-bounds methods, which we leave to
future work.

\begin{figure}

\centering{

\pandocbounded{\includegraphics[keepaspectratio]{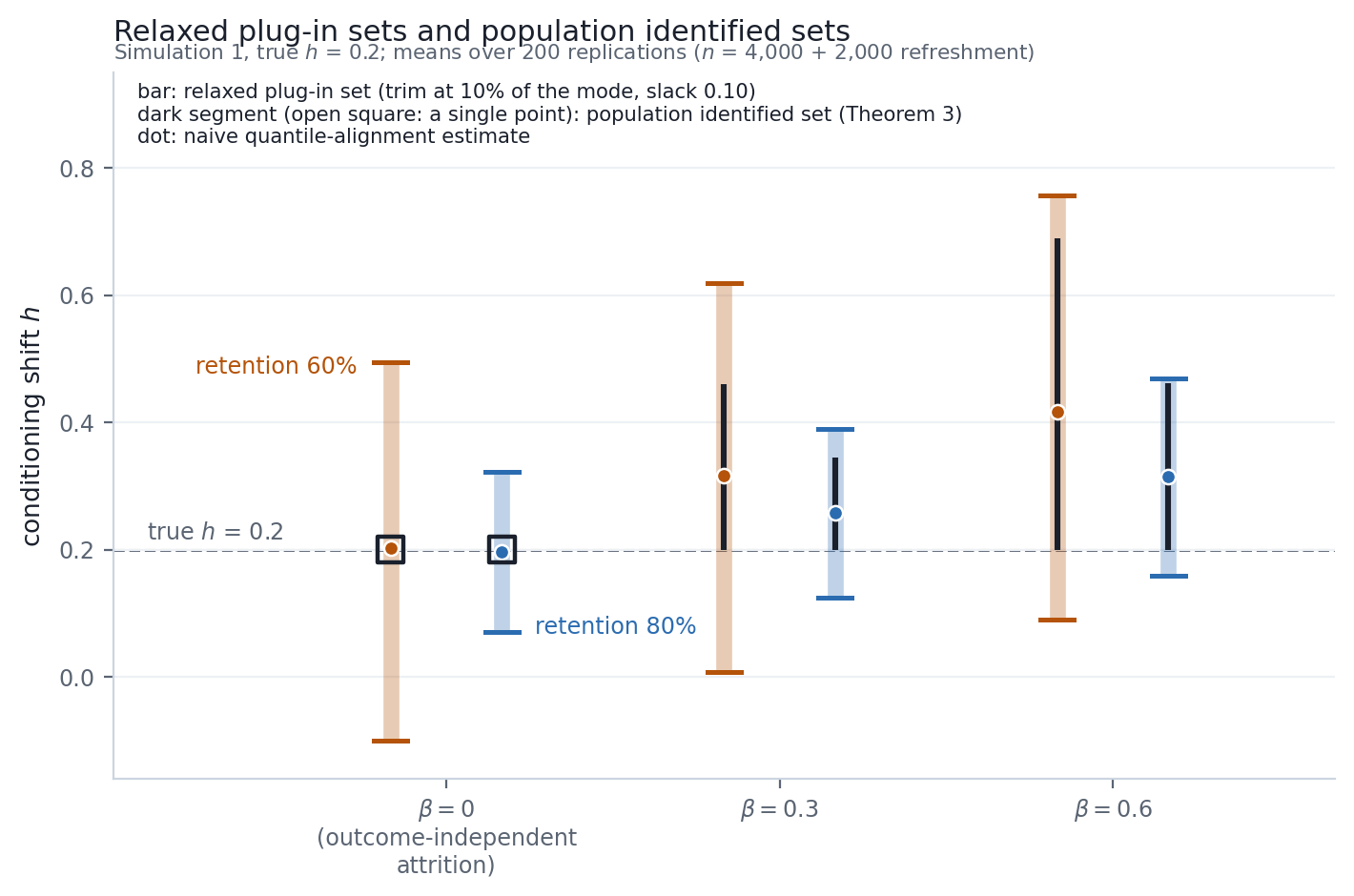}}

}

\caption{\label{fig-funnel}Relaxed plug-in sets and population
identified sets, Simulation 1, true \(h = 0.2\). Bars: mean over 200
replications of the relaxed plug-in set (support trimmed at 10\% of the
refreshment density's mode, slack 0.10). Dark segments: the population
identified set of Theorem~\ref{thm-bounds} for the same data-generating
process (open squares: a single point). Dots: mean naive
quantile-alignment estimate. The plug-in widths reflect the trim and the
slack; the population sets widen with selection strength and narrow with
retention.}

\end{figure}%

\textbf{Simulation 2 (Theorem~\ref{thm-tworef},
Corollary~\ref{cor-generational}).} \emph{Data-generating mechanism}:
three cohorts entering at waves 1, 5 and 13 (\(n = 3{,}000\) each),
compared at wave 13; a latent trait \(u_i \sim N(0,1)\) and per-wave
retention \(\Lambda(1.6 + 0.6u_i)\), so that the first cohort has passed
twelve response decisions and the second eight; target item
\(y = g_e + \tau(s) + u + e\) with \(e \sim N(0,1)\) and
\(\tau(s) = .15(1 - e^{-(s-1)/4})\); a \(K = 23\) negative-control
battery measured at entry, \(x_k = .6u + \eta_k\); a generational
variant adds a cohort effect \(g = .35\) to the wave-13 entrants' target
item; a contaminated variant adds recall drift \(.15\) to three battery
items of the first cohort; 300 replications. \emph{Estimands}, in units
of the target's population standard deviation \(\sqrt 2\): the incumbent
increment
\((\tau(13) - \tau(9))/\sqrt2 = .15(e^{-2} - e^{-3})/\sqrt2 = .0091\)
and the level contrast \((\tau(13) - \tau(1))/\sqrt2 = .1008\).
\emph{Methods}: contrasts of survivors' raw means, divided by
\(\sqrt 2\); negative-control corrections subtract \(\Gamma\) times the
mean raw battery contrast of the same pair of cohorts, with the true raw
loading \(\Gamma = 1/.6\) (exact (A2)) or a unit loading on the
population-standardized scale (\(\Gamma = \sqrt2/\sqrt{1.36}\), an
under-correction by the factor \(.73\)); a screened correction removes
items flagged by a per-item outlier screen (\(|z| > 2.5\), centred at
the median item contrast) before pooling by the mean, and a
median-pooled correction uses the median item contrast. The script also
computes the population limit of each estimator by numerical
integration. \emph{Performance}: the naive increment is biased by
\(.152\) (population limit \(.151\)) from differential selection; the
correction with the true loading has bias \(.002\) (Monte Carlo standard
error \(.002\); population limit zero) and RMSE \(.042\); the unit
loading leaves \(.043\) (limit \(.041\)): (A2) is a substantive
assumption, and a miscalibrated loading under-corrects in proportion.
The corrected level contrast has bias \(.003\) (standard error \(.002\))
in the age-matched design and \(-.250\) in the generational design,
where the corrected contrast contains \(-g/\sqrt2 = -.247\); the
incumbent increment stays unbiased there (\(-.000\)). The
battery-homogeneity \(Q\) test rejects at \(.003\)--\(.007\) at nominal
\(.05\) without drift and at \(.34\) with three drifted items; the
screen flags at least one of the three drifted items in \(68\%\) of
replications, against \(2\%\) for the same three items without drift,
and flags an undrifted item at a per-item rate below \(.01\). Undetected
drift contaminates the mean-pooled correction (bias \(-.017\));
screening before pooling (\(-.007\)) and median pooling (\(-.006\))
reduce it.

\textbf{Simulation 3 (constructive designs and diagnostics;
Section~\ref{sec-constructive}--Section~\ref{sec-diagnostics}).}
\emph{Data-generating mechanism}: model (M) of
Proposition~\ref{prp-fail} with a continuing cohort of
\(n_{\text{old}} = 4{,}800\) entering at wave 1 and a fresh cohort of
\(n_{\text{new}} = 960\) entering at wave 5 (dose \(k = 4\)); latent
outcome \(Y^{*}_{it} = 0.5X_i + 0.5U_i + \varepsilon_{it}\) with
\(X_i \sim\) Bernoulli\((.5)\), \(U_i \sim N(0,1)\),
\(\varepsilon_{it} \sim N(0,1)\) independent across waves, untreated
standard deviation \(1.146\); a heterogeneous conditioning shift
\(\tau_i(s) = 0.30\{1 - e^{-(s-1)/4}\}(1 + \gamma U_i)\) whose
population mean at \(s = 5\) is \(0.190\) (\(0.166\) standard
deviations); and a dropout logit
\(-2.2 + a_X X_i + b_w U_i + a_E\varepsilon_{iw}\) for the decision
taken at wave \(w\). Six regimes: MCAR; MAR on \(X\) (\(a_X = .8\));
MNAR on the trait (\(a_X = .4\), \(b_w = .8\)); non-stationary MNAR on
the trait, in which the trait coefficient is \(1.4\) at the continuing
cohort's waves before the comparison wave and \(.8\) at every other wave
of both cohorts, so that trait selection is stronger in the continuing
cohort's selection period; state-dependent MNAR (\(a_X = .4\),
\(b_w = .4\), \(a_E = .8\)); and both failures together. Heterogeneity
is \(\gamma = .3\) (survivor targets \(h^{S}\) and \(h^{S+}\) differ
from the population shift and from each other, the case B5 excludes) or
\(\gamma = 0\). \emph{Estimands}: \(h^{S}\) (SM, EC, naive), \(h^{S+}\)
(SSM), the population shift (IPW). \emph{Methods}: the five estimators
with influence-function standard errors that include the estimated group
shares (Section~\ref{sec-variance}); the diagnostics
\(\mathcal{T}_{\mathrm{NS}}\) and \(\mathcal{T}_{\mathrm{SD}}\) at the
5\% level. \emph{Performance measures}: bias against each estimator's
own estimand and against the population shift, coverage of 95\%
intervals, rejection rates. There are \(R = 1{,}000\) replications per
regime and setting. The two settings of \(\gamma\) within a regime use
the same seed, so their data are identical apart from the shift (common
random numbers) and paired differences isolate the effect of
heterogeneity. Part D adds the variance experiments described below.
Code: \texttt{sims/sim3\_designs.R}.

\begin{table}[htbp]
\caption{Simulation 3, bias against each estimator's own estimand in
outcome units (95\% coverage in parentheses) and rejection rates of the
two diagnostics at the 5\% level; 1,000 replications per regime. Monte
Carlo standard errors of the biases are at most 0.002 and of the
rejection rates at most 0.016.}\label{tbl-sim3}
\begin{tabular}{@{}
  >{\raggedright\arraybackslash}p{(\linewidth - 14\tabcolsep) * \real{0.3043}}
  >{\raggedright\arraybackslash}p{(\linewidth - 14\tabcolsep) * \real{0.1087}}
  >{\raggedright\arraybackslash}p{(\linewidth - 14\tabcolsep) * \real{0.1087}}
  >{\raggedright\arraybackslash}p{(\linewidth - 14\tabcolsep) * \real{0.1087}}
  >{\raggedright\arraybackslash}p{(\linewidth - 14\tabcolsep) * \real{0.1087}}
  >{\raggedright\arraybackslash}p{(\linewidth - 14\tabcolsep) * \real{0.0870}}
  >{\raggedright\arraybackslash}p{(\linewidth - 14\tabcolsep) * \real{0.0870}}
  >{\raggedright\arraybackslash}p{(\linewidth - 14\tabcolsep) * \real{0.0870}}@{}}
\toprule\noalign{}
\begin{minipage}[b]{\linewidth}\raggedright
Regime (\(\gamma = 0.3\))
\end{minipage} & \begin{minipage}[b]{\linewidth}\raggedright
naive
\end{minipage} & \begin{minipage}[b]{\linewidth}\raggedright
SM
\end{minipage} & \begin{minipage}[b]{\linewidth}\raggedright
SSM
\end{minipage} & \begin{minipage}[b]{\linewidth}\raggedright
EC
\end{minipage} & \begin{minipage}[b]{\linewidth}\raggedright
IPW
\end{minipage} & \begin{minipage}[b]{\linewidth}\raggedright
\(\mathcal{T}_{\mathrm{NS}}\) rej.
\end{minipage} & \begin{minipage}[b]{\linewidth}\raggedright
\(\mathcal{T}_{\mathrm{SD}}\) rej.
\end{minipage} \\
\midrule\noalign{}
MCAR & +.001 (.94) & −.000 (.95) & −.000 (.95) & +.001 (.95) & .000 &
.046 & .053 \\
MAR on \(X\) & −.058 (.75) & −.001 (.96) & −.002 (.95) & −.001 (.96) &
−.001 & .041 & .056 \\
MNAR trait & −.198 (.00) & .000 (.94) & +.002 (.94) & −.001 (.94) &
−.198 & .039 & .056 \\
MNAR non-stationary & −.279 (.00) & \textbf{−.081 (.65)} & \textbf{−.068
(.76)} & −.002 (.95) & −.292 & \textbf{.623} & .078 \\
MNAR state-dependent & −.127 (.17) & \textbf{+.104 (.53)} & +.002 (.95)
& \textbf{+.103 (.36)} & −.115 & .040 & \textbf{.993} \\
MNAR, both failures & −.292 (.00) & +.012 (.93) & \textbf{−.062 (.80)} &
\textbf{+.078 (.59)} & −.306 & \textbf{.430} & \textbf{.945} \\
\bottomrule\noalign{}
\end{tabular}
\end{table}

The first five rows of Table~\ref{tbl-sim3} match
Proposition~\ref{prp-fail}(i)--(ii) row by row: survival matching fails
under non-stationarity and under state dependence, symmetric matching
only under non-stationarity, the entry-wave correction only under state
dependence, and under state dependence survival matching and the
entry-wave correction carry the same bias (\(+.104\) and \(+.103\)), so
that \(\mathcal{T}_{\mathrm{NS}}\) keeps its nominal size there. The
sixth row is a case the diagnostics cannot classify: with both failures
present, \(\mathcal{T}_{\mathrm{SD}}\) rejects almost always and
\(\mathcal{T}_{\mathrm{NS}}\) less than half the time, and the estimator
with the smallest bias is survival matching; its restriction fails, but
the two failures bias it in opposite directions. A rejection is evidence
of incompatibility among the maintained restrictions, and a
non-rejection of one statistic is not evidence that the corresponding
restriction holds.

\emph{Survivor versus population targets (Proposition~\ref{prp-ec}).}
This part illustrates, for a continuous outcome, the distinction that
Das, Toepoel, and van Soest (\citeproc{ref-dasetal2011}{2011, 41}, and
online Appendix 1, pp.~27--28) draw for binary items between the
population and the survivor target of the same functional. With
heterogeneous shifts (\(\gamma = 0.3\)), dropout increases with \(U\),
so the survivors carry smaller shifts than the population: under MNAR on
the trait the survivor target averages \(.170\) against a population
shift of \(.190\). The entry-wave correction is unbiased for the
survivor effect \(h^{S}\) (\(-.001\)) and biased for the population
shift (\(-.021\), or 11\% of \(\tau(5)\)); with homogeneous shifts
(\(\gamma = 0\)) the two targets coincide and the population bias is
\(-.001\). Because the two settings share their data, the paired
difference of the population biases, \(-.020\), is the composition term
\(\mathbb{E}[\tau_i \mid S] - \mathbb{E}[\tau_i]\) up to a Monte Carlo
error below \(.0001\). Read as an estimate of the population effect (the
target of their Assumption 3, a stationarity-of-attrition-bias
assumption on the conditioned outcome, which fails in this design
because survival selects on the heterogeneity of the shift), the
functional is off by that composition term; read as an estimate of
\(h^{S}\) under B4 (their Assumption 3\(^{\mathrm{Alt}}\) for a binary
item), it is unbiased here. Survival matching and symmetric matching
show the same gap of about \(.02\) between their survivor and population
readings. The paired comparison also shows the role of B5: the
\(\mathcal{T}_{\mathrm{SD}}\) rejection rate under MNAR on the trait is
\(.056\) with heterogeneous shifts and \(.044\) when the composition
term \(h^{S+} - h^{S}\) is switched off (paired difference \(.012\),
Monte Carlo standard error \(.005\)).

\emph{Variance (Section~\ref{sec-variance}).} In Part A the share terms
are small relative to the within-group variances, and the coverage of
the entry-wave correction is essentially the same with and without them
(for example \(.945\) and \(.944\) under MNAR on the trait). Part D
shows when they matter. In D1, a null experiment with well-separated
groups (entry and comparison responses both equal to \(U + N(0, .1^2)\)
with \(U \sim\) Bernoulli\((.5)\), survival \(S = U\) in both cohorts,
500 entrants per cohort, no conditioning, 5,000 replications), the
influence-function variance of the entry-wave correction averages
\(.00104\), against an empirical variance of \(.00106\) and a
theoretical value of \(.00104\), and gives coverage \(.947\) and a
\(\mathcal{T}_{\mathrm{NS}}\) rejection rate of \(.056\); the
within-group formula without share terms averages \(.00054\) and gives
\(.839\) and \(.706\). In D2, an exact counterexample (\(E = S\),
\(Y = 1\) among survivors, a constant fresh outcome), the entry-wave
estimate equals the survival share \(\hat p\), whose variance is
\(p(1-p)/n\); the influence-function estimate is its plug-in version
\(\hat p(1 - \hat p)/n\), while the within-group formula returns zero.
In D3, a person-level bootstrap within cohort that recomputes every arm
and every share (200 replications, 200 resamples each) agrees with the
influence-function standard errors to within one per cent in D1 and in
the MNAR-trait and MNAR-state regimes.

\textbf{Deterministic checks (no sampling).}
\texttt{check\_funnel\_examples.R} reproduces, on a fine grid with
stated tolerances, the following: (a) the identified set \(\{-1\}\) for
\(f_2 = N(0,1)\), \(Q = N(-1,1)\), \(p = .5\); (b) the attainment of
\(\Delta^{*} = \beta\sigma^2\) in the Gaussian--logistic example at
\(\alpha = -1\); (c) the point \(\{0\}\) under outcome-independent
attrition with Gaussian tails; (d) the Laplace and logistic half-widths
\(\log(1/p)\) and the Cauchy half-widths \((1-p)/\sqrt p\); (e) the
central-binding example after Corollary~\ref{cor-funnel} at
\(\sigma = .5\) and \(.8\); (f) the disconnected identified sets of the
modulated Cauchy and logistic densities of
Corollary~\ref{cor-independent}, together with the exponential-tail
outer bound and the failure of the closed form without the Lipschitz
condition; (g) the Lipschitz neighbourhood of
Corollary~\ref{cor-independent}(iii); (h) the discrete examples of
Corollary~\ref{cor-discrete}; (i) the general-map example of
Corollary~\ref{cor-repeated}; (j) the retention-only covariate of
Proposition~\ref{prp-covariates}; (k) a case in which the per-pair outer
bound of Theorem~\ref{thm-tworef}(a) is strictly larger than the
common-loading set (block 11 of the script): a Laplace episode with
\(C_1 = 0\), \(N_1 = 1\) and a Gaussian episode with \(C_2 = 10\),
\(N_2 = 1\), both at \(p = .5\), with \(\Gamma \in [0, 1]\), where the
Gaussian episode forces \(\Gamma = 0\), so that the common-loading set
for the first episode is \(\{0\}\) and not \([-\log 2, 0]\); and (l) the
example following Theorem~\ref{thm-tworef}, in which the loading set
computed from the margins, \([0, .817]\), strictly contains the set
under the joint version of (A1), \([0, .263]\) (block 12).
\texttt{check\_multiwave\_completion.R} enumerates, without sampling, a
four-wave cohort on the latent lattice \(\{-5, \dots, 5\}\) with a
binary entry-only negative control, all eight response patterns
(interrupted ones included), true conditioning by integer location
shifts \((0, 1, 0, 1)\), and a selection kernel that depends on the
whole latent path including the never-reported coordinates and that
excludes the lowest latent value from responding at wave 4. For every
candidate map \(c'_4 = c_4 + \epsilon\),
\(\epsilon \in \{-3, \dots, 3\}\), with one arbitrary choice of
candidate maps at the intermediate waves (shifts of \(+1\) and \(-1\) at
waves 2 and 3), the completion of Lemma~\ref{lem-seqcompletion} is
constructed explicitly. Exactly one non-trivial candidate,
\(\epsilon = +1\), satisfies the wave-4 domination inequality. For it
and for \(\epsilon = 0\) every pattern sub-measure, the entry marginal
including the negative control, and the refreshment marginal are
reproduced within the script's tolerance of \(10^{-10}\), the kernel is
a valid Markov kernel, the survivors' latent mean at wave 4 moves by
\(-\epsilon\), and the negative-control contrast does not move. For
every other \(\epsilon\) the leftover \(L_4\) has a negative entry.
Restricting the same example to monotone patterns verifies the
sequential form \(\pi'_j = dN_{j+1}/dN_j\). The script also enumerates
all 35 weakly increasing maps on four categories and confirms in every
instance that the mass criterion Equation~\ref{eq-mass} coincides with a
linear-programming feasibility test, with the identity as the sole
feasible map when all categories are reported and non-identity maps
feasible only when a category goes unreported. It also confirms the
Gaussian widths of Corollary~\ref{cor-meanbounds} (\(.390\), \(.700\),
\(1.288\), \(1.596\) at \(p = .9, .8, .6, .5\)) by closed form,
numerical integration, and a linear programme on a grid.
\texttt{check\_design\_rank.py} computes, in exact rational arithmetic,
the rank and the projected kernel of every design quoted in
Section~\ref{sec-design} and Table~\ref{tbl-strides}, the example of
Theorem~\ref{thm-rank}(iii) and its null vector, the four-cohort example
of Proposition~\ref{prp-stride}(iii), and the effect of level and
plateau anchors on the Japanese panel's support. It also checks, on all
2,214 staggered trapezoids with three to five cohorts, first entry at
wave 1, later entries in \(\{2, \dots, 10\}\) and up to eight periods of
follow-up, the component count, the sufficient follow-up condition, the
exact three-cohort formula and the upper bound of the companion
identification paper's Lemma 1. On 65 of these designs every pair of
observed tenures is linked through tenures observed in common periods
although the identified set is larger than \(d\).

\section{Appendix: Variance of the Design
Estimators}\label{sec-variance}

Each estimator of Section~\ref{sec-constructive} and each diagnostic of
Section~\ref{sec-diagnostics} is a difference of subset means computed
within two independent cohorts. For a cohort of \(n\) persons and a
subset indicator \(B_i\), the mean
\(\bar Y_B = \sum_i B_i Y_i / \sum_i B_i\) has influence function
\(\psi_i = B_i(Y_i - \mu_B)/\Pr(B)\); a linear combination of subset
means within the cohort has the corresponding combination of influence
functions, and its variance is estimated by \(\sum_i \hat\psi_i^2/n^2\).
For the entry-wave correction, with \(S\) survival through \(t\), \(E\)
the entry response and \(Y\) the response at \(t\), the
continuing-cohort component \(\hat A = \bar Y_S - (\bar E_S - \bar E)\)
satisfies \[
n\,\mathrm{Var}(\hat A) \;\approx\; \frac{\mathrm{Var}\{Y - (1-p)E \mid S\}}{p} \;+\; (1-p)\,\mathrm{Var}(E \mid \text{NS}) \;+\; p(1-p)\,(\mu_{E,S} - \mu_{E,\text{NS}})^2 ,
\] where the last term comes from estimating the survivors' share \(p\);
a variance computed conditionally on the group sizes omits it. The
diagnostics
\(\mathcal{T}_{\mathrm{NS}} = (\bar E_S - \bar E) - (\bar Y_{m} - \bar Y)\)
and
\(\mathcal{T}_{\mathrm{SD}} = (\bar Y_{S+} - \bar Y_S) - (\bar Y_{m+1} - \bar Y_{m})\),
where \(\bar Y_{S+}\) averages the continuing survivors who also respond
at \(t+1\) and \(\bar Y_m\), \(\bar Y_{m+1}\) the fresh respondents who
go on to respond for \(k\) and \(k+1\) further waves, contain the
analogous share terms in both cohorts. A person-level bootstrap within
cohort that recomputes every arm and every share targets the same
variance. Simulation 3 and version 0.1.4 of the R package
\texttt{panelcond} use the influence-function formulas; the empirical
results of Section~\ref{sec-implementation} use the bootstrap.

\section{Appendix: The Negative-Control Battery at the 2019
Episode}\label{sec-ncbattery}

\begin{longtable}[]{@{}lll@{}}
\caption{Negative-control contrasts between the 2007 cohort at tenure 13
and the 2011 cohort at tenure 9, on their common birth-year support
(1966--1986), standardized and pooled over three-year birth band by sex
strata; aggregates computed from the licensed JLPS microdata by the
replication kit. \textsuperscript{a} Rejected at a Benjamini--Hochberg
\(q < .10\). \textsuperscript{b} Standard error under independence
across items; the person-bootstrap standard error is \(.0186\) (design
effect \(3.84\)).}\label{tbl-ncbattery}\tabularnewline
\toprule\noalign{}
Item (reported once at entry) & \(d\) & SE \\
\midrule\noalign{}
\endfirsthead
\toprule\noalign{}
Item (reported once at entry) & \(d\) & SE \\
\midrule\noalign{}
\endhead
\bottomrule\noalign{}
\endlastfoot
Possessions in the home at age 15: air conditioner & −.033 & .045 \\
television & −.032 & .047 \\
personal computer or word processor & .015 & .045 \\
video recorder & .040 & .045 \\
piano & −.074 & .045 \\
radio & −.092 & .045 \\
refrigerator & −.103 & .046 \\
second home & .009 & .049 \\
study desk & −.033 & .045 \\
reception room furniture & −.051 & .045 \\
owner-occupied home & −.014 & .045 \\
literature collection or illustrated encyclopedia & −.012 & .045 \\
farmland or forest (excluding a kitchen garden) & −.015 & .045 \\
encyclopedia & .025 & .045 \\
works of art or antiques & .023 & .045 \\
car & .022 & .045 \\
microwave oven & −.020 & .045 \\
telephone (mobile telephones included)\textsuperscript{a} & −.176 &
.045 \\
bath & −.069 & .046 \\
Number of books in the home at age 15 & .078 & .049 \\
Family atmosphere at age 15 & −.106 & .045 \\
Household living standard at age 15 & −.096 & .045 \\
School grades in junior high school & −.069 & .046 \\
Inverse-variance pooled value & −.035 & .0095\textsuperscript{b} \\
\end{longtable}

\section*{Data and code availability}\label{data-and-code-availability}
\addcontentsline{toc}{section}{Data and code availability}

The simulations and deterministic checks generate their own data; no
restricted data are needed to reproduce the appendix. Their code and
outputs, the figure script, and the manuscript source are in the
replication archive at
\url{https://github.com/sokubo/paper-refreshment-designs-replication};
the version corresponding to this manuscript is tag \texttt{paper-v1.5}
(computational commit \texttt{f4f7c2b}). The archive's README maps each
table, figure and quoted simulation value to the script and output file
that produce it, and \texttt{sims/check\_manuscript\_values.R} asserts
those values against named output cells. The empirical illustration of
Section~\ref{sec-implementation} uses the integrated release (waves
1--19) of the Japanese Life Course Panel Surveys (JLPS) distributed to
participants of the survey project; the microdata cannot be
redistributed. Other researchers can apply for JLPS releases to the
SSJDA Data Archive of the Center for Social Research and Data Archives,
Institute of Social Science, The University of Tokyo (at the time of
writing, the current releases are the study numbers PY160 and PM160).
Those releases differ from the integrated file, so an analysis of them
with the archived code is a re-analysis with the same design and
definitions, not a replication. The archive contains the code that reads
the licensed files and produces every aggregate quoted in
Section~\ref{sec-implementation} (\texttt{kit/R/}), a document stating
the access route, the variables and the recoding
(\texttt{kit/KIT\_README.md}), and a synthetic data set with the outputs
the code produces on it, so that the workflow for the 2011 episode can
be run end to end without the licensed data. The scripts for the 2019
negative-control battery need the confidential map from items to
variables and are not exercised by the synthetic example. Every script
reads one explicitly named input file and records its name and SHA-256;
\texttt{kit/KIT\_README.md} gives the file name and fingerprint of the
input behind Section~\ref{sec-implementation}, and
\texttt{kit/check\_manuscript\_values\_T2.R} reproduces every quoted
aggregate from the aggregate outputs after checking that all of them
record that input. The aggregate outputs computed from the licensed data
are reported in the paper (Section~\ref{sec-implementation},
Table~\ref{tbl-ncbattery}) and are not otherwise redistributed. The R
package \texttt{panelcond}, version 0.1.6
(\url{https://github.com/sokubo/panelcond}, tag \texttt{v0.1.6}),
implements the estimators, the diagnostics, the variance formulas of
Section~\ref{sec-variance} and the design-rank calculation; the design
run behind Section~\ref{sec-implementation} used that version, whose
covariate-standardized entry-wave correction integrates both regression
terms over the survivors' covariate distribution (versions up to 0.1.5
integrated them over two populations, Section~\ref{sec-implementation}).
The analytic-variance correction of version 0.1.4 changed no bootstrap
result, whereas the standardized functional of 0.1.6 and its bootstrap
summaries were recomputed for this version. The standard errors and
diagnostic uncertainties for the Section~\ref{sec-implementation} design
comparisons are computed using the joint person bootstrap rather than
the package's analytic-variance columns. The negative-control analysis
separately reports item-level uncertainties, the independence-based
pooled benchmark, and the person-bootstrap pooled uncertainty.

\section*{References}\label{references}
\addcontentsline{toc}{section}{References}

\phantomsection\label{refs}

\end{document}